\documentclass[10pt]{article}
\usepackage{amssymb}
\usepackage{amsfonts}
\usepackage{bbold}

\usepackage{geometry}

\usepackage{amsmath}
\usepackage{amsthm}
\usepackage{mathrsfs} 
\usepackage{graphicx}
\usepackage{mathtools}
\usepackage{enumitem}
\usepackage{caption}
\usepackage[nameinlink]{cleveref}
\numberwithin{equation}{section}
\usepackage{xcolor}

\usepackage{yfonts}
\usepackage{amssymb}
\usepackage{amsmath}
\theoremstyle{theorem}
\newtheorem{theorem}{Theorem}[section]
\newtheorem{lemma}[theorem]{Lemma}
\newtheorem{proposition}[theorem]{Proposition}

\newtheorem{definition}[theorem]{Definition}
\newtheorem{remark}[theorem]{Remark}

\newtheorem{example}[theorem]{Example}
\newtheorem{corollary}[theorem]{Corollary}

\def\sp{\hskip -5pt}
\def\spa{\hskip -3pt}

\def\bearray{\begin{eqnarray}}
\def\earray{\end{eqnarray}}
\def\beq{\begin{equation}}
\def\eeq{\end{equation}}

\def\b0{{\bf 0}}

\def\mpasto{\mapsto}

\def\cA{{\cal A}}

\def\cD{{\cal D}}

\newsymbol\bt 1202             

\def\bC{{\mathbb C}}           

\def\bN{{\mathbb N}}

\def\bR{{\mathbb R}}

\def\bZ{{\mathbb Z}} 

\newsymbol\rest 1316         
\def\gA{{\mathfrak A}}       
\def\gB{{\mathfrak B}}

\begin{document} 

\par
\bigskip
\large
\noindent
{\bf The classical limit of  Schr\"{o}dinger operators in the framework of  Berezin  quantization  and spontaneous symmetry breaking as emergent phenomenon}
\bigskip
\par
\rm
\normalsize

\noindent {\bf Valter Moretti$^{a}$$^1$ and  Christiaan J.F.  van de Ven$^{a}$$^{b}$$^2$}\\
\par

\noindent 
 $^a$ Department of  Mathematics, University of Trento, and INFN-TIFPA,
 via Sommarive 14, I-38123  Povo (Trento), Italy\\
$^1$ Email: valter.moretti@unitn.it\\
$^b$ Marie Sk\l odowska-Curie Fellow of Istituto Nazionale di Alta Matematica.  \\ $^2$ Email:
 christiaan.vandeven@unitn.it\\

\par

\rm\small

\hfill 03 Oct 2021

\rm\normalsize

\par
\bigskip

\noindent
\small
{\bf Abstract.}
The algebraic properties of a strict deformation quantization are analysed on the classical phase space $\bR^{2n}$.  The corresponding quantization maps enable us to take the  limit for $\hbar \to 0$ of a suitable  sequence of  algebraic
vector states induced by $\hbar$-dependent eigenvectors of several quantum models, in which the sequence converges 
to a probability measure on $\bR^{2n}$, defining a classical algebraic state. The observables are here represented in terms of  a Berezin quantization map which associates classical observables (functions on the phase space) to quantum observables (elements of $C^*$ algebras) parametrized by $\hbar$.  The existence of this classical limit is in particular proved for ground states of a wide class 
of Schr\"{o}dinger operators, where the classical limiting state is obtained  in terms of a Haar integral. The support of the classical state (a probability measure on the phase space) is included in certain orbits in $\bR^{2n}$ depending on the symmetry of the potential. In addition, since this $C^*$-algebraic approach allows for both quantum and classical theories, it is highly suitable to study the theoretical concept of spontaneous symmetry breaking (SSB) as an emergent phenomenon when passing from the quantum realm to the classical world by switching off $\hbar$. To this end, a detailed mathematical description is outlined and it is shown how this algebraic approach sheds new light on spontaneous symmetry breaking in several physical models.
\normalsize
\tableofcontents

\section{Introduction}
\subsection{Main results and structure of the work}
The aim of this paper is to provide a mathematical rigorous theory regarding  the existence of the so-called {\em classical limit} of quantum systems and a certain tupe of Schr\"{o}dinger operators with compact resolvent. Simultaneously, this machinery is also exploited  to study  the occurrence of {\em spontaneous symmetry breaking} as an {\em emergent} phenomenon, arising when passing from the quantum to the classical regime.  Our approach is highly based on the theory of $C^*$-algebras rather than the (perhaps more) general approach on microlocal analysis and pseudodifferential calculus. Despite the fact that the latter approaches are very useful for controlling estimates in terms of a semiclassical parameter \cite{GriSjo,MZ}, they are less suitable for understanding emergence, like SSB in the classical limit. Instead, we will see that the concept of a continuous bundle of $C^*$-algebras is the natural way to formalize SSB as emergent phenomenon, since it encodes {\em both} quantum and classical theory by means of relatively simple algebraic relations.

Generally speaking the aforementioned issues are investigated within the framework of {\em (deformation) quantization procedures}. The basic idea is to consider a classical theory, whose observables are described by sufficiently regular functions over typically a space of phases $X$,
as the limit for $\hbar \to 0^+$ of a sequence of theories whose observables are represented by selfadjoint operators on a corresponding sequence of Hilbert spaces. Obviously, only a selection of sequences of observables, parametrized by $\hbar \geq 0$, makes physical  sense. These are sequences with a suitable continuity property reformulated in terms of $*$-algebras  (more precisely, $C^*$ algebras as explained below). Technical problems often arising in the setting of Hilbert spaces are  typically avoided in this way. In this {\em algebraic approach}, the quantum observables are given by selfadjoint elements of abstract $*$-algebras $\gA_\hbar$ of formal operators $a\in \gA_\hbar$ and the algebraic states are  linear complex-valued  functionals on such algebra
$\omega_\hbar : \gA_\hbar \to \bC$
 with the physical meaning of $\omega_\hbar(a)$ as the expectation values of the observable $a=a^*$ in the state $\omega_\hbar$. Another benefit 
arising from the use of the algebraic approach is that, differently from the Hilbert space formulation, the algebraic approach  is suitable even for classical theories. This is because  the set of functions $f$ on the space of phases  $X$ representing observables (also extending the functions to complex valuated maps) has a natural structure of {\em commutative}  $*$-algebra $\gA_0$. The algebraic states  are there  nothing but probability measures over the space of phases: $\omega_0(f) = \int_X f d\mu_\omega$.
To avoid technical problems with topologies, the family of algebras $\{\gA_\hbar\}_{\hbar\geq 0}$ is chosen to be made of more tamed $C^*$-algebras rather than $*$-algebras.  It is important to stress that this  more abstract viewpoint actually encompasses the Hilbert space formulation. It is  because the   celebrated GNS reconstruction theorem (see e.g. \cite{Lan17,Mor}) permits to recast the abstract algebraic perspective to a standard  Hilbert space framework.
The sequence of $C^*$-algebras $\{\gA_\hbar\}_{\hbar\geq 0}$, where $\gA_0$ is the algebra of classical  observables, is formally encapsulated in the structure of  {\em $C^*$-bundle} we shall review in Section
\ref{Quanproc}.

A  machinery of utmost relevance in this framework is the notion of {\em quantization map}. Its  design can be traced back to Dirac's foundational ideas on the quantum theory  and, in the modern view, it  consists of a map $Q_\hbar : \gA_0 \ni f \mapsto Q_\hbar(f) \in \gA_\hbar$ which associates classical  observables to quantum ones. Obviously the quantization map is requested to satisfy a number of conditions of various nature. For instance, one of them regards the interplay of the quantum commutator and the Poisson bracket referred to the symplectic structure of the space of phases $X$. After Dirac,  one expects that in the $\frac{1}{\hbar}[Q_\hbar(f),Q_\hbar(g)]$ tends to 
$Q_\hbar(\{f,g\})$ for $\hbar \to 0^+$. This condition stated within a suitable topological formulation is nowadays known as the  
{\em Dirac-Groenewold-Rieffel condition}. As is known, the set of all naturally expected  requirements on $Q_\hbar$ is contradictory as proved in the various versions of the {\em Groenewold-van Hove theorem} \cite{Lan17}. These no-go theorems gave rise to the birth of a number of different types of quantization maps whose different nature depends on the specific  choice of a  subset of mutually compatible requirements. The most popular quantization map  is the one attributed to Weyl $Q^W_\hbar$, whereas one of the class of the  most effective quantization maps, adopted in this work, is the so-called {\em Berezin quantization map} $Q_\hbar^B$
which is reviewed in Section \ref{WeylBerezin} \cite{Ber}. It is worth mentioning  that in the special case when dealing with a compact K\"{a}hler manifold, the theory of {\em Toeplitz quantization} has also proven its great importance  \cite{BMS94,SCHL}.

The version of the deformation quantization procedure we adopt in this work is called {\em strict deformation quantization} and has been introduced by Rieffel in \cite{Rie89}, who showed how deformation quantization makes sense in a $C^*$-algebraic context providing a precise link between Poisson (and symplectic geometry) and  non-commutative geometry \cite{Con}.
 For the purpose of this paper we focus on the physical systems proper of  quantum mechanics, although quantization theory can be profitably exploited to study  quantum field theory and quantum gravity as well \cite{Lan17}.
 
Using the aforementioned concepts we attempt to bring forward two important topics in this area of  mathematical physics:
\begin{itemize}
\item[(1)] Existence of the classical limit of a sequence of $\hbar$-indexed eigenvectors $\{\psi_\hbar\}_{\hbar>0}$ corresponding to a quantum Hamiltonian.
\item[(2)] Occurrence of Spontaneous Symmetry Breaking (SSB) as a emergent phenomenon arising in the classical limit $\hbar \to 0$.
\end{itemize}

Regarding (1), the issue is to study whether or not  the sequence of quantum algebraic states $\langle \psi_\hbar,  Q_\hbar(f)   \psi_\hbar \rangle$ tends to some
classical algebraic state $\omega_0( f)$ for any classical observable $f\in \gA_0= C_0(\bR^{2n})$, when $\hbar \to 0^+$. This issue is presented from a technical perspective
 in  Section \ref{claslim}. In Section \ref{SemiclassicalpropertiesofBerezinquantizationmaps} we study the semiclassical properties of Berezin quantization maps
$Q_\hbar^B$ which is the main technical tool of the investigation. Our main theorems concern the localization of eigenvectors (viz. Prop. \ref{localization2}), and the classical limit of a sequence of eigenvectors associated to Berezin quantizations (viz. Thm. \ref{MT2} and Thm. \ref{thm:claslimNW2}): we prove that the classical limit exists for every $f\in C_0(\bR^{2n})$ and
for a sequence of eigenvectors of  operators $Q_\hbar^B(e_h)$, where the observable $e_h$ is   constructed out of 
 the classical hamiltonian $h(q,p) = p^2+V(q)$ in $\bR^{2n}$ equipped with the standard symplectic and Poisson structure.        In Section \ref{MainB} a certain class of Schr\"{o}dinger operators $H_\hbar = -\hbar^2 \Delta + V$ is studied and the relation with quantization maps is given in terms of approximation theorems (Prop. \ref{IMPORTANT}, Prop. \ref{IMPORTANT2}, and Thm. \ref{propDIM}).
As  a matter of fact we discuss the interplay of these Sch\"{o}dinger operators and the corresponding obeservables $Q_\hbar^B(e_h)$.
 In Section \ref{MainS} we finally readapt some of the theorems of Section \ref{SemiclassicalpropertiesofBerezinquantizationmaps} and prove the localization of eigenvectors  (Prop. \ref{localizationSCH}), followed by the existence of the classical limit of a sequence of eigenvectors corresponding to Schr\"{o}dinger operators (Thm. \ref{MT12}, Thm. \ref{thm:claslimNEWSchr}) $H_\hbar$ instead of the associated Berezin observable $Q_\hbar^B(e_h)$.

The found results about the classical limit distinguish between the case where a symmetry group $G$ acts on the physical system at classical and quantum level,  or there is not such a group.  This distinction plays a central role in developing the issue (2). In case a symmetry  group $G$ exists and the considered states are {\em ground states}
of a given Hamiltonian (quantum or classical), spontaneous  symmetry breaking occurs if there are no {\em extremal} (roughly speaking {\em pure}) ground states that are symmetric under the action of $G$. It is interesting to study if the phenomenon of spontaneous  symmetry breaking  arises as an emergent phenomenon, i.e., when taking the limit $\hbar \to 0^+$.
In Section \ref{SSBEMERGENCE} we prove that it is the case for a large class of Hamiltonians (that includes, in particular,  the {\em double well system} and the {\em Mexican hat system}) and for a topological compact or finite group $G$. As a matter of fact,
we put our findings in a broader perspective  than the one of the previous sections and relate the found results to the concept of spontaneous symmetry breaking in the form of Prop. \ref{NOSSB}, Prop. \ref{WEAKSSB}, and Remark \ref{remfin}.

We point out to the reader that for quantum spin systems these topics are partially related by the famous work by Lieb \cite{Lieb}, and for Schr\"{o}dinger operators instead some results are obtained by Simon \cite{SJHE,Sim84,Sim85}. However, different approaches were used and no limit of a sequence of eigenvectors was taken. More recent studies \cite{LMV,MV,Ven2020,MV2} have shown a well-established relation between the classical limit of a sequence of eigenvectors associated to mean-field quantum spin system Hamiltonians and spontaneous symmetry breaking. In these relatively simple systems, where the Hamiltonian is always a bounded operator, the symmetry group is typically finite. In the context of Schr\"{o}dinger operators where one, e.g., considers $SO(n)$-symmetries, this seems to be a challenging,  and an open problem up to our knowledge.

The works ends with an appendix.  A number of proofs of various technical propositions and intermediate lemmata stated in the main text are included therein.

\subsection{Notations and conventions}
If $X$ is a Hausdorff locally compact space, $C_0(X)$ indicates the space of functions $f: X \to \bC$ such that, for every $\epsilon>0$ there is a compact $K_\epsilon$ such that  $|f(x)|< \epsilon$
if $x\not \in K_\epsilon$.  $C_c(X)$ and $C_c^\infty(X)$ respectively  denote  the subspace of continuous compactly supported functions $f: X \to \bC$
and the analog for smooth functions when $X$ is a smooth manifold. 

The {\em Lebesgue measure} on $\bR^n$ will be always denoted by $dx$ (or $dp$, $da$, etc). 
In the phase space $\bR^{2n}$ we often use the normalized Liouville measure $\mu_\hbar :=\frac{dpdq}{(2\pi\hbar)^n}$.
The definition of {\em Borel measure} and {\em regular} Borel measure, the statements of corresponding {\em Riesz' theorems}  are  adopted from  \cite{Rudin}.

If $Z$ is a complex  Banach space, then $\gB(Z)$ denotes the unital Banach algebra of the bounded operators $A: Z \to Z$ with respect to the standard operator norm.
If ${\cal H}$ is a complex Hilbert space $\gB_1(\cal H)$, $\gB_2(\cal H)$, $\gB_\infty(\cal H)$ respectively denote the two-sided $*$-ideals of trace class, Hilbert-Schmidt, and compact operators in  $\gB({\cal H})$. $^*$ is used to denote   the adjoint $A^*$ of an operator $A: D(A) \to {\cal H}$ with $D(A)$ dense in ${\cal H}$, and also the abstract adjoint $a^*$ of an element $a$ in a $*$-algebra $\gA$.

If $\gA$ is a $C^*$-algebra without unit, an {\em algebraic state} is a positive $(\omega(a^*a) \geq 0)$, continuous, and norm-$1$ linear functional $\omega : \gA \to \bC$. (As is known, these requirements are equivalent to positivity and normalization $\omega(1\spa1)=1$ if $\gA$ admits unit $1\spa1$ for a linear functional $\omega : \gA \to \bC$).

The {\em Fourier transform} \cite{RS1} and its inverse are respectively  defined as follows,  for $f,g$ in the {\em Schwartz space}  ${\cal S}(\bR^n)$,
$$ {\cal F}_\hbar(f)(p) := \int_{\bR^n} \sp\spa  e^{-ip\cdot x/\hbar} f(x) \frac{dx}{(2\pi \hbar)^{n/2}}\:; \qquad 
 {\cal F}_\hbar^{-1}(g)(x) :=\int_{\bR^n} \sp\spa e^{ip\cdot x/\hbar} \hat{h}(p) \frac{dp}{(2\pi \hbar)^{n/2}}\:,$$
and we also use the notation $\hat{f}_\hbar(p) := {\cal F}_\hbar(f)(p)$, $\check{g}_\hbar(x) := {\cal F}_\hbar^{-1}(g)(x)$.
We shall also take advantage of  the natural continuous linear extensions of ${\cal F}_\hbar$ and its inverse  (1) to the  space ${\cal S}'(\bR^n)$ of {\em Schwartz distributions} and (2)
to $L^2(\bR^n, dx)$. The latter is known as the {\em Fourier-Plancherel transformation} which will be denoted by $F_\hbar$ and is a unitary operator on $L^2(\bR^n, dx)$.
With our conventions, the {\em convolution theorem} in the Schwartz space reads
${\cal F}_\hbar(f*g) = (2\pi \hbar)^{n/2}{\cal F}(f)_\hbar \:,$
where the {\em convolution}  of $f$ and $g$ is defined as 
$(f*g)(x) := \int_{\bR^n} f(x-y)g(y) dy$ as usual.
The analogous  statement is valid for the said pair of extensions of ${\cal F}_\hbar$. When $\hbar=1$, we simply omit the index writing for instance ${\cal F}$, $\hat{h}$, and $F$.

\section{Quantization maps, classical limit, and all that}\label{Quantizationmapsandrelatedideas}
A quick  summary on the machinery used to formalize the idea of a quantization procedure is the main goal of this section. We shall focus in particular on the {\em Berezin quantization} map  which will be used throughout the rest of the work. An intermediate technical step concerns the more popular {\em Weyl quantization}.

\subsection{Quantization procedures}\label{Quanproc}
We start from the following technical definition of $C^*$-{\em bundle} according to \cite[Definition C 121]{Lan17}. This structure is a quite general  arena where  a  quantization procedure based on $C^*$ algebras 
is performed.

\begin{definition}\label{def:continuous algebra bundle}
A  {\bf $C^*$-bundle}\footnote{Called {\em continuous field of $C^*$-algebras} in \cite{Lan98}.} is a triple $\cA:=(I,\gA,\pi_\hbar:\gA\to \gA_\hbar)$,
 where $I$ is a locally compact Hausdorff space,
 $\gA$ is a complex $C^*$-algebra, and $\pi_\hbar:\gA\to \gA_\hbar$
 is a
surjective homorphism  of complex $C^*$-algebras for each $\hbar \in I$ such that
\begin{itemize}
\item[(i)] $\|a \|= \sup_{\hbar\in I} \|\pi_\hbar(a)\|_\hbar$,  where $\|\cdot\|$ (resp. $\|\cdot\|_\hbar$) denoting the $C^*$-norm of $\gA$ (resp. $\gA_\hbar$);
\item[(ii)] there exist an action $C_0(I) \times \gA \to \gA$ satisfying $f(\hbar)\pi_\hbar(a) = \pi_\hbar(f a)$ for any $\hbar\in I$ and $f\in C_0(I)$.
\end{itemize} 
A {\bf section} of the bundle is an element $\{a_{\hbar}\}_{\hbar\in I}$ of $\Pi_{\hbar\in I}\gA_\hbar$ for which there exists an $a\in \gA$
 such that $a_\hbar=\pi_\hbar(a)$ for each $\hbar\in I$. A {\bf $C^*$-bundle} $\cA$ is said to be {\bf continuous}, and its sections are called {\bf continuous sections},  if  it satisfies
\begin{itemize}
\item[(iii)] for $a \in \gA$, the norm function $I \ni\hbar \mapsto \|\pi_\hbar(a)\|_\hbar$ is in $C_0(I)$. 
\end{itemize}
It is not requested that the $C^*$-algebras $\gA$, $\gA_\hbar$ are unital. 
If they are iwith units $1\spa1_\gA$, $1\spa1_{\gA_\hbar}$ respectively, then   $\pi_\hbar$ is supposed to be 
unit preserving: $\pi_\hbar(1\spa1_\gA) = 1\spa1_{\gA_\hbar}$.
\end{definition}

\begin{remark}{\em Since the $\pi_\hbar$ are homomorphisms of $C^*$-algebras,  the $*$-algebra operations in $\gA$ are induced by the corresponding pointwise operations of the sections $I\ni \hbar \mapsto \pi_\hbar(a)$. Condition (ii) reinforces the linearity preservation condition permitting coefficients continuously depending on $\hbar$.}
\end{remark}

Let us  introduce the notion of 
{\em deformation quantization} of classical structures, a Poisson manifold in particular,   according to \cite[Definition 71]{Lan17}. The above $C^*$-{\em bundle} is therefore  specialized to the case where $\gA_0$ is a commutative $C^*$ algebra representing the classical structure achieved in the {\em classical limit} $\hbar \to 0^+$
from corresponding quantum structures defined in the quantum fibers $\gA_\hbar$ with $\hbar >0$.  The {\em quantization maps} act along the opposite direction, associating to a classical observable $f \in \gA_0$ (or a substructure of it) a quantum observable $Q_\hbar(f) \in \gA_\hbar$.

\begin{definition}\label{def:deformationq}
A {\bf deformation quantization}\footnote{Named  {\em continuous quantization}
of a Poisson manifold in  \cite[Definition II 1.2.5]{Lan98}.} of a Poisson manifold $(X, \{\cdot,\cdot\})$ consists of:
 \begin{itemize}
\item[(1)]  A {continuous  $C^*$-bundle} $(I, \gA, \pi_\hbar)$,
 where $I$ is a subset of $\mathbb{R}$ containing $0$ as accumulation point and $\gA_0=C_0(X)$ 
equipped with norms $||\cdot||_{\hbar}$;
\item[(2)] a dense $*$-subalgebra $\tilde{\gA}_0$ of  $C_0(X)$ closed under the action  Poisson brackets 
(so that $(\tilde{\gA}_0, \{\cdot, \cdot\})$ is a complex Poisson algebra);
\item[(3)]  a  collection of   {\bf quantization maps} $\{Q_\hbar\}_{\hbar\in I}$, namely linear maps $Q_{\hbar}:\tilde{\gA}_0  \to \gA_{\hbar}$ 
(possibly defined on $\gA_0$ itself and next restricted to $\tilde{\gA}_0$)
such that: 
\begin{enumerate}
\item[(i)] $Q_0$ is the inclusion map $\tilde{\gA}_0 \hookrightarrow \gA_0$ (and $Q_{\hbar}(1\spa1_{\gA_0})=1\spa1_{\gA_{\hbar}}$
if  $\gA_0,$ and $\gA_\hbar$ are unital for all $\hbar \in I$);
\item[(ii)] $Q_{\hbar}(\overline{f}) = Q_{\hbar}(f)^*$, where $\overline{f}(x):=\overline{f(x)}$;
\item[(iii)] for each $f\in\tilde{\gA}_0$, the assignment
$
0\mapsto f, \quad 
\hbar\mapsto Q_{\hbar}(f)$ when $\hbar \in I\setminus \{0\},
$
defines  a continuous section of $(I, \gA, \pi_\hbar)$,
\item[(iv)]  each pair $f,g\in \tilde{\gA}_0$ satisfies the {\bf Dirac-Groenewold-Rieffel condition}:
\begin{align*}
\lim_{\hbar\to 0}\left|\left|\frac{i}{\hbar}[Q_{\hbar}(f),Q_{\hbar}(g)]-Q_{\hbar}(\{f,g\})\right|\right|_{\hbar}=0.
\end{align*}
\end{enumerate}
\end{itemize}
 If $Q_\hbar(\tilde{\gA}_0)$ is dense in $\gA_\hbar$ for every $\hbar \in I$, then the deformation quantization is called {\bf strict}.\\
(If $Q_\hbar$ is defined on the whole $C_0(X)$, all conditions except (iv) are assumed to be valid on $C_0(X)$.)
\end{definition}
Elements of $I$ are interpreted as possible values of Planck’s constant $\hbar$ and $\gA_{\hbar}$
 is the quantum algebra of observables of the theory at the given value of $\hbar\neq 0$. For real-valued $f$, the operator
$Q_\hbar(f)$ is the quantum observable associated to the classical observable $f$. This is possible because of condition $(ii)$ in Definition \ref{def:deformationq}.

 It immediately follows from the definition of a continuous bundle of $C^*$-algebras that for any $f\in \tilde{\gA}_0$ the  continuity properties holds called the {\bf Rieffel condition}
\begin{align}\label{Rifc}
\lim_{\hbar\to 0}\|Q_{\hbar}(f)\|_\hbar=\|f\|_{\infty}\:.
\end{align}
 Also the so-called {\bf von Neumann condition}
\begin{align}\label{vNc}
\lim_{\hbar\to 0}\|Q_{\hbar}(f)Q_{\hbar}(g)-Q_{\hbar}(fg)\|_\hbar=0
\end{align}
is valid. Indeed, the section $I \ni \hbar \mapsto Q_{\hbar}(f)Q_{\hbar}(g)-Q_{\hbar}(fg)$ is a continuous section because, it is constructed with the pointwise operations of the $C^*$-algebra $\gA$ and 
$(I, \gA, \pi_\hbar)$ is a {continuous  $C^*$-bundle}, finally $Q_{0}(f)Q_{0}(g)-Q_{0}(fg) = fg-fg=0$, hence (iii) in Definition \ref{def:continuous algebra bundle} implies (\ref{vNc}).

\begin{remark}{\em  Suppose we have a strict deformation quantization according to the previous definition.
If we also requires the quantization 
maps $Q_{\hbar}$ to be injective for each $\hbar$ 
 and that $Q_{\hbar}(\tilde{\gA}_0)$ is 
a $*$-subalgebra of $\gA_{\hbar}$ (for each $\hbar\in I$), then   a (non-commutative, associative) {\bf quantization deformation product} $\star_\hbar$ turns out to be
 implicitly defined \cite{Lan98} in $\tilde{\gA}_0$ from\footnote{The reader should pay attention to the fact that, in \cite{Lan98},  these hypotheses are included in the definition of {\em strict deformation quantization} as stated in  \cite[II Definition 1.1.2]{Lan98}. Furthermore our notion of {\em deformation quantization} adopted from \cite{Lan17} is named {\em continuous quantisation} in \cite{Lan98}.}
$Q_\hbar(f\star_\hbar g) = Q_\hbar(f) Q_\hbar(g)$ whose antisymmetric part is related to the Poisson structure.}
\end{remark}

To define a deformation  quantization it is not necessary starting from a continuous $C^*$-bundle, but it is sufficient to assign  
quantization maps satisfying some conditions. That is due to the following  result  adapted from \cite[Theorem II 1.2.4]{Lan98}.

\begin{theorem}\label{equivalenttheorem}
Let $X$ be a smooth manifold (possibly Poisson), $\{\gA_\hbar\}_{\hbar \in I}$
 a collection of $C^*$-algebras
 where
 $\gA_0 := C_0(X)$ and   $I\subset \bR$, with $0\in I$ 
as accumulation point. \\ Consider a collection of 
 maps $Q_\hbar : \tilde{\gA}_0  \to \gA_\hbar$, $\hbar \in I$,  $\tilde{\gA}_0 \subset \gA_0$ being a dense $*$-subalgebra (possibly $\gA_0$ itself or a    Poisson algebra), 
 satisfying
\begin{itemize}
\item[(a)] 
(i) and (ii) in (3) Definition \ref{def:deformationq}  ((iv) possibly);
\item[(b)]   $I \ni \hbar \mapsto ||Q_\hbar(f)||_\hbar$ is in $C_0(I)$  for every $f\in \tilde{\gA}_0$;
\item[(c)]  the Rieffel condition;
\item[(d)] the von Neumann condition;
\item[(e)]  $\overline{Q_\hbar(\gA_0)}= \gA_\hbar$ for every $\hbar \in I$;
\item[(f)]  $I\setminus \{0\}$  is discrete, or the $C^*$-algebras
$\gA_{\hbar}$ are identical for $\hbar\in I \setminus\{0\}$  and $I \ni \hbar \mapsto Q_\hbar(f)$ is continuous for every $f\in \tilde{\gA}_0$.
\end{itemize}
Then
the following facts are true.
\begin{itemize}
 \item[(1)]  There exists a unique continuous $C^*$-bundle $(I, \gA, \pi_\hbar)$ such that every $\{Q_\hbar(f)\}_{\hbar \in I}$, $f\in C_0(X)$
is a continuos 
section. In this case  definition \ref{def:deformationq} is valid (up to (iv) in (3) possibly).
\item[(2)] If a family of maps $Q'_\hbar : \tilde{\gA}_0  \to \gA_\hbar$, $\hbar \in I$ satisfies  the hypothes (a)-(f) and
$$||Q'_\hbar(f) -Q_\hbar(f)||_\hbar \to 0\quad  \mbox{for $\hbar \to 0$ and every $f \in \tilde{A}_0$,}$$  then the maps $Q'_\hbar$ determine the same $C^*$-bundle $(I, \gA, \pi_\hbar)$ as the maps
$Q_\hbar$.
\end{itemize}
\end{theorem}

\subsection{Weyl and Berezin quantizations maps on $\mathbb{R}^{2n}$}\label{WeylBerezin}
Let us consider  the commutative $C^*$-algebra
 $C_0(\mathbb{R}^{2n})$, where $\mathbb{R}^{2n}$ plays the role of the classical
 {\em phase space} equipped with the standard symplectic structure. The natural symplectic coordinates of $\bR^{2n}$ will be denoted by 
$(q,p)= (q^1,\ldots, q^n, p_1,\ldots, p_n)$. The Liouville measure will therefore coincide with the standard $2n$-dimensional Lebesgue measure $dqdp$ and
the Poisson bracket will take the form
$$\{f,g\} :=  \sum_{k=1}^n\frac{\partial f}{\partial p_k}  \frac{\partial g}{\partial q^k}-   \sum_{k=1}^n\frac{\partial g}{\partial p_k}  \frac{\partial f}{\partial q^k}.$$ 
The {\em strict deformation  Weyl quantization}
is a subcase of a more general quantization procedure acting on the space of Schwartz distributions $f \in {\cal S}'(\bR^{2n})$
 to which it associates a possibly unbounded operator, or more generally, just a quadratic form corresponding to the {\em quantum-Fourier antitransform of $f$}:
\beq
Q^W_\hbar(f) := \int_{\bR^{2n}} \sp\sp
 e^{i\overline{a \cdot X+b\cdot P}} 
 \widehat{f}(a,b) \frac{dadb}{(2\pi)^n}= 
 \int_{\bR^{2n}} \sp \sp e^{ia \cdot X} e^{ib\cdot P}e^{-\frac{i\hbar}{2} a\cdot b}  \widehat{f}(a,b) \frac{dadb}{(2\pi)^n}, \label{decwayl}
\eeq
where $\widehat{f} \in {\cal S}'(\bR^{2n})$ is the Fourier transform of $f$. The $\hbar$ appearing in $Q_\hbar^W$ takes place  both in the exponents $e^{\pm i\hbar a\cdot b/2}$ and in the definition of $P_k$, that is the unique selfadjoint extension of $-i\hbar \frac{\partial }{\partial x^k}$ ancting on ${\cal S}(\bR^{n})$.
As said, $Q_\hbar^W(f)$ is  only defined  in the sense of {\em quadratic forms} in general: 
the {\bf Weyl quantization map}, in this sense, is defined as
\beq \langle \psi, Q_\hbar^W(f) \phi \rangle := \frac{1}{(2\pi)^n} \int_{\bR^{2n}} \langle e^{-ia \cdot X}\psi,  e^{ib\cdot P} \phi \rangle e^{-\frac{i\hbar}{2} a\cdot b}  \widehat{f}(a,b) \: da db\:, \quad \hbar >0; \label{Wgen}\eeq
for all $\psi,\phi \in {\cal S}(\bR^{n}) $. By construction
$\bR^{2n} \ni (a,b) \mapsto \langle e^{-ia \cdot X}\psi,  e^{ib\cdot P} \phi \rangle e^{-\frac{i\hbar}{2} a\cdot b} $
is an element of ${\cal S}(\bR^{2n})$ as it is, up to a phase,  the Fourier transform of a function in that space, so that the definition is well posed.
There are cases where the quadratic form uniquely determines a true operator $Q_\hbar^W(f)$ which is  is in general
  {\em unbounded and densely defined}. This happens when the distributions $f$ are actually functions of various spaces.
Further restricting the space of functions one finally obtains  an everywhere defined and bounded operator as it is proper of the $C^*$-algebra appraoch. 
Let us examine some of those cases.

\begin{proposition}\label{QWgen} With the given definition of the quadratic form $\langle \psi, Q_\hbar^W(f)\phi\rangle$ (\ref{Wgen}), where $\hbar>0$, $f\in {\cal S}(\bR^{2n})$ and $\psi, \phi \in {\cal S}(\bR^n)$,  the following facts are valid.
\begin{itemize}
\item[(1)] If $f \in  {\cal S}(\bR^{2n})$, then  there is a unique e
$Q_\hbar^W(f)\in \gB(L^2(\bR^{2n},dx))$ whose quadratic form on ${\cal S}(\bR^n)$ is (\ref{decwayl}) and 
\beq ||Q^W_\hbar (f) || \leq \frac{1}{(2\pi)^n}  \int_{\bR^{2n}} |\widehat{f}(a,b)| dadb\:.\label{estimate} \eeq
where  the right-hand side does not depend on $\hbar$.
With the considered  hypotheses (\ref{Wgen}) is valid for generic $\psi,\phi \in L^2(\bR^n, dx)$ and the everywhere defined operator
$Q^W_\hbar (f)$.
\item[(2)]  If $f: \bR^{2n} \to \bC$ is a  Borel function which is constant in the variable $p$ and polynomially bounded in the variable $x$, then there is  a unique operator $Q_\hbar^W(f)\in \gB(L^2(\bR^{2n},dx))$  whose quadratic form on ${\cal S}(\bR^n)$ is  (\ref{decwayl}),  that is 
$Q_\hbar^W(f) = f(X)|_{{\cal S}(\bR^n)},$
where $f(X)$ is defined by spectral calculus. In particular, $$||Q_\hbar^W(f)|| = ||f(X)|_{{\cal S}(\bR^n)}||\leq ||f||_\infty \leq +\infty,$$
and the  bound  does not depend on $\hbar$.
\item[(3)]   If $f: \bR^{2n} \to \bC$ is a  Borel function which is constant in the variable $x$ and polynomially bounded in the variable $p$, then there is  a unique operator $Q_\hbar^W(f)$ on ${\cal S}(\bR^n)$ whose quadratic form  is (\ref{decwayl}), that is  
$Q_\hbar^W(f) = f(P)|_{{\cal S}(\bR^n)},$
where $f(P)$ is defined by spectral calculus. In particular, $$||Q_\hbar^W(f)|| = ||f(P)|_{{\cal S}(\bR^n)}||\leq ||f||_\infty \leq +\infty,$$
where  the  bound  does not depend on $\hbar$.
\item[(4)]  If $f(q,p) = f_1(q) f_2(p)$ with  $f_1, f_2 \in {\cal S}(\bR^n)$, then the operator 
$Q_\hbar^W(f_1f_2)\in \gB(L^2(\bR^{2n},dx))$ that exists due to  the case (1),
is completely determined by its quadratic form:
\beq \label{formfin}\langle \psi, Q_\hbar^W(f_1f_2))\phi \rangle =
\frac{1}{(2\pi)^{n/2}} \int_{\bR^n} \widehat{f}_2(b) \left\langle \psi, e^{ib\cdot P} f_1\left(X+ \frac{\hbar}{2}bI\right) \phi \right\rangle db\eeq
for $\psi,\phi \in L^2(\bR^n, dx)$.
\end{itemize}

\end{proposition}

\begin{proof} See Appendix \ref{appA}.
\end{proof}

When restricting to ${\cal S}(\bR^{2n})$, $Q^W$ defines a deformation quantization map according to definition \ref{def:deformationq} and re-adapting 
 \cite[Theorem 2.6.1]{Lan98} to our definitions.\footnote{This follows from Theorem \ref{equivalenttheorem} using \cite[Theorem II 2.6.1]{Lan98} combined with the
 fact that the map  $\hbar\mapsto Q_\hbar^W(f)$ is continuous for all $f\in \mathcal{S}(\bR^{2n})$ (as follows from the proof in \cite[Lemma II 2.6.2]{Lan98})
 and the continuity property of the norm, in that the map $\hbar\mapsto ||Q_\hbar^W(f)||$ is in $C_0(I)$ for all $f\in \mathcal{S}(\bR^{2n})$, as in turn follows from
 the proof in \cite[Thm II 2.6.5]{Lan98}.} Remarkably, it turns in particular out that  $Q^W_\hbar({\cal S}(\bR^{2n})) \subset \gB_{\infty}(L^2(\bR^n, dx))$.

\begin{theorem}  The family of maps  $Q^W_\hbar : {\cal S}(\bR^{2n}) \to \gB(L^2(\bR^n, dx)), \quad  \hbar>0\:, $ constructed as in (1) 
of Proposition \ref{QWgen} together  with $Q^W_0 := id_{ {\cal S}(\bR^{2n})}$
defines a deformation quantization  of the Poisson manifold $(\bR^{2n}, \{\cdot, \cdot\})$
 according to definition \ref{def:deformationq},
where $I:= [0,+\infty)$ and
$$Q^W_\hbar({\cal S}(\bR^{2n})) \subset \gA_\hbar = \gB_{\infty}(L^2(\bR^n, dx)) $$
for all $\hbar >0$.
  This deformation quantization enjoys the following properties 
\begin{itemize}
\item[(1)] it is  strict, i.e., $\overline{Q^W_\hbar({\cal S}(\bR^{2n}))}= \gB_{\infty}(L^2(\bR^n, dx))$ for all $\hbar \in I\setminus\{0\}$,
\item[(2)] the maps  $Q^W_\hbar$ are injective  for all $\hbar \in I$,
\item[(3)] $Q^W_\hbar( {\cal S}(\bR^{2n}))$ is a $*$-subalgebra of $\gB_{\infty}(L^2(\bR^n, dx))$
 for all $\hbar \in I\setminus\{0\}$ so that $\star_\hbar$ can be defined.\\
\end{itemize}
\end{theorem}

\noindent 
 The {\bf Berezin quantization map} on $\bR^{2n}$, as a quadratic form, is defined as \cite{Lan98}
 (see \cite{Com} or \cite{Lan98} for a detailed presentations also on completely general symplectic manifolds in place of $\bR^{2n}$ within the coherent and pure state quantization approach)
\beq
\langle \psi, Q_\hbar^B(f) \phi\rangle  := \int_{\bR^{2n}} f(q,p)  \langle \psi, \Psi_\hbar^{(q,p)}\rangle \langle \Psi_\hbar^{(q,p)}, \phi \rangle \frac{dq dp}{(2\pi\hbar)^n},  \quad \hbar >0; \label{defQB}
\eeq
where $\psi, \phi \in L^2(\bR^n, dx)$,  $f \in L^\infty\left(\bR^{2n}, \frac{dqdp}{(2\pi \hbar)^n}\right)$. Above, for any  given $(q,p) \in \bR^{2n}$,
\beq\label{vecPSI}
\Psi_\hbar^{(q,p)}(x) := \frac{e^{-\frac{i}{2}p\cdot q/\hbar} e^{ip\cdot x/\hbar} e^{-(x-q)^2/(2\hbar)}}{(\pi \hbar)^{n/4}} \:, \quad x \in \bR^n\:, \hbar >0;\:
\eeq
is a {\em unit} vector in $L^2(\bR^n, dx)$ also called a {\bf Schr\"{o}dinger coherent state}.
It turns out that  \cite{Lan98}  the integral (\ref{defQB}) satisfies,
$\langle \psi, Q^B_\hbar(1) \psi \rangle = 1$
 for $\psi \in L^2(\bR^n,dx)$ with $||\psi||=1$, hence
$ |\langle \psi, Q_\hbar^B(f) \psi\rangle | \leq ||f||_\infty$.
By a polarization argument,
$ |\langle \psi, Q_\hbar^B(f) \phi\rangle | \leq ||f||_\infty||\psi||\: ||\phi||\:,\quad \psi, \phi \in L^2(\bR^n, dx)\:.$
Therefore, by the Riesz lemma, $Q_\hbar^B(f) \in \gB(L^2(\bR, dx))$ and
\beq
||Q_\hbar^B(f)|| \leq ||f||_\infty\:, \quad f \in L^\infty\left(\bR^{2n}, \frac{dqdp}{(2\pi \hbar)^n}\right) \:.\label{boundQB}
\eeq
Let us summarize the general properties of $Q_\hbar^B$ (Theorems II 1.3.3 and II 1.3.5 in \cite{Lan98} specialized to $\bR^{2n}$) relevant for our work.
Positivity condition (2) is one of  the most important improvements which differentiate the Berezin quantization from the Weyl quantization. 

\begin{theorem}\label{teoQB1} The linear map $$Q^B_\hbar :L^\infty\left(\bR^{2n}, \frac{dqdp}{(2\pi \hbar)^n}\right)   \to\gB(L^2(\bR^n, dx)), \quad  \hbar>0;$$ defined as 
\beq
Q_\hbar^B(f)   := \int_{\bR^{2n}} f(q,p) |\Psi_\hbar^{(q,p)}\rangle \langle \Psi_\hbar^{(q,p)}| \frac{dq dp}{(2\pi\hbar)^n}, \label{defQB2}
\eeq
 in the sense of (\ref{defQB}) and satisfying  (\ref{boundQB}) enjoys  the following properties.
\begin{itemize}
\item[(1)]  $Q_\hbar^B(1) = I_{L^2(\bR^n, dx)}$, where $1 \in  L^\infty\left(\bR^{2n}, \frac{dqdp}{(2\pi \hbar)^n}\right)$ is the constant $1$ map;
\item[(2)] $f\geq 0$ a.e. implies $Q_\hbar^B(f)\geq 0$ for $f \in L^\infty\left(\bR^{2n}, \frac{dqdp}{(2\pi \hbar)^n}\right)$;
\item[(3)]  $Q_\hbar^B(\overline{f}) = Q_\hbar^B(f)^*$ for $f \in L^\infty\left(\bR^{2n}, \frac{dqdp}{(2\pi \hbar)^n}\right)$;
\item[(4)]  $Q_\hbar^B(C_0(\bR^{2n})) = \gB_\infty(L^2(\bR^n, dx))$, $\quad Q_\hbar^B(L^1\cap L^\infty) = \gB_1(L^2(\bR^n, dx))$;
\item[(5)]  $Tr(Q_\hbar^B(f)) = \int_{\bR^{2n}}\frac{dqdp}{(2\pi \hbar)^n}f(q,p)$,   for $f\in L^{1}(\bR^{2n}, \frac{dqdp}{(2\pi \hbar)^n})\cap L^{\infty}(L^2(\bR^{2n}, \frac{dqdp}{(2\pi \hbar)^n}))$.
\end{itemize}
\end{theorem}

Finally, the Berezin map defines a deformation quantization according to the following pair of theorems (readaptation 
of Theorem II 2.41 and Proposition II 2.6.3 in \cite{Lan98}, item  (5)  -- (2.117) in \cite{Lan98} --  is obtained by simply comparing (\ref{vecPSI}) 
and (\ref{Wgen}) by writing down the action of $e^{\hbar \Delta_{2n}/4}$
 in terms of Gaussian convolution).\footnote{Similar as in the case of Weyl quantization,
 this follows from Theorem \ref{equivalenttheorem} using \cite[Theorem II 2.4.1]{Lan98}, the fact that the map  $\hbar\mapsto Q_\hbar^B(f)$ is continuous  for all $f\in C_c(\bR^{2n})$, and that the map
 $\hbar\mapsto ||Q_\hbar^B(f)||$ is in $C_0(I)$ for all $f\in C_c(\bR^{2n})$ (as follows from the proof of \cite[Thm II 2.6.5]{Lan98}.)}%

\begin{theorem}\label{BerezinS}   The family of maps  $Q^B_\hbar : C_c^\infty(\bR^{2n}) \to \gB(L^2(\bR^n, dx)), \quad \hbar>0\:,$
defined in  theorem \ref{teoQB1}  together with $Q^B_0 := id_{C_c^\infty(\bR^{2n})}$
defines a deformation quantization  of the Poisson manifold $(\bR^{2n}, \{\cdot, \cdot\})$
 according to definition \ref{def:deformationq},
where $\tilde{\gA}_0 := C_c^\infty(\bR^{2n})$, $I:= [0,+\infty)$, and
$$Q^B_\hbar(C_c^\infty(\bR^{2n})) \subset \gA_\hbar := \gB_{\infty}(L^2(\bR^n, dx)) \quad \mbox{for all $\hbar >0$.}$$
  This deformation quantization enjoys the following properties,
\begin{itemize}
\item[(1)] it is  strict, i.e., $\overline{Q^B_\hbar(C_c^\infty(\bR^{2n}))}= \gB_{\infty}(L^2(\bR^n, dx))$ for all $\hbar >0$.
\item[(2)] the maps  $Q^B_\hbar$ are injective  for all $\hbar \in I$.
\end{itemize}
\end{theorem}

A weaker version of the result above is obtained when working directly on $\gA_0 := C_0(\bR^{2n})$ since $Q_\hbar^B$ is defined thereon.

\begin{theorem} \label{BerezinC0} The family of maps  $Q^B_\hbar : C_0(\bR^{2n}) \to \gB(L^2(\bR^n, dx))\:,\quad  \hbar >0\:,$
defined in  theorem \ref{teoQB1}  together with $Q^B_0 := id_{C_0(\bR^{2n})}$
give rise to  a deformation quantization  of the Poisson manifold $(\bR^{2n}, \{\cdot, \cdot\})$, except for the 
Dirac-Groenewold-Rieffel condition,
 according to definition \ref{def:deformationq},
where $I:= [0,+\infty)$ and
$$Q^B_\hbar(C_0(\bR^{2n})) \subset \gA_\hbar = \gB_{\infty}(L^2(\bR^n, dx))\:, \quad \hbar >0, $$
for all $\hbar >0$.
  This deformation quantization enjoys the following properties 
\begin{itemize}
\item[(1)] it is  strict and more strongly $Q^B_\hbar(C_0(\bR^{2n}))= \gB_{\infty}(L^2(\bR^n, dx))$ for all $\hbar > 0$;
\item[(2)] it is {\bf asymptotically equivalent} to $Q^W_\hbar$ on ${\cal S}(\bR^{2n})$, i.e., 
the map $$I \ni \hbar \to ||Q_\hbar^B(f)- Q_\hbar^W(f)||_\hbar $$ is continuous   if $f\in {\cal S}(\bR^{2n})$ and 
$||Q_\hbar^B(f)- Q_\hbar^W(f)||_\hbar \to 0\quad \mbox{for $\hbar \to 0^+$.}$
\item[(3)] If $f \in {\cal S}(\bR^{2n})$,
\beq Q^B_\hbar(f) = Q^W_\hbar\left(e^{\hbar\Delta_{2n}/4} f\right) \label{Delta2n}\:,\quad \hbar >0;\eeq
where the exponential denotes the one-parameter semigroup generated by the selfadjoint extension  on $L^2\left(\bR^{2n}, \frac{dqdp}{(2\pi \hbar)^n}\right)$
of 
$\Delta_{2n} := \sum_{k=1}^n\frac{\partial^2}{\partial q^{k2}} +  \sum_{k=1}^n\frac{\partial^2}{\partial p_k^{2}} \:,$
initially defined on $C_c^\infty(\bR^{2n})$.
\end{itemize}
\end{theorem}

\noindent As  final technical result we quote the  following one which will be of crucial relevance  in our paper.

\begin{proposition}\label{prooff} Referring to the definition (\ref{defQB2}) of $Q_\hbar^B$, the following facts are true. 
\begin{itemize}
\item[(a)] Consider  $f \in C_0(\bR^n)$. If interpreting $f$ as a function in $L^\infty(\bR^{2n}, \frac{dqdp}{(2\pi \hbar)^n})$
constant in the variable $p$, then
\beq
||Q_\hbar^B(f)- f|| \to 0 \quad \mbox{for $\hbar \to 0^+$}\:,
\eeq
where both operators are defined in $L^2(\bR^{n}, dx)$ and 
$(f\psi)(x):= f(x)\psi(x)$ if $\psi \in L^2(\bR^{n}, dx)$ and $x\in \bR^n$.
\item[(b)] Consider $f\in {\cal S}(\bR^n)$. If interpreting $f$ as a function in $L^\infty(\bR^{2n}, \frac{dqdp}{(2\pi \hbar)^n})$
constant in the variable $q$, then
 it holds
\beq
||Q_\hbar^B(f) - \check{f}_\hbar*|| \to 0 \quad \mbox{for $\hbar \to 0^+$}\:,
\eeq
where both operators are defined in $L^2(\bR^{n}, dx)$, where $\check{f}_\hbar := {\cal F}_\hbar^{-1}(f)$ and $(\check{f}_\hbar*)(\psi):= \check{f}_\hbar*\psi$ if $\psi \in L^2(\bR^{n}, dx)$.
In particular, for every chosen $t>0$,
$$||Q_\hbar^B(e^{-tp^2}) - e^{t\hbar^2 \Delta}|| \to 0 \quad \mbox{for $\hbar \to 0^+$}\:.$$
\end{itemize}
\end{proposition}

\begin{proof}
See Appendix \ref{appA}.
\end{proof}

\subsection{Classical limit}\label{claslim}
In the rest of the paper  we shall deal with  the notion of the {\em classical limit}.  Let us denote by $H_{\hbar}$ some ($\hbar$-dependent) Hamiltonian encoding a quantum theory on the Hilbert space $\mathcal{H}=L^2(\mathbb{R}^n,dx)$, and by $\{\psi_{\hbar}\}_{\hbar>0}$ a sequence of corresponding normalized eigenvectors of $\{H_{\hbar}\}_{\hbar> 0}$, of course, assuming they exist.  We want to investigate what happens in the regime of very small $\hbar$ 
to the interplay of these states and the observables of the system.
It is a hard problem to capture the behaviour of the sequence $\{\psi_{\hbar}\}_{\hbar>0}$ in $\mathcal{H}$, and in general it has not even a limit herein. However in a $C^*$-algebraic setting things can become more smooth. The most common way is to consider the corresponding algebraic vector states $\{\omega_\hbar\}_{\hbar>0}$ on $\gB(\mathcal{H})$, defined by 
\begin{align}
\omega_\hbar(\cdot).=\langle\psi_\hbar,(\cdot)\psi_\hbar\rangle. \label{vectorstate}
\end{align}
A natural question is to find a suitable set of physical observables for which the sequence of states $\{\omega_\hbar\}_{\hbar>0}$ converge (in some topology) to a state on a certain commutative algebra, establishing then a link between the quantum and classical theory. If it does, then the limit is also called the {\bf classical limit} of the sequence of eigenvectors $\{\psi_{\hbar}\}_{\hbar>0}$. This strongly depends on this sequence which behaviour is in general unknown, especially when $\hbar$ changes. One does therefore not escape making assumptions, but instead of imposing conditions on the eigenvectors $\psi_{\hbar}$ and the Hilbert space $\mathcal{H}$ we impose conditions on the algebra of observables, as will become clear soon. For purpose of this paper we are interested in the semiclassical behaviour of $\hbar$-dependent eigenvectors $\{\psi_{\hbar}\}_{\hbar>0}$ in $L^2(\mathbb{R}^n)$ corresponding to 
\begin{itemize}\label{twocases}
\item[(1)] operators of the form $Q_\hbar^B(e_h)$, where $e_h\in C_0(\mathbb{R}^{2n})$ and $Q_\hbar^B$ is the Berezin quantizaton map defined by \eqref{defQB}, and $e_h$ is in principle related with a classical hamiltonian $h$;
\item[(2)] Schr\"{o}dingers operator of the form $H_\hbar=-\hbar^2\Delta + V$, (with $V$ a potential satisfying some conditions) densely defined on the Hilbert space $L^2(\mathbb{R}^n)$.
\end{itemize}
We will prove that under some hypotheses the vector state \eqref{vectorstate} associated to such sequences of eigenvectors admits a classical limit with respect to the observables  $Q_\hbar^B(f)$, $f\in C_0(\mathbb{R}^{2n})$ with $\overline{f}=f$.\footnote{By $(ii)$ in Definition \ref{def:deformationq} it follows that $Q_\hbar^B(f)$ is self-adjoint, as should be the case in order to define a physical observable.}  By Theorem \ref{BerezinC0} the set $Q_\hbar^B(C_0(\mathbb{R}^{2n}))$ equals the algebra of compact operators on $L^2(\mathbb{R}^n, dx)$, so that a relatively large class of observables is considered. Using these observables the statement that the sequence of eigenvectors $\{\psi_\hbar\}_{\hbar>0}$ admits a classical limit now means that
\begin{align}
\omega_0(f)=\lim_{\hbar\to 0^+}\omega_\hbar(Q^B_{\hbar}(f)), \label{classical limit}
\end{align}
exists for all $f\in C_0(\mathbb{R}^{2n})$ and and defines a state $\omega_0$ on $C_0(\mathbb{R}^{2n})$. The state $\omega_0$ may be regarded as the classical limit of the  sequence of vector states $\omega_{\hbar}$ defined in \eqref{vectorstate}. From the mathematical side, by the Riesz Representation Theorem the statement in \eqref{classical limit} means that for all $f\in C_0(\mathbb{R}^{2n})$ one has
\begin{align}
 \mu_0(f)=\lim_{\hbar \to 0^+}\int_{\bR^{2n}}f(q,p) d\mu_{\psi_\hbar}(q,p),\label{classical limit 3} 
\end{align} 
where $\mu_0$ is the probability measure corresponding to the state $\omega_0$ and $\mu_{\psi_\hbar}$, with $\hbar>0$,  is a probability measure on $\bR^{2n}$ with density $B_{\psi_\hbar}(\sigma):=|\langle \Psi_\hbar^{\sigma},\psi_\hbar\rangle|^2$ (where $\Psi_\hbar^{\sigma}$ is the coherent state vector defined by \eqref{vecPSI}) also called the {\em Husimi density function} associated to the unit vector $\psi_\hbar$. In other words 
$\mu_{\psi_\hbar}$ is given by
\begin{align}
d\mu_{\psi_\hbar}(q,p)=|\langle \Psi_\hbar^{(q,p)},\psi_\hbar\rangle|^2dqdp/(2\pi\hbar)^n\:,\quad  \sigma\in\mathbb{R}^{2n}\:.
\end{align}
It still remains a challenging problem to work with \eqref{classical limit} (or equivalently with the Husimi function) as there is a priori no information on the eigenvectors $\{\psi_{\hbar}\}_{\hbar>0}$ corresponding to the quantum Hamiltonians.
Nonetheless, we will see that the semi-classical behaviour of the sequence of eigenvectors $\{\psi_{\hbar}\}_{\hbar>0}$ corresponding to case (1) above is encoded by the algebraic properties of the function $e_h$, which are relatively well manageable and allow us to prove the existence of the classical limit. The semi-classical behaviour of eigenvectors corresponding to  Schr\"{o}dinger operators at first sight seems a more complicated issue. However, it turns out that also these operators are related to quantization maps, in the sense that $e^{-tH_\hbar}$ is an approximation\footnote{The precise meaning of this will become clear in the next sections.} of the $Q_\hbar^B(e^{-th})$, for  $h(q,p) = p^2+V(q)$. Indeed, we shall see that under some assumptions  on $V$ this happens and also in that case the existence of the classical limit is guaranteed and well-understood: the limit only depends on the nature  of the classical potential  $V$.

\section{Semiclassical properties of Berezin quantization maps}\label{SemiclassicalpropertiesofBerezinquantizationmaps}
We prove the existence of the classical limit of a sequence of eigenvectors for operators of the form $Q_\hbar^B(e)$ where $e$ is some function in $C_0(\bR^{2n})$. We hereto first introduce and prove some preliminaries on Berezin quantization. We then prove an important localization result of a sequence of eigenvectors of $Q_\hbar^B(e)$ in the semi-classical regime $\hbar\to 0^+$. 
Finally, we prove the existence of the classical limit in the form of Thm. \ref{MT2} and Thm. \ref{thm:claslimNW2} respectively, where the first theorem considers the case when no symmetry is present, whilst the latter theorem deals with a certain symmetry.

\subsection{Isometric embedding $L^2(\bR^n,dx) \subset L^2(\bR^{2n}, d\mu_\hbar)$  and $Q^B$-equivariant group representations}
We are going to introduce some preparatory results. necessary to study the classical limit.  Several of them   actually hold in a more general setting of a coherent pure state Berezin quantization on a symplectic manifold \cite{Lan98,Lan17}. We start with a proposition on the Berezin quantization
specialized to $\bR^{2n}$. (a) is actually \cite[Proposition II 1.5.2]{Lan98} specialized to the symplectic manifold $(\bR^{2n}, \sum_{k=1}^n dp_k \wedge dq^k)$  where $\mu_\hbar=\frac{dpdq}{(2\pi\hbar)^n}$, and $\mathcal{H}_\hbar=L^2(\mathbb{R}^n,dx)$ for $\hbar>0$. We now prove the proposition.
\begin{proposition}\label{newhilbertspace2} Referring to the coherent state vectors (\ref{vecPSI}) $\Psi_\hbar^{\sigma} \in L^2(\mathbb{R}^n,dx)$ used to construct the quantization Berezin map 
 (\ref{defQB2}),
\begin{itemize}
\item[(a)]
 there exists  an  isometry $W: L^2(\mathbb{R}^n,dx) \to L^2(\bR^{2n}, d\mu_\hbar)$,  completely defined by
\beq (W\phi)(q,p)=\langle \Psi_{\hbar}^{(q,p)},\phi\rangle
= \int_{\bR^n} e^{\frac{ip\cdot q}{2\hbar}} e^{-\frac{ip\cdot x}{\hbar}}
e^{-\frac{(x-q)^2}{2\hbar}} \phi(x) \frac{dx}{(\pi \hbar)^{n/2}}
\:, \label{alternate}\eeq
in particular  $W^*W=I_{L^2(\mathbb{R}^n,dx)}$ and  $p:= WW^*:  L^2(\bR^{2n}, d\mu_\hbar) \to  L^2(\bR^{2n}, d\mu_\hbar)$ is the  orthogonal  projector onto $\mbox{ran}(W)= \overline{\mbox{ran}(W)}$;
\item[(b)] It holds $W({\cal S}(\bR^n)) \subset {\cal S}(\bR^{2n})$.
\end{itemize}
\end{proposition}

\begin{proof}
From now on 
\begin{align}
||\Phi||_{L^2(\bR^{2n}, d\mu_\hbar)}=\int_{\bR^{2n}} d\mu_\hbar(\sigma)|\langle\Psi_{\hbar}^\sigma,\phi\rangle|^2.\label{newnorm2}
\end{align}
(a) By property (1) in theorem \ref{teoQB1} , for every $\phi \in L^2(\bR^n, dx)$ the associated function 
$\bR^{2n} \ni \sigma \mapsto (W\phi)(\sigma):=\langle\Psi_{\hbar}^{\sigma},\phi\rangle$ satisfies
$W\phi \in L^2\left(\bR^{2n},d\mu_\hbar\right)\quad \mbox{and}\quad ||W\phi||_{L^2(\bR^{2n},d\mu_\hbar)}= ||\phi||_{L^2(\bR^n, dx)}\:.$
Hence $W^*W=I$, the remaining part is a standard property of isometric maps in Hilbert spaces.\\
(b) With trivial manipulations,
$(W\psi)(q,p)  =  \frac{1}{(\pi \hbar)^{n/4}} \int_{\bR^n}e^{-i p\cdot z/\hbar} e^{-(2z-x)^2/(2\hbar)} \psi(z+q/2) dz$,
so that, where $a,b,c,d,k$ are multiindices, passing the derivatives under the integration symbol (for dominated convergence theorem) and using 
integration by parts since, in particular,  $\psi \in {\cal S}(\bR^n)$,
$$p^a q^c \partial_p^b \partial^d_q (W\psi)(q,p)= \int_{\bR^n}e^{-i p\cdot z/\hbar} e^{-(2z-x)^2/(2\hbar)} \sum_ kp^{abcd}_k(z,q) \partial^{k}_z \psi(z+q/2) dz,$$
where the sum is finite and $p^{abcd}_k(z,q)$ are polynomials in $z$ and $q$. Hence, for some constant $K_{abcd}>0$, it holds
 $|\sum_ kp(z,q)_k \partial^{k}_z \psi(z+q/2)| \leq K_{abcd}$, so that
$$|p^a   q^c \partial_p^b\partial^d_q (W\psi)(q,p)| \leq 
K_{abcd} \int_{\bR^n} e^{-(2z-x)^2/(2\hbar)} dz = C_{abcd}<+\infty\quad \mbox{for all $(q,p) \in \bR^{2n}$.}$$
Arbitrariness of $a,b,c,d$ implies the thesis.
\end{proof}

If there moreover exists an  action of a group $G$ acting by symplectomorphisms on the manifold $(\bR^{2n}, \sum_{k=1}^n dp_k \wedge dq^k)$, it follows that the quantization maps $Q_{\hbar}^B$ are equivariant
under a suitable unitary representation of $G$ in $L^2(\mathbb{R}^n,dx)$. The precise statement is given in the proposition below.
For any $g\in G$ we denote the pullback of the action of $G$ on functions $f: \bR^{2n} \to \bC$ by $\zeta_g$, i.e., 
$$(\zeta_g f)(\sigma)=f(g^{-1}\sigma)\:,  \quad \sigma\in \bR^{2n}\:, g \in G\:.$$

\begin{proposition}\label{equivariance2}
Let $G$ be a group acting by symplectomorphisms on  the symplectic manifold $(\bR^{2n}, \sum_{k=1}^n dp_k \wedge dq^k)$.  Then there exists a unitary representation  $U:G\to B(L^2(\mathbb{R}^n,dx))$ such that  
the maps $Q_\hbar^B$ are $\zeta$-equivariant,
\begin{align}
U_gQ_{\hbar}^B(f)U^{-1}_g=Q_{\hbar}^B(\zeta_g f)\:,  \ \ g\in G,\  f\in L^\infty\left(\bR^{2n},d\mu_\hbar\right). \label{equivariantofQ2}
\end{align}
The representation $U$ is completely defined by the requirement
\beq\label{equivariantofQ23}
WU_g\psi = p\zeta_g(W \psi)\quad \mbox{for every $\psi \in L^2(\bR^n,dx)$ and $g\in G$.}
\eeq
\end{proposition}
\begin{proof} 
Let us write $\mathcal{H}^\hbar:=L^2\left(\bR^{2n},d\mu_\hbar\right)$.
By definition of $W$ for any $\phi\in\mathcal{H}_\hbar$,  as $\Phi := W\phi$  is a function on $\bR^{2n}$, we can define the operator $u_g
: L^2\left(\bR^{2n},d\mu_\hbar\right) \to L^2\left(\bR^{2n},d\mu_\hbar\right)$ by $(u_g\Phi)(\sigma):=\Phi(g^{-1}\sigma)\:.$
 This map is isometric, since the action of the  group preserves  the Liouville measure $\mu_L$ (and thus also $d\mu_\hbar$) as the action  is made of symplectomorphisms, and it is finally  surjective as the reader immediately proves since the action of each $g$ is bijective. 
By construction we also have $u_g(\mbox{ran}(W)) \subset  \mbox{ran}(W)$.
By construction $u_1 = I$ and $u_gu_{g'}= u_{g\cdot g'}$. Next
$U_g:\mathcal{H}_{\hbar}\to \mathcal{H}_{\hbar}$ is  defined by $U_g :=W^*u_gW : \mathcal{H}_{\hbar} \to \mathcal{H}_{\hbar}$. This operator is unitary since $u_g$ is unitary and
$U^*_g U_g  =   W^*u^*_gWW^*u_gW =W^*u^*_gpu_gW =W^*u^*_gu_gW= W^*W=I$ together with 
$U_g U^*_g  =  W^*u_gWW^*u^*_gW =W^*u_gpu^*_gW  = W^*u_gpu_{g^{-1}}W
W^*u_gu_{g^{-1}}W = W^*W=I$.
By  construction,
$U_1 = I$ and $U_gU_{g'}= U_{g\cdot g'}$ so that $G \ni g \mapsto U_g$ is a unitary representation.
Let us finally prove the equivariance property. For any $\phi,\psi\in \mathcal{H}_{\hbar}$, we now compute
\begin{align}
&\langle \phi, Q^B_{\hbar}(\zeta_{g^{-1}} f)\psi\rangle=
\int_{\bR^{2n}}d\mu_\hbar(\sigma)\overline{(W\phi)(\sigma)}(W\psi)(\sigma)f(g\sigma)=
\int_{\bR^{2n}}d\mu_\hbar(\sigma)\overline{\Phi(\sigma)}\Psi(\sigma)f(g\sigma) \nonumber \\& =\int_{\bR^{2n}}d\mu_\hbar(\sigma)\overline{\Phi(g^{-1}\sigma)}\Psi(g^{-1}\sigma)f(\sigma)=
\int_{\bR^{2n}}d\mu_\hbar(\sigma)\overline{({u}_g\Phi)(\sigma)}({u}_g\Psi)(\sigma)f(\sigma)\nonumber \\
& =
\int_{\bR^{2n}}d\mu_\hbar(\sigma)\overline{(WU_g\phi)(\sigma)}(WU_g\psi)(\sigma)f(\sigma)
=
\langle U_g\phi , Q^B_{\hbar}(f) U_g\psi\rangle=\langle \phi, U_g^*Q^B_{\hbar}(f)U_g\psi\rangle, \nonumber 
\end{align}
where we used the fact that the measure $d\mu_\hbar$ is $G$-invariant. Since this holds for any $\phi,\psi\in\mathcal{H}_\hbar$, we can conclude that 
$U^*_gQ_{\hbar}^B(f)U_g=Q_{\hbar}^B(\zeta_{g^{-1}} f)\:.$
replacing $g$ for $g^{-1}$ and noticing that $U_{g^{-1}} =(U_g)^{-1} = U^*_g$, we have the thesis.
The last statement is a rephrasing of $U_g=W^*u_gW$, using $W^*W=I$ and $WW^* = p$.
\end{proof}
Notice  that both $W$ and $U$ generally depend on the value of $\hbar$. 

\subsection{Localization of eigenvectors}
We are about establishing establish a first important  result concerning the localization of eigenvectors, which allows us to prove the existence of the classical limit in section \ref{Generalizations}. This topic is not new, several similar results have been achieved over the years \cite{Charles,Com,DimSjo,Porte,Ven2020,MZ} where different classes of neighbourhoods of localisation were exploited.\footnote{Our neighbourhoods ${\cal V}_\epsilon$ are particularly adapted to the remaining proofs of our work.} In such works one typically studies semiclassical defect measures induced by the relevant eigenvectors and uses techniques from pseudo-differential calculus. In this section, we will not go into such details, we however aim to provide a rather  neat algebraic approach that forms the basis for spontaneous symmetry breaking as explained in the next sections. 

\begin{proposition}\label{localization2}
Consider the manifold $(\bR^{2n}, \sum_{k=1}^n dp_k \wedge dq^k)$ with associated Berezin quantization maps $Q^B_\hbar$.
Given a real-valued $e\in  C_0(\bR^{2n})$, let $\{\phi_{\hbar}\}_{\hbar}\subset L^2(\mathbb{R}^n)$ be a sequence of eigenvectors of $Q^B_\hbar(e)$ with  eigenvalues $\{\lambda_{\hbar}\}_{\hbar}$  such that, for some $\Lambda\in \text{ran}(e)$ is such that
\begin{align}
&\lambda_\hbar\to \Lambda    \quad \mbox{for $\hbar\to 0^+$.}
\end{align}
The following facts are true where $\Phi_{\hbar}:= W\varphi_\hbar$ as before.
\begin{itemize}
\item[(1)] Referring to any open neighborhood of the set $e^{-1}(\Lambda)$ of the form
\begin{align}
{\cal V}_\epsilon:= e^{-1}((\Lambda-\epsilon, \Lambda+\epsilon)), \label{nbhV2}
\end{align}
for every given $\epsilon>0$ one has
\begin{align}
||\Phi_\hbar||_{L^2(\bR^{2n}\setminus {\cal V}_\epsilon, d\mu_\hbar)}\to 0, \quad \mbox{for}\quad \hbar \to 0^+ \label{lim}\:.
\end{align}
\item[(2)] If $e^{-1}(\Lambda) = \{\sigma_0\} \in \bR^{2n}$ and the family of sets $\{{\cal V}_\epsilon\}_{\epsilon>0}$ is a fundamental system of neighbourhoods of $\sigma_0$
then
$\langle \varphi_\hbar, Q_\hbar(f) \varphi_\hbar \rangle \to f(\sigma_0)\quad \mbox{as $\hbar \to 0^+$  for every $f\in C_0(\bR^{2n})$.}$
\end{itemize}
\end{proposition}

\begin{proof}
(1) The thesis arises  from the following fact we shall prove later
\beq
\langle \varphi_\hbar, Q^{B}_\hbar((e-\Lambda)^2) \varphi_\hbar \rangle \to 0\quad \mbox{for $\hbar \to 0$,} \label{central2}
\eeq
where we are using $Q_\hbar^B$ defined on $L^\infty\left(\bR^{2n},d\mu_\hbar\right)$.
Indeed, (\ref{central2}) can be rephrased to
$$\int_{\bR^{2n}} (e(\sigma) -\Lambda)^2 |\Phi_\hbar(\sigma)|^2 d\mu_\hbar \to 0 \quad \mbox{for $\hbar \to 0$}\:,  $$
so that also
$$\int_{\bR^{2n}\setminus {\cal V}_\epsilon} (e(\sigma) -\Lambda)^2 |\Phi_\hbar(\sigma)|^2 d\mu_\hbar  \to 0 \quad \mbox{for $\hbar \to 0$}\:,  $$
as well because the integrand is non-negative.
 However, by definition of ${\cal V}_\epsilon$, $|e(x)-\Lambda| \geq\epsilon$ if $x \in \bR^{2n}\setminus {\cal V}_\epsilon$ and thus 
$$0 \leq  \int_{\bR^{2n}\setminus {\cal V}_\epsilon} |\Phi_\hbar(x)|^2 d\mu_\hbar  \leq \frac{1}{\epsilon^2} \int_{\bR^{2n}\setminus {\cal V}_\epsilon} (e(\sigma) -\Lambda)^2 |\Phi_\hbar(\sigma)|^2 d\mu_\hbar  \to 0,$$
which implies (\ref{lim}). To conclude, it is sufficient to prove (\ref{central2}). To this end we have, by using linearity of $Q^{B}_\hbar$, the fact that $||\varphi_\hbar||=1$ and $Q_\hbar^B(e)\varphi_\hbar=\lambda_\hbar\varphi_\hbar$ in particular,
\begin{align*}
\langle \varphi_\hbar, Q^{B}_\hbar((e-\Lambda)^2) \varphi_\hbar \rangle &= \langle \varphi_\hbar, Q^{B}_\hbar(e^2) \varphi_\hbar \rangle
-2\Lambda \langle \varphi_\hbar, Q^{B}_\hbar(e) \varphi_\hbar \rangle + \Lambda^2 \langle \varphi_\hbar, \varphi_\hbar \rangle
\\&=  \langle \varphi_\hbar, Q^{B}_\hbar(e^2) \varphi_\hbar \rangle
-2\lambda_\hbar \Lambda\langle \varphi_\hbar,  \varphi_\hbar \rangle + \Lambda^2 \langle \varphi_\hbar, \varphi_\hbar \rangle\:.
\end{align*}
That is,
$\langle \varphi_\hbar, Q^{B}_\hbar((e-\Lambda)^2) \varphi_\hbar \rangle =  \langle \varphi_\hbar, Q^{B}_\hbar(e^2) \varphi_\hbar \rangle + \Lambda^2-2\lambda_\hbar \Lambda\:.$
On the other hand it also holds, 
\beq \langle \varphi_\hbar, Q^{B}_\hbar(e^2) \varphi_\hbar \rangle \to \Lambda^2\:.\label{central25}\eeq
In fact, we already know from theorem \ref{BerezinC0} (and exactly here we exploit  $e \in C_0(\bR^{2n})$)   that $||Q^{B}_\hbar(e^2)  - Q^{B}_\hbar(e)Q^{B}_\hbar(e)  || \to 0$
where the norms are the $C^*$ ones. This, in turn,  entails 
$||Q^{B}_\hbar(e^2)\varphi_\hbar  - Q^{B}_\hbar(e)Q^{B}_\hbar(e)\varphi_\hbar  ||| \to 0$
referring to the Hilbert space norms. Using the hypothesis $\lambda_\hbar\to \Lambda$ and $Q_\hbar^B(e)\varphi_\hbar=\lambda_\hbar\varphi_\hbar$ we deduce
$||Q^{B}_\hbar(e^2- \Lambda^2)\varphi_\hbar|| \leq ||Q^{B}_\hbar(e^2)\varphi_\hbar- \lambda_\hbar^2\varphi_\hbar||+|\lambda_\hbar^2 -\Lambda^2|\:||\varphi_\hbar|| \to 0\:,$
so that, from the Cauchy-Schwartz inequality,
$\left|\langle \varphi_\hbar, Q^{B}_\hbar(e^2) \varphi_\hbar \rangle - \Lambda^2\right|  =\left|\langle \varphi_\hbar, Q^{B}_\hbar(e^2-\Lambda^2) \varphi_\hbar
 \rangle \right| \leq ||Q^{B}_\hbar(e^2- \Lambda^2) \varphi_\hbar  || \to 0\:.
$
Finally, by the triangle inequality we estimate
\begin{align*}
|\langle \varphi_\hbar, Q^{B}_\hbar(e^2) \varphi_\hbar \rangle + \Lambda^2 -2\lambda_\hbar \Lambda|\leq  |\langle \varphi_\hbar, Q^{B}_\hbar(e^2) \varphi_\hbar \rangle - \Lambda^2|+|2\Lambda^2 -2\lambda_\hbar \Lambda|,
\end{align*}
which goes to zero by the previous observations, as $\hbar\to 0$. This establishes (\ref{central2}) ending the proof.
\\\\
(2) For every given $m>0$,
$$\left|\langle \varphi_\hbar, Q_\hbar^B(f) \varphi_\hbar\rangle -f(\sigma_0)\right| = \left| \int_{\bR^{2n}}  |\Phi_\hbar(\sigma)|^2 f(\sigma) \frac{d\sigma}{(2\pi \hbar)^n}-
f(\sigma_0) \right|$$
$$=  \left| \int_{\bR^{2n}}  |\Phi_\hbar(\sigma)|^2 (f(\sigma)  -f(\sigma_0))d\mu_\hbar(\sigma) \right| \leq
\int_{\bR^{2n}}  |\Phi_\hbar(\sigma)|^2 |f(\sigma)  -f(\sigma_0)|d\mu_\hbar(\sigma)  $$
$$\leq \int_{\bR^{2n}\setminus {\cal V}_{1/m}}  |\Phi_\hbar(\sigma)|^2 |f(\sigma)  -f(\sigma_0)|d\mu_\hbar(\sigma) 
+  \int_{ {\cal V}_{1/m}}  |\Phi_\hbar(\sigma)|^2 |f(\sigma)  -f(\sigma_0)|d\mu_\hbar(\sigma) \:.$$
In summary,
$$\left|\langle \varphi_\hbar, Q_\hbar^B(f) \varphi_\hbar\rangle -f(\sigma_0)\right| \leq 2||f||_\infty \int_{\bR^{2n}\setminus {\cal V}_{1/m}}  |\Phi_\hbar(\sigma)|^2d\mu_\hbar(\sigma) 
+\sup_{\sigma \in {\cal V}_{1/m}}|f(\sigma)-f(\sigma_0)|\:.$$
Take $\epsilon>0$.
Since the sets  ${\cal V}_{1/m}$ are a fundamental system of neighbourhoods of $\sigma_0$ and $f$ is continuous, there is $m_\epsilon \in \bN$ such that
 $\sup_{\sigma \in {\cal V}_{1/m_\epsilon}}|f(\sigma)-f(\sigma_0)|< \epsilon/2$. Due to statement (1), we can also find $H_\epsilon>0$ such that 
 $0<\hbar< H_\epsilon$ implies
$$ \int_{\bR^{2n}\setminus {\cal V}_{1/m_\epsilon}}  |\Phi_\hbar(\sigma)|^2 |f(\sigma)  -f(\sigma_0)|d\mu_\hbar(\sigma) < \epsilon/2\:.$$
Summing up, for every $\epsilon>0$, there is  $H_\epsilon$ such that 
$\left|\langle \varphi_\hbar, Q_\hbar^B(f) \varphi_\hbar\rangle -f(\sigma_0)\right| \leq \epsilon$
if  $0<\hbar< H_\epsilon$. This is the thesis we wanted to prove.
\end{proof}

\subsection{Classical limits for Berezin quantization maps}\label{Generalizations}
We prove here the existence of the classical limit for a sequence of eigenvectors corresponding to the operators $Q_\hbar^B(e)$, where $e\in C_0(\mathbb{R}^{2n})$. The function $e$ is a classical observable  related to some Hamiltonian function. Since standard 
Hamiltonians of the form 
$h(q,p) =p^2+V(q)$ are unbounded, they stay outside the domain of the quantization map $Q_\hbar^B : C_0(\bR^{2n}) \to \gB(L^2(\bR^n, dx))$ and we cannot take advantage of the $C^*$-algebra formalism. 
  A possibility is to interpret  \beq e(q,p)=e^{-t(p^2+V(q))} \label{bee}\eeq for  $t>0$. Under mild and physically natural conditions on $V$, 
like  $V \in C^\infty(\bR^n)$ with $V(q) \to +\infty$ for $|q| \to +\infty$, 
the map $e$
belongs to $C_0(\bR^{2n})$ and in particular $Q^B_\hbar(e) \in B_\infty(L^2(\bR^n, dx))$ and also it is a positive operator.
 In this case the spectrum 
$\sigma(Q^B_\hbar(e))$ is a pure point spectrum made of non-negative eigenvalues  with $0$ as the unique, possibly, accumulation point (possibly in the continuous spectrum).

For the sake of generality, in this section we shall not assume a that precise form of the function $e$ as in (\ref{bee}),
we only assume that $e\in C_0(\bR^{2n})$, so that $Q^B_\hbar(e)$ is compact with a point spectrum except for $0$ at most,
 and we focus on a sequence of eigenvalues $\lambda_n$ and corresponding eigenvectors $\psi_n$ of the maps $Q^B_\hbar(e)$
such that $\lambda_\hbar \to \Lambda= e(\sigma_0)$ for $\hbar \to 0^+$. The precise form (\ref{bee}) will be adopted when discussing the interplay with  Schr\"odinger operators. 

Though our results assume the existence of such  sequences of eigenvectors,  as we shall prove in the next section (corollary \ref{corollarylink}),
if   $e$ of the form (\ref{bee}) and under suitable hypotheses on $V$,
the said sequence does exist when referring to the sequence of the maximal eigenvalues of $Q^B_\hbar(e)$. In that case
 $\Lambda = \max_{\sigma \in \bR^{2n}} e(\sigma)$.

\begin{theorem}[Classical limit without symmetry]\label{MT2}  
Consider  $e\in C_0(\bR^{2n})$ and $\Lambda \neq 0$ such that 
\begin{align}
 \Lambda =e(\sigma_0) \ \ \text{for a unique point $\sigma_0 \in \bR^{2n}$}.
\end{align}
Let $\{\varphi_\hbar\}_{\hbar>0}$ be a family of eigenvectors with eigenvalues $\{\lambda_\hbar\}_{\hbar>0}$ of $Q_\hbar^B(e)$  converging to $\Lambda$, as $\hbar\to0$.
With these assumptions,
$$\langle \varphi_\hbar, Q_\hbar(f) \varphi_\hbar \rangle \to f(\sigma_0)\quad \mbox{as $\hbar \to 0^+$, for every $f\in C_0(\bR^{2n})$.}$$

\end{theorem}

\begin{proof}
Taking proposition \ref{localization2} (2) into account,  the only fact to be proved is that  the family $\{{\cal V}_\epsilon\}_{\epsilon>0}=e^{-1}((\Lambda-\epsilon, \Lambda+\epsilon))$ forms a fundamental system of $\sigma_0$. Notice that $\sigma_0$ is the unique point in $\bR^{2n}$ such that $e(\sigma_0) = \max_{\sigma\in \bR^{2n}}|e(\sigma)|$. Now suppose that there is an open ball $B_r$ centred at $\sigma_0$ of radius $r>0$ such that ${\cal V}_{1/m} \not \subset B_{r}$  for $m>m_0$ (i.e., $\{{\cal V}_\epsilon\}_{\epsilon>0}$ is {\em not} a fundamental system of neighbourhoods of $\sigma_0$). As a consequence, for every $m>m_0$ there is $\sigma_m \in \bR^{2n}\setminus B_r$ such that  $\sigma_m \in {\cal V}_{1/m}$, i.e.,  $|e(\sigma_m)- \Lambda| < 1/m$.  
Since $\Lambda=e(\sigma_0)\neq 0$ and $e\in C_0(\bR^{2n})$,  it follows that $\lim_{|\sigma|\to\infty}|\Lambda-e(\sigma)| =|\Lambda|>0$ so that the sequence of $\sigma_m$ must be bounded because $|e(\sigma_m)- \Lambda|\to 0\neq \Lambda$.
In summary, $\{\sigma_m\in\bR^{2n}\setminus B_r \ |\ |\Lambda-e(\sigma_m)|\leq 1/m \}\subset\mathbb{R}^{2n}$ is  contained in a compact set $K$, whenever $m>m_0$.
Since $K\setminus B_r$ is compact as well, we can extract a subsequence  $\sigma_{m_k}' \to \sigma_0'\in K\setminus B_r$.  Since $|e(\sigma_{m_k}) -\Lambda| < 1/m_k \to 0$ for $k\to +\infty$, continuity of $e$ implies  $e(\sigma_0')= \Lambda$.  We observe that $\sigma_0'\neq \sigma_0$ since $\sigma_0'\in  K\setminus B_r$ and $\sigma_0\in B_r$. However, $e(\sigma_0)=\Lambda=e(\sigma_0')$ so that two distinct points reach the same value. This is in contradiction with the fact that the value is attained in a unique point. 
\end{proof}

We now pass  to a more  elaborated case where a symmetry of $e$ is present. We  consider a  group $G$  acting by symplectomorphism
$G \ni g : \mathbb{R}^{2n} \ni (q,p) \mapsto g(q,p) \in \mathbb{R}^{2n}$
 on $(\mathbb{R}^{2n}, \sum_{k=1}^n dp_k \wedge dq^k)$,
and focus on  the two cases:  either $G$ is a compact topological group or (b)  $G$ is a discrete group.\\

\begin{example}\label{physicalexamples}
{\em When assuming $e$ of the form (\ref{bee}),   two typical examples are
\begin{itemize}
\item[(i)]  the {\em double well} system on $\bR^2$   
with $G= \bZ_2 = \{\pm 1\}$, so that $g(q,p)= (\pm q,\pm p)\:;$
\item[(ii)]   the {\em Mexican hat} system on 
$\bR^{2n}$, with $n>1$,
and $SO(n)$  as symmetry group $G$ so that
$g(q,p)= (g q, gp)$,
where $gz$ is the standard rotation of $z\in \bR^n$ according to $g\in SO(n)$. 
\end{itemize}
 In both cases
\beq
h(q,p) = p^2 + (q^2-1)^2\:, 
\eeq
where $(q,p) \in \bR\times \bR$ or $(q,p) \in \bR^n\times \bR^n$ respectively.}
\end{example}

Differently from the simpler result established with theorem \ref{MT2}, we now also assume that the eigenspaces  
associated to the sequence of eigenvalues $\lambda_\hbar$ of $Q_\hbar^B(e)$ have dimension $1$. Again this condition 
is satisfied  if   $e$ is of the form (\ref{bee}) and under suitable hypotheses on $V$ when  referring to the maximal eigenvalues of $Q^B_\hbar(e)$  and with  $\Lambda= \max_{\sigma \in \bR^{2n}} e(\sigma)$, as we shall see in the next section (corollary \ref{corollarylink}).

\begin{theorem}[Classical limit with symmetry]\label{thm:claslimNW2}
Consider a group $G$ either finite or topological compact, $e\in C_0(\mathbb{R}^{2n})$ and  assume the following hypotheses.
\begin{itemize} 
\item[(a)] $G$ acts continuously in the topological-group   case\footnote{The action $G\times \mathbb{R}^{2n} \ni
 (g,\sigma) \mapsto g\sigma \in \mathbb{R}^{2n}$ is continuous.}, 
on  $(\mathbb{R}^{2n}, \sum_{k=1}^n dp_k \wedge dq^k)$ in terms of symplectomorphism. 
\item[(b)]  $e$ is invariant under $G$.
\item[(c)] The action of $G$ is transitive on $e^{-1}(\{\Lambda\})$.
\end{itemize}
  Then the following facts are valid  for every chosen $\sigma_0\in e^{-1}(\{\Lambda\})$ and
a family  of eigenvectors  $\{\varphi_\hbar\}_{\hbar>0}$
of  
$Q_\hbar^B(e)$  with non-degenerate eigenvalues $\{\lambda_\hbar\}_{\hbar>0}$ converging to $\Lambda$ as $\hbar\to 0$
for some $\Lambda\in \text{ran}(e) \setminus\{0\}$. 

\begin{itemize}
\item[(1)] If $G$ is topological and compact and $\mu_{G}$ is the normalized Haar measure,
\begin{align}
\lim_{\hbar\to 0^+}\langle\varphi_\hbar,Q_{\hbar}^B(f)\varphi_\hbar\rangle=\int_{G} f(g\sigma_0)d\mu_G(g),  \quad  \mbox{for every $f\in C_0(\mathbb{R}^{2n})$.}\label{classical limit2}
\end{align}
\item[(2)] If $G$ is finite and $N_G$ is the number of elements of $G$,
\begin{align}
\lim_{\hbar\to 0^+}\langle\varphi_\hbar,Q_{\hbar}^B(f)\varphi_\hbar\rangle=\frac{1}{N_G} \sum_{g\in G} f(g \sigma_0),  \quad  \mbox{for every $f\in C_0(\mathbb{R}^{2n})$.}\label{classical limit2b}
\end{align}
\end{itemize}
The left- and right-hand sides of (\ref{classical limit2}) and (\ref{classical limit2b}) are independent of the choice of $\sigma_0$.
\end{theorem}
\begin{proof}
We assume that $G$ is a compact topological group and prove (1). Compactness guarantees in particular the existence of a unique  normalized two-sided Haar measure 
$\mu_G$.  The case (2)  has the same (actually easier) proof 
just  by everywhere replacing the integral with $\frac{1}{N_G} \sum_{g\in G}$ which is the same as saying that the Haar measure is discrete.\\
To go on, choose $g\in G$ arbitrarily.
Since $U_g$ commutes with $Q_\hbar^B(e)$ and the eigenspace spanned by $\varphi_\hbar$ is one-dimensional, 
$U_g \varphi_\hbar= e^{ia}\varphi_\hbar$ for some real $a$ and thus $\Phi_\hbar(g^{-1}\sigma)= 
pWU_g \varphi_\hbar = e^{ia}\Phi_\hbar(\sigma)$. Since the measure $d\sigma$ is $G$-invariant, the measure 
$|\Phi_\hbar|^2 \frac{d\sigma}{(2\pi\hbar)^{2n}}$ is also $G$-invariant. We can  therefore replace $\sigma$ with $g^{-1}\sigma$ in both $  |\Phi_\hbar(\sigma)|^2$ and $d\sigma$:
$$\langle\varphi_\hbar,Q_{\hbar}^B(f)\varphi_\hbar\rangle=  \int_{\bR^{2n}}  |\Phi_\hbar(\sigma)|^2 f(\sigma) \frac{d\sigma}{(2\pi \hbar)^n}=
\int_{\bR^{2n}}  |\Phi_\hbar(g^{-1}\sigma)|^2 f(\sigma) \frac{dg^{-1}\sigma}{(2\pi \hbar)^n}\:.$$
Replacing again $\sigma$ with $g\sigma$ in the complete integrand,
$$\langle\varphi_\hbar,Q_{\hbar}^B(f)\varphi_\hbar\rangle=  \int_{\bR^{2n}}  |\Phi_\hbar(\sigma)|^2 f(g\sigma) \frac{d\sigma}{(2\pi \hbar)^n}\:.$$ Since the left-hand side is independent of $g$, we can integrate both sides with respect to the normalized 
Haar measure on $G$, obtaining
$$\langle\varphi_\hbar,Q_{\hbar}^B(f)\varphi_\hbar\rangle
=  \int_G\int_{\bR^{2n}}  |\Phi_\hbar(\sigma)|^2 f(g\sigma) \frac{d\sigma}{(2\pi \hbar)^n} d\mu_G(g)=\int_{\bR^{2n}}  |\Phi_\hbar(\sigma)|^2 \int_G f(g\sigma) d\mu_G(g) \frac{d\sigma}{(2\pi \hbar)^n}\:. $$
Above, in order to interchange the two integrals, we have also used the fact that $\bR^{2n}\times G \ni (\sigma,g) \mapsto  |\Phi_\hbar(\sigma)|^2|f(g\sigma)|$ is integrable in the product measure  since
$$ \int_G\int_{\bR^{2n}}  |\Phi_\hbar(\sigma)|^2 |f(g\sigma)| \frac{d\sigma}{(2\pi \hbar)^n} d\mu_G(g)
\leq 
||f||_\infty \int_G\int_{\bR^{2n}}  |\Phi_\hbar(\sigma)|^2 \frac{d\sigma}{(2\pi \hbar)^n} d\mu_G(g) 
$$ $$=
||f||_\infty \int_{\bR^{2n}} |\Phi_\hbar(\sigma)|^2 \frac{d\sigma}{(2\pi \hbar)^n} \int_G1 d\mu_G(g)
=  ||f||_\infty \int_{\bR^{2n}} |\Phi_\hbar(\sigma)|^2 \frac{d\sigma}{(2\pi \hbar)^n}  =1\:, $$
and then we have used the Fubini-Tonelli theorem.
In summary,
$$\langle\varphi_\hbar,Q_{\hbar}^B(f)\varphi_\hbar\rangle= \int_{\bR^{2n}}  |\Phi_\hbar(\sigma)|^2F(\sigma) \frac{d\sigma}{(2\pi \hbar)^n},\: \quad \mbox{where}$$
\begin{align}
 F(\sigma) := \int_G f(g\sigma) d\mu_G(g), \quad (\sigma\in \mathbb{R}^{2n}). \label{Ff}
\end{align}
Notice that this function is  (i) bounded (since $f$  is bounded and $\mu_G$ finite), (ii) continuous (as it immediately arises form Lebesgue's dominated convergence since, again, $f$  is bounded and $\mu_G$ finite),
(iii) constant on $e^{-1}(\Lambda)$, since $F(g'\sigma) = F(\sigma)$ because $\mu_G$ is $G$-invariant and  the action of $G$ on $e^{-1}(\{\Lambda\})$ is transitive.
To go on, 
since  $\int_{\bR^{2n}}  |\Phi_\hbar(\sigma)|^2\frac{d\sigma}{(2\pi \hbar)^n}=1$,
we can write, 
$\left|\langle\varphi_\hbar,Q_{\hbar}^B(f)\varphi_\hbar\rangle - F(\sigma_0) \right| = 
\left| \int_{\bR^{2n}}  |\Phi_\hbar(\sigma)|^2(F(\sigma)  - F(\sigma_0)) \frac{d\sigma}{(2\pi \hbar)^n} \right|$ for any arbitrarily taken $\sigma_0 \in e^{-1}(\Lambda)$,
Defining ${\cal V}_\delta := e^{-1}(\Lambda-\delta, \Lambda+\delta)$ and for every given  $m\in \bN \setminus\{0\}$, we can now estimate
$$\left|\langle\varphi_\hbar,Q_{\hbar}^B(f)\varphi_\hbar\rangle - F(\sigma_0) \right|$$ 
$$\leq  \int_{{\cal V}_{1/m}}  |\Phi_\hbar(\sigma)|^2|F(\sigma)  - F(\sigma_0)|\frac{d\sigma}{(2\pi \hbar)^n} 
+  \int_{\bR^{2n}\setminus {\cal V}_{1/m}}  |\Phi_\hbar(\sigma)|^2|F(\sigma)  - F(\sigma_0)|\frac{d\sigma}{(2\pi \hbar)^n}$$
$$\leq \sup_{\sigma \in {\cal V}_{1/m}}|F(\sigma)  - F(\sigma_0)| \int_{{\cal V}_{1/m}}  |\Phi_\hbar(\sigma)|^2\frac{d\sigma}{(2\pi \hbar)^n} 
+ 2||F||_\infty \int_{\bR^{2n}\setminus {\cal V}_{1/m}}  |\Phi_\hbar(\sigma)|^2\frac{d\sigma}{(2\pi \hbar)^n} $$
\beq \leq \sup_{\sigma \in {\cal V}_{1/m}}|F(\sigma)  - F(\sigma_0)| + 2||F||_\infty \int_{\bR^{2n}\setminus {\cal V}_{1/m}}
  |\Phi_\hbar(\sigma)|^2\frac{d\sigma}{(2\pi \hbar)^n}\:.\label{lastline} \eeq
Let us now  focus on the two terms in the last line of (\ref{lastline}) separately.
 The following lemma is true.
\begin{lemma}\label{fclaim} Under the hypotheses of theorem \ref{thm:claslimNW2} and with $F$ defined in (\ref{Ff}),
 for every $\epsilon>0$, there is  $m_\epsilon\in \bN $ such that $\sup_{\sigma \in {\cal V}_{1/m_\epsilon}}|F(\sigma)  - F(\sigma_0)| < \epsilon/2\:.$
\end{lemma}
\begin{proof}
See Appendix \ref{appA}.
\end{proof}
\noindent 
Keeping that $m_\epsilon$ and exploiting (1) in proposition \ref{localization2}, we can find $H_\epsilon>0$ such that
\begin{align}
2||F||_\infty \int_{\bR^{2n}\setminus {\cal V}_{1/m_\epsilon}}  |\Phi_\hbar(\sigma)|^2\frac{d\sigma}{(2\pi \hbar)^n} < \epsilon/2
\nonumber,
\end{align}
for $0<\hbar< H_\epsilon$. Looking again at the last line of  (\ref{lastline}), we conclude that  for every $\epsilon>0$, there is $H_\epsilon$ such that  $0<\hbar< H_\epsilon$ implies
$\left|\langle\varphi_\hbar,Q_{\hbar}^B(f)\varphi_\hbar\rangle - F(\sigma_0) \right| < \epsilon$,
concluding the proof.
\end{proof}

%

\section{Schr\"{o}dinger operators versus Berezin quantization maps}\label{MainB}
In this section we discuss the interplay between Schr\"{o}dinger operators and Berezin quantization maps $Q^B_\hbar(e^{-th})$. We will see that under some assumptions on the potential appearing both  in the Schr\"{o}dinger operator and in $h=p^2+V$, the operators are related. We shall provide several theorems making this relation precise.
\subsection{General setting}
We consider $\hbar$-dependent (unbounded) Schr\"{o}dinger operators $H_{\hbar}$ defined on some dense domain of $\mathcal{H}=L^2(\mathbb{R}^n, dx)$. Such operators are typically given by
\begin{align}
H_{\hbar}:=\overline{-\hbar^2\Delta + V}, \quad \hbar>0\:,\label{Schroper}
\end{align}
where $\Delta$ denotes the Laplacian on $\mathbb{R}^{n}$, and V denotes multiplication by some real-valued function on $\mathbb{R}^n$, playing the role of the potential. 
One should impose conditions on the potential $V$ in order to make $H_{\hbar}$  self-adjoint when $-\hbar^2\Delta + V$ is initially defined on  $C_c^\infty(\bR^n)$.
Our general hypotheses will be the following ones.
\begin{itemize}\label{HypotV}
\item[{\bf (V1)}] $V$ is a  real-valued $C^\infty(\bR^n)$ function.
\item[{\bf (V2)}] $ V(x) \to +\infty $ for $|x| \to +\infty$ (i.e., for every $M>0$, there is $R_M>0$
 such that $V(x) > M$ if $|x|>R_M$).
\item[{\bf (V3)}] $e^{-tV} \in {\cal S}(\bR^n)$ for $t>0$.
\end{itemize}

An elementary example consists of a real polynomial satisfying (V2). It  also satisfies (V1) and (V3) automatically.

It is known that when 
 $V \in L^2_{loc}(\bR^n, dx)$ and $V(x) \geq 0$ pointwise then the operator  $-\hbar^2\Delta + V$
is essentially-self adjoint
 on $C_0^\infty(\bR^n)$ in view of Theorem X.28 \cite{RS2}
and that $H_\hbar\geq 0$.  
Referring to our general hypotheses, we observe that evidently (V1) implies $V \in L^2_{loc}(\bR^n, dx)$. 
Furthermore, (V1) and (V2) immediately  imply that
 \beq \inf_{x\in \bR^n} V = \min_{x\in \bR^n}V(x) = c >-\infty\:. \label{defC}\eeq
Since $(H_\hbar+ aI)^*= H_\hbar^*+ aI$, the domain being the one of $\overline{H_\hbar}= H_\hbar^*$ for both sides, the essential selfadjointness property still holds
 when relaxing $V\geq 0$ to (V2), also obtaining $H_\hbar \geq cI$.  Hence (V1)-(V2) imply that $H_\hbar$ is selfadjoint and \beq H_\hbar \geq (\min V) I\:. \label{bb}\eeq

\subsection{Comparison of $e^{-tH_\hbar}$ and $Q_\hbar^B(e^{-th})$}
The physical idea is now that, in some sense $Q_\hbar^B(p^2+V)$ is a good approximation of $H_\hbar$ as $\hbar \to 0^+$. We wish to implement this idea in the framework of $C^*$-algebras
of operators. Since the above operators are unbounded an appealing idea is to consider bounded functions of $p^2+V$ and $H_\hbar$ in place of themselves.  An  apparent promising 
 idea seems to be the use  the {\em resolvent algebra} of operators \cite{BuchGrun}. However this approach reveals to be awkward, essentially because the classical counterpart of the $C^*$-algebra of resolvents is too involved to be profitably  handled.
A different approach, already adoptend in the previous section,  relies on the use of (real) exponential functions and the associated quantum counterpart consisting of contraction semigroups. This idea indeed works
and it is based upon the following list of comparison results.

\begin{proposition}\label{IMPORTANT}
Let $V$ satisfy (V1)-(V3)  and consider   the associated  family  of  operators in $\gB(L^2(\bR^n, dx))$  indexed by $\hbar>0$ and defined by $e^{-tH_\hbar}$  according to (\ref{Schroper}), where $t>0$ is given. It holds
\beq Q^{B}_\hbar(e^{-t(p^2+ V)}) - e^{-t H_\hbar} \to 0 \quad \mbox{in strong sense for $\hbar \to 0^+$}\:.\label{1lim}\eeq
More precisely both 
\beq e^{-t H_\hbar} \quad \mbox{and}\quad  Q^{B}_\hbar(e^{-t(p^2+ V)}) \to e^{-t V}  \quad \mbox{in strong sense for $\hbar \to 0^+$}\:.\label{2lim}\eeq
\end{proposition}

\begin{proof} See Appendix \ref{appA}.
\end{proof}

We now pass to a comparison of the ground state of the Hamiltonian operator $H_\hbar$ and the corresponding of 
$Q_\hbar^B(e^{-t (p^2+V)})$. More precisely, we compare the maximal eigenvalues of 
$e^{-tH_\hbar}$ and $Q_\hbar^B(e^{-t (p^2+V)})$ for $t>0$ given.  Assuming (V1)-(V3), $V \in  L^1_{loc}(\bR^n, dx)$ is bounded from below and $\quad V(x) \to +\infty$ for $|x| \to +\infty$ and thus the resolvent of $H_\hbar$ is compact 
due to Theorem XIII.67 in \cite{RS4}.
According to standard results on positive compact operators (see, e.g., \cite{Mor}), if $\hbar>0$,
\begin{itemize}
\item[(a)]  the spectrum of $H_\hbar$ is a pure point spectrum and there is a corresponding Hilbert basis of eigenvectors $\{\psi^{(j)}_\hbar\}_{j=0,1,\ldots}$
with corresponding eigenvalues
\beq \sigma(H_\hbar) = \{E^{(j)}_\hbar\}_{j=0,1,2,\ldots}\quad \mbox{with $0 \leq E^{(j)}_\hbar \leq E^{(j+1)}_\hbar \to +\infty$ as $j\to +\infty$,}\label{specH}\eeq
where  every eigenspace has finite dimension;
\item[(b)]  $e^{-tH_\hbar}$ 
is compact with  spectrum 
$\sigma(e^{-tH_\hbar}) = \{0\}\cup  \{e^{-tE^{(j)}_\hbar}\}_{j=0,1,2, \ldots}\:,$
 $0$ being the unique point of the continuous spectrum, and the  eigenspaces of $H_\hbar$ and $e^{-tH_\hbar}$ coincide;  
\item[(c)]  the minimal eigenvalue $E^{(0)}_\hbar$ of $H_\hbar$ corresponds to the maximal eigenvalue of  $e^{-tH_\hbar}$
according to
\beq
e^{-t E_\hbar^{(0)}} = ||e^{-tH_\hbar} ||\:.
\eeq
\end{itemize}
Under the said hypotheses on $V$,  we also  know that $Q^{B}_\hbar(e^{-t(p^2+ V)})$ is a positive compact operator due to theorems \ref{teoQB1}  and
 \ref{BerezinS} when $\hbar>0$.
These facts permit us to focus on interplay  of the maximal eigenvalues of $e^{-tH_\hbar}$ and $Q^{B}_\hbar(e^{-t(p^2+ V)})$ with the following proposition.

\begin{proposition}\label{IMPORTANT2}
Assuming the hypotheses (V1)-(V3), given $t>0$,  let  $\lambda^{(0)}_{\hbar}$ be  the maximal eigenvalue of $Q^{B}_\hbar(e^{-t(p^2+ V)})$.
The following facts are true referring to $H_\hbar$ ($\hbar>0$) as in proposition \ref{IMPORTANT} with eigenvalues as in (\ref{specH}).
\begin{itemize}
\item[(1)]  Both  $\mbox{$\lambda_\hbar^{(0)}$ and $e^{-tE^{(0)}_\hbar} \to e^{-t\min V}$
as $\hbar \to 0^+$.}$
\item[(2)]  Both    $\mbox{$||Q^{B}_\hbar(e^{-t(p^2+ V)})||$ and $||e^{-t H_\hbar}|| \to e^{-t \min V}$
as $\hbar \to 0^+$.}$
\end{itemize}
\end{proposition}

\begin{proof} See Appendix \ref{appA}.
\end{proof}

\begin{corollary}\label{REMEn} Under the hypotheses (V1)-(V3), the minimal eigenvalues of 
$H_\hbar := \overline{-\hbar \Delta + V}$ satisfy 
$$E^{(0)}_\hbar \to \min_{x\in \bR^n} V(x) = \min_{(q,p) \in \bR^{2n}} p^2 + V(q)\quad \mbox{for} \quad  \hbar \to 0^+\:.$$
\end{corollary}

Let us finally pass to prove that the dimension of the eigenspaces of the maximal eigenvalues of respectively $Q_\hbar^B(e^{-t(p^2 +V)})$ and
$e^{-tH_\hbar}$ is $1$ in both cases.

\begin{proposition}\label{propDIM} Under the hypotheses (V1)-(V3) and for $t>0$, the eigenspaces of $Q_\hbar^B(e^{-t(p^2 +V)})$ and
$H_\hbar$ associated to the respective maximal  and minimal  eigenvalues $\lambda^{(0)}_\hbar$ and 
$E^{(0)}_\hbar$  have both dimension $1$.
\end{proposition}

\begin{proof} See Appendix \ref{appA}.
\end{proof}
\noindent In summary, if (V1)-(V3) are valid, then   the eigenspaces of maximal eigenvalues of $Q_\hbar^B(e^{-t(p^2 +V)})$ and $e^{-tH_\hbar}$
are spanned respectively by corresponding unit eigenvectors $\varphi_\hbar^{(0)}$ and $\psi_\hbar^{(0)}$  defined up to phases. 
The vector $\psi_\hbar^{(0)}$ also defines the {\em ground state} of $H_\hbar$. 
It would nice to to study  how these eigenvectors are related. 
 We expect that they should coincide as algebraic states when acting on the observables $Q_\hbar^B(f)$ in the limit of small $\hbar$. 
We postpone this discussion to corollary \ref{corcomp}, where we shall prove that it is in fact the case for suitable choices of $V$.

We finally have an important corollary regarding the validity of some hypotheses assumed in stating theorem  \ref{MT2}  and theorem \ref{thm:claslimNW2}.

\begin{corollary}\label{corollarylink} If $e \in C_0(\bR^{2n})$ takes the form (\ref{bee}) and (V1)-(V3) are satisfied for $V$, the following facts are valid.
\begin{itemize}
\item[(1)] If $\min e$ is achieved in a unique point $\sigma_0=(0,q_0) \in \bR^{2n}$, then all the  hypotheses of Thm \ref{MT2} 
are valid  with $\Lambda= e^{-tV(q_0)}$, $\lambda_\hbar := \lambda_\hbar^{(0)}$,   $\varphi_\hbar = \varphi_\hbar^{(0)}$.
\item[(2)] If $\Lambda= e^{-t \min V}$, and the   hypotheses (a)-(c) of  Thm \ref{thm:claslimNW2} hold, then the theorem 
is valid  with $\Lambda= e^{-t \min V}$, $\lambda_\hbar := \lambda_\hbar^{(0)}$, $\varphi_\hbar = \varphi_\hbar^{(0)}$.
\end{itemize}
In both cases $\varphi_\hbar = \varphi_\hbar^{(0)}$ is 
the unique unit eigenvector (up to phases) corresponding to eigenvalue $\lambda_\hbar^{(0)}$.
\end{corollary}

\section{Semiclassical properties of Schr\"{o}dinger operators}\label{MainS}
In this section we prove the existence of the classical limit of a sequence of eigenvectors of minimal eigenvalues of   corresponding to Schr\"{o}dinger operators $H_\hbar$. Similar as in Section \ref{SemiclassicalpropertiesofBerezinquantizationmaps}, we first prove a localization result (Thm. \ref{localizationSCH}) of such sequences, followed by two main theorems (Thm. \ref{MT12} and Thm. \ref{thm:claslimNEWSchr}) where again distinction is made between the presence of a symmetry or not. We point out to the reader that localization of eigenvectors for Schr\"{o}dinger operators has been extensively studied and often yields exponential decay \cite{Hel,SJHE,Sim85}.

\subsection{Localization of eigenvectors}
As Schr\"{o}dinger operators are unbounded, proposition \ref{localization2} cannot directly be applied to eigenvectors of such operators. We hereto prove a similar result yielding localization of ground state eigenvectors $\{\psi_{\hbar}^{(0)}\}_{\hbar}$, with minimal eigenvalue $E_\hbar^{(0)}$, of the Schr\"{o}dinger operators $H_\hbar=\overline{-\hbar\Delta + V}$ on $L^2(\bR^n, dx)$ where $V$ satisfies conditions (V1)-(V3).  

For $h(p,q):=p^2+V(q)$ and some given $t>0$, let us focus on  the preimage $e^{-1}(\{\max e^{-th}\})$, where  $e(p,q):=e^{-th(p,q)}$. 
We stress that this set generally contains more than one point and also it coincides with the preimage $h^{-1}(\{\min h\}) = \{p=0\}\times V^{-1}(\{\min V\})$.

If 
 $e_1(p):=e^{-tp^2}$ and $e_2(q):=e^{-tV(q)}$,  consider a class of open neighborhoods ${\cal U}_\epsilon$ of $e^{-1}(\{\max e^{-th}\})$ 
\begin{align}
{\cal U}_\epsilon:={\cal U}_\epsilon^1\times {\cal U}_\epsilon^2, \label{NBH}
\end{align}
$${\cal U}_\epsilon^1:=e_1^{-1}((1-\epsilon,1+\epsilon))\quad \mbox{and}\quad {\cal U}_\epsilon^2=e_2^{-1}((e^{-t\min V}-\epsilon,e^{-t\min V}+\epsilon))\:, \quad \epsilon> 0\:.$$
If we assume $\epsilon>0$ is sufficiently small, it easy to prove that there is  $C>0$, independent of $\epsilon$, such that
\beq
\sigma \in {\cal U}_\epsilon \quad \mbox{implies} \quad |e^{-th(\sigma)}- \max e^{-th}| < C \epsilon \label{Ct}\:.
\eeq

This yields the following proposition.

\begin{proposition}\label{localizationSCH}  Let  $H_\hbar$ be as in \eqref{HypotV}, where $V$ satisfies (V1)-(V3).
Let $\{\psi_{\hbar}^{(0)}\}_{\hbar}$ be a sequence of eigenvectors of a Schr\"{o}dinger operators $\{H_\hbar\}_\hbar$ with minimal  eigenvalues $\{E_{\hbar}^{(0)}\}_{\hbar}$ such that, according to corollary \ref{REMEn}
\begin{align}
E_\hbar^{(0)}\to \min V = \min h, \  \text{for $\hbar\to 0^+$.}
\end{align}
\begin{itemize}
\item[(1)]  If $\Psi_\hbar^{(0)}:= W\psi_\hbar^{(0)}$ and  the open neighborhood ${\cal U}_\epsilon$ of $\min h$ is defined as in \eqref{NBH} for every (sufficiently small) $\epsilon >0$, then
\begin{align}
||\Psi_{\hbar}^{(0)}||_{L^2(\mathbb{R}^{2n}\setminus {\cal U}_\epsilon, \frac{d^npd^nq}{(2\pi\hbar)^n})}\to 0, \quad \mbox{for}\quad \hbar \to 0^+ \label{limt}\:.
\end{align}
\item[(2)] If $V^{-1}(\{\min V\}) = \{q_0\} \in \bR^{2n}$ and the family of sets $\{{\cal U}_\epsilon\}_{\epsilon>0}$ is a fundamental system of neighborhoods of $\sigma_0:= (0,q_0)$,
then
$$\langle \varphi_\hbar, Q_\hbar(f) \varphi_\hbar \rangle \to f(\sigma_0)\quad \mbox{as $\hbar \to 0^+$  for every $f\in C_0(\bR^{2n})$.}$$
\end{itemize}
\end{proposition}

\begin{proof} (1)  Take $t>0$.
As already done in previous proofs, without loss of generation we assume that that $\text{min}_{x\in\mathbb{R}^n}V(x)=0$. 
Since  $E_\hbar^{(0)}\to 0$, clearly  $e^{-tE_\hbar^{(0)}}\to 1$, as $\hbar\to 0$ (and $t>0$).  A suitable use of Jensen's inequality for probability measures leads to the following result.

\begin{lemma}\label{lemmaSCH1} Under the hypotheses of theorem \ref{localizationSCH} with $\min V=0$,
\begin{align}
\lim_{\hbar\to 0}\langle\psi_\hbar^{(0)},e^{-tV}\psi_\hbar^{(0)}\rangle=1 \quad \mbox{and} \quad  \lim_{\hbar\to 0}\langle\psi_\hbar^{(0)},e^{t\hbar^2\Delta}\psi_\hbar^{(0)}\rangle=1\:. \label{convJensen}
\end{align}
\end{lemma}

\begin{proof}  See Appendix \ref{appA}.
\end{proof}
The next step is to show that
 \beq \langle\psi_\hbar^{(0)}, Q_\hbar^B((e^{-tV}-1)^2)\psi_\hbar^{(0)}\rangle\to 0\quad \mbox{if $\hbar\to 0$}, \label{limQ}\eeq 
using the former in \eqref{convJensen}. We proceed as in the proof of Proposition \ref{localization2}. Hereto, we first observe
\begin{align}
\langle\psi_\hbar^{(0)}, Q_\hbar^B((e^{-tV}-1)^2)\psi_\hbar^{(0)}\rangle=\langle\psi_\hbar^{(0)}, Q_\hbar^B(e^{-2tV})\psi_\hbar^{(0)}\rangle+1-2\langle\psi_\hbar^{(0)}, Q_\hbar^B(e^{-tV})\psi_\hbar^{(0)}\rangle. \label{decQV}
\end{align}
Proposition \ref{prooff} particularly yields,   $||Q_\hbar^B(e^{-tV}) - e^{-tV}|| \to 0$  
as $\hbar \to 0^+$ in $\gB(L^2(\bR^n, dx))$ for every chosen $t>0$. This together with the former in \eqref{convJensen} and $||\psi_\hbar||=1$, they imply that $\langle\psi_\hbar, Q_\hbar^B(e^{-2tV})\psi_\hbar\rangle$ and $\langle\psi_\hbar, Q_\hbar^B(e^{-tV})\psi_\hbar^{(0)}\rangle$ both converge to $1$ as $\hbar\to 0$. Directly from (\ref{decQV}), we  conclude that (\ref{limQ}) holds.
With a strictly analogous procedure, exploiting the second identity in 
(\ref{convJensen}) and
using again  Proposition \ref{prooff}, for every given $t>0$,  $||Q_\hbar^B(e^{-tp^2}) - e^{t\hbar^2\Delta}|| \to 0$  
as $\hbar \to 0^+$ in $\gB(L^2(\bR^n, dx))$  one also obtains,
\begin{align}
\lim_{\hbar\to 0}\langle\psi_\hbar^{(0)}, Q_\hbar^B((e^{t\hbar^2p^2}-1)^2)\psi_\hbar^{(0)}\rangle=0. \label{conD}
\end{align}
To conclude, we have
\begin{align*}
&||\Psi_\hbar^{(0)}||_{\mathbb{R}^{2n}\setminus (\bR^n \times  \mathcal{U}^2_\epsilon)}^2=\int_{\mathbb{R}^{2n}\setminus (\bR^n \times \mathcal{U}^2_\epsilon)} \sp\sp\sp \sp\sp\sp|\Psi_\hbar^{(0)}(x)|^2d\mu_\hbar \leq
\frac{1}{\epsilon^2}\int_{\mathbb{R}^{2n}\setminus (\bR^n \times \mathcal{U}^2_\epsilon)} \sp\sp\sp \sp\sp\sp (1-e^{-V(q)})^2|\Psi_\hbar^{(0)}(x)|^2d\mu_\hbar\\&\leq
\frac{1}{\epsilon^2}\int_{\mathbb{R}^{2n}} (1-e^{-V(q)})^2|\Psi_\hbar(x)|^2d\mu_\hbar   =
\frac{1}{\epsilon^2} \langle \psi_\hbar^{(0)}, Q^B_\hbar((e^{-tV}-1)^2) \psi_\hbar^{(0)} \rangle \to 0\:,
\end{align*}
for $\hbar \to 0^+$.
A similar procedure relying on (\ref{conD}) proves that
$
||\Psi_\hbar^{(0)}||_{\mathbb{R}^{2n}\setminus (\mathcal{U}^1_\epsilon\times  \bR^n)}^2 \to 0
$ if $\hbar \to 0^+$.
Just taking the union of the two considered sets 
$$\mathbb{R}^{2n}\setminus ({\cal U}^1_\epsilon\times  \bR^n) \bigcup \mathbb{R}^{2n}\setminus (\bR^n \times {\cal U}^2_\epsilon)= \mathbb{R}^{2n}\setminus {\cal U}_\epsilon,$$
and using the established properties, we eventually arrive at (\ref{limt}).\\
(2) The proof is essentially identical to that of (2) in proposition \ref{localization2}, just exploiting part (1) and everywhere 
replacing $\varphi_\hbar$ for $\psi^{(0)}_\hbar$, taking in particular $\Phi_\hbar = W\psi^{(0)}_\hbar$.
\end{proof}

\subsection{Classical limit of ground states of  Schr\"{o}dinger operators}
We now prove that the classical limit of ground states of Schr\"{o}dinger operators exists. We start with the simple case when no symmetry is present in the potential.
\begin{theorem}[Classical limit without symmetry]\label{MT12}  Let  $H_\hbar$ be as in \eqref{HypotV}, where $V$ satisfies (V1)-(V3) and such that $\min_{q\in \bR^n} V(q)= V(q_0)$ for a unique point $q_0 \in \bR^n$ and define  $\sigma_0:=(0,q_0)$.

If $\{\psi^{(0)}_\hbar\}_{\hbar>0}$ is a family of eigenvectors with minimal eigenvalues $\{E^{(0)}_\hbar\}_{\hbar>0}$ of $H_\hbar$, then
$$\langle \psi^{(0)}_\hbar, Q_\hbar(f) \psi^{(0)}_\hbar \rangle \to f(\sigma_0)\quad \mbox{as $\hbar \to 0^+$, for every $f\in C_0(\bR^{2n})$.}$$
\end{theorem}

\begin{proof} The proof is obtained by straightforwardly rephrasing the proof of  theorem \ref{MT2}, by  replacing $\varphi_\hbar$ for $\psi^{(0)}_\hbar$, taking in particular $\Phi_\hbar := W\psi^{(0)}_\hbar$,   and by exploiting the localization properties established in proposition 
 \ref{localizationSCH},  and using   (\ref{Ct}) in particular.
 To fulfil the hypotheses of   theorem \ref{MT2}, observe that 
the eigenvalues $\lambda_\hbar := e ^{-tE^{(0)}_\hbar}$ converge to $\Lambda := e^{-t \min V} = \max e^{-t h}\neq 0$ and this value is reached at a unique point $\sigma_0$ as in theorem \ref{MT2}.
\end{proof}

Note that this theorem in particular proves the classical limit of the simple quantum harmonic oscillator. 

Similarly as in the previous section we now focus on more complex systems and consider the case when a symmetry is present.
Let us again take a group $G$ (a compact topological group or a discrete group) acting by symplectomorphism on $(\mathbb{R}^{2n}, \sum_{k=1}^n dp_k \wedge dq^k)$,
$G \ni g : \mathbb{R}^{2n} \ni (q,p) \mapsto g(q,p) \in \mathbb{R}^{2n}$.
We know from proposition \ref{equivariance2}, that $G$ admits a unitary representation $G\ni g \mpasto U_g \in 
\gB(L^2(\bR^n,dx))$
which acts equivariantly (\ref{equivariantofQ2}) on the quantization map $Q_\hbar^B$. 

Exploiting these notions we are now in the position to prove our main result concerning the classical limit of a sequence of ground-state eigenvectors of $
H_\hbar = \overline{\hbar^2 \Delta + V}$. In a complete analogous way as in the previous section we consider the strict deformation quantization of the Poisson manifold  $(\bR^{2n}, \sum_{k=1}^n dp_k \wedge dq^k)$ (in the sense of Definition \ref{def:deformationq}) associated with Berezin quantization maps $Q_\hbar^B$ acting either on the  space $C_0(\mathbb{R}^{2n})$ or $C_c^\infty(\bR^{2n})$, and prove the existence of the classical 
limit with respect to the observables $Q_\hbar^B(f)$. The following theorem contains the precise statement.

\begin{theorem}[Classical limit with symmetry]\label{thm:claslimNEWSchr}
Consider a group $G$ either finite or topological compact,  a selfadjoint Schr\"odinger operator on $L^2(\bR^n, dx)$
 $H_\hbar:=\overline{ -\hbar^2 \Delta + V},$
as in (\ref{Schroper}) where  $V: \bR^n \to \bR$ satisfies  (V1)-(V3),
and  assume the following hypotheses.
\begin{itemize} 
\item[(a)] $G$ acts, continuously in the topological-group case\footnote{The action $G\times \mathbb{R}^{2n} \ni
 (g,\sigma) \mapsto g\sigma \in \mathbb{R}^{2n}$ is continuous.}, 
on  $(\mathbb{R}^{2n}, \sum_{k=1}^n dp_k \wedge dq^k)$ in terms of symplectomorphism. 
\item[(b)] Defining $h(q,p) :=p^2 + V(q)$, the action of $G$ is leaves invariant $h^{-1}(\{\min h\})$ and is  transitive on it. 
\item[(c)] The unitary representation $U$ of $G$ on $L^2(\bR^n, dx)$
constructed according to proposition \ref{equivariance2} (whose action is 
equivariant on $Q_\hbar^B$) leaves $H_\hbar$ invariant as in (\ref{invH}) for every given $\hbar>0$.
\end{itemize}
  Then the following facts are valid  for every chosen $\sigma_0\in h^{-1}(\{\min h\})$ and for 
a family\footnote{This family always exists as seen in corollary \ref{REMEn}.}  $\{\psi^{(0)}_\hbar\}_{\hbar>0}$ of (normalized) eigenvectors of
$H_\hbar$
   with (non-degenerate) minimal eigenvalues $\{E^{(0)}_\hbar\}_{\hbar>0}$ converging to $\min_{q\in bR^n} V(q)= \min_{(q,p)\in \bR^{2n}} h(q,p)$ as $\hbar\to 0$. 
\begin{itemize}
\item[(1)] If $G$ is topological and compact,
\begin{align}
\lim_{\hbar\to 0^+}\langle\psi^{(0)}_\hbar,Q_{\hbar}^B(f)\psi^{(0)}_\hbar\rangle=\int_{G} f(g\sigma_0)d\mu_G(g),  \quad  \mbox{for every $f\in C_0(\mathbb{R}^{2n})$;}\label{classical limit22}
\end{align}
where $\mu_{G}$ is the normalized Haar measure of $G$. 
\item[(2)] If $G$ is finite,
\begin{align}
\lim_{\hbar\to 0^+}\langle\psi^{(0)}_\hbar,Q_{\hbar}^B(f)\psi^{(0)}_\hbar\rangle=\frac{1}{N_G} \sum_{g\in G} f(g \sigma_0),  \quad  \mbox{for every $f\in C_0(\mathbb{R}^{2n})$;}\label{classical limit2b2}
\end{align}
where $N_G$ is the number of elements of $G$.
\end{itemize}
The left- and right-hand sides of (\ref{classical limit22}) and (\ref{classical limit2b2}) are independent of the choice of $\sigma_0$.
\end{theorem}

\begin{proof}  We only  consider the case of $G$ topological, since the finite case is similar and easier 
as already discussed in the proof of theorem \ref{thm:claslimNW2}. 
Let us define $\Psi_\hbar(\sigma) = W\psi_\hbar^{(0)}$. Since $U_g$ commutes with $H_\hbar$  for (b) and 
the eigenspace of $\psi_\hbar^{(0)}$ has dimension $1$, it must be $U_g\psi_\hbar^{(0)}= e^{ia_g} \psi_\hbar^{(0)}$ for some real 
$a_g$. (\ref{equivariantofQ23}) implies that $\Psi_\hbar(g^{-1}\sigma) =  e^{ia_g} \Psi_\hbar(\sigma)$ so that, using also (a),  the probability  measure 
$|\Psi_\hbar(\sigma)|^2 \frac{d\sigma}{(2\pi \hbar)^n}$ turns out to be $G$-invariant.  This result permits out to straightforwardly follow the proof  of theorem \ref{thm:claslimNW2}, finding
$$\langle\psi^{(0)}_\hbar,Q_{\hbar}^B(f)\psi^{(0)}_\hbar\rangle= \int_{\bR^{2n}}  |\Psi_\hbar(\sigma)|^2F(\sigma) \frac{d\sigma}{(2\pi \hbar)^n},$$
where $F$ was defined in (\ref{Ff}). Exactly as in the proof of theorem \ref{thm:claslimNW2}, 
$F$ turns out to be (i) bounded, (ii) continuous and (iii) constant 
on $h^{-1}(\{\min h\})$.
Going on as in the proof of theorem \ref{thm:claslimNW2}, if  $\sigma_0 \in h^{-1}(\{\min h\})$ is any point, we end up with the estimate
$$\left|\langle\varphi_\hbar,Q_{\hbar}^B(f)\varphi_\hbar\rangle - F(\sigma_0) \right| = 
\left| \int_{\bR^{2n}}  |\Phi_\hbar(\sigma)|^2(F(\sigma)  - F(\sigma_0)) \frac{d\sigma}{(2\pi \hbar)^n} \right|\:.$$
Rephrasing the proof of theorem \ref{thm:claslimNW2}, with ${\cal U}_\delta$  now defined as in (\ref{NBH}) with $e= e^{-th}= e^{-tp^2}e^{-tV}$ for some $t>0$ (notice that
$e^{-1}(\{\max e\})= h^{-1}(\{\min h\})$), we find
\beq\left|\langle\psi^{(0)}_\hbar,Q_{\hbar}^B(f)\psi^{(0)}_\hbar\rangle - F(\sigma_0) \right|
 \leq \sup_{\sigma \in {\cal U}_{1/m}}|F(\sigma)  - F(\sigma_0)| + 2||F||_\infty \int_{\bR^{2n}\setminus {\cal U}_{1/m}}
  |\Psi_\hbar(\sigma)|^2\frac{d\sigma}{(2\pi \hbar)^n}\:,\label{lastline223} \eeq
 for every given  $m\in \bN \setminus\{0\}$.
Analogously to lemma \ref{fclaim}, we now have
\begin{lemma}\label{fclaim2} Under the hypotheses of Thm \ref{thm:claslimNEWSchr},  (b) in particular,  and $F$ defined in (\ref{Ff}),
 for every $\epsilon>0$, there is  $m_\epsilon\in \bN $ such that $\sup_{\sigma \in {\cal U}_{1/m_\epsilon}}|F(\sigma)  - F(\sigma_0)| < \epsilon/2\:,$
where  ${\cal U}_\delta$ is defined in  (\ref{NBH}).
\end{lemma}
\begin{proof}
See Appendix \ref{appA}.
\end{proof}
\noindent 
Keeping that $m_\epsilon$ and exploiting (1) in proposition \ref{localizationSCH}, we can find $H_\epsilon>0$ such that
\begin{align}
2||F||_\infty \int_{\bR^{2n}\setminus {\cal U}_{1/m_\epsilon}}  |\Phi_\hbar(\sigma)|^2\frac{d\sigma}{(2\pi \hbar)^n} < \epsilon/2
, \nonumber 
\end{align}
for $0<\hbar< H_\epsilon$. Looking at (\ref{lastline223}), we conclude that  for every $\epsilon>0$, there is $H_\epsilon$ such that  $0<\hbar< H_\epsilon$ implies
$\left|\langle\psi^{(0)}_\hbar,Q_{\hbar}^B(f)\psi^{(0)}_\hbar\rangle - F(\sigma_0) \right| < \epsilon$,
concluding the proof.
\end{proof}
%
%

We stress that, in general,  $G$ 
does not  necessarily defines a {\em simultaneous}  symmetry group for $h:= p^2+ V$ and  $H_\hbar = \overline{-\hbar^2\Delta + V}$ if $\hbar>0$
and this requirement is however unnecessary for the validity of the theorem above.
 Nevertheless it happens  for example  in the special physically relevant cases of examples \ref{physicalexamples} and in these cases {\em both   theorem \ref{thm:claslimNW2} and theorem
\ref{thm:claslimNEWSchr}  are valid}.

\begin{proposition}\label{TGB}  Let the Hamiltonian in $L^2(\bR^n, dx)$, $n\geq 1$, be of the form $$H_\hbar = \overline{-\hbar^2 \Delta + V}$$
with $V$ satisfying (V1)-(V3).
Consider the natural action   of a group of (some) isometries of $\bR^n$, $G \ni (R,a) : \bR^n \ni x \mapsto a + Rx \in \bR^n$, with $R \in O(n)$ ($R \in \bZ_2$ if $n=1$) and $a\in \bR^n$,
in terms of symplectomorphysms on $\bR^{2n}$,
$$(R,a): \bR^{2n} \ni (q,p) \mapsto  (a+Rq, Rp) \in \bR^{2n}\:, \quad  (R,a) \in G\:.$$
If $V$ is $G$-invariant, then  the  unitary representation $G \ni g \mpasto U_g \in \gB(L^2(\bR^n, dx))$ constructed according to proposition \ref{equivariance2}
  leaves  $H_\hbar$ invariant for every given $\hbar>0$:
\beq\label{invH}
U_g H_\hbar U_g^{-1} = H_\hbar\:.
\eeq
\end{proposition}

\begin{proof} Per direct inspection from  $U_g :=W^*u_gW$ and (\ref{alternate})   we have that
$(U_{(R,a)} \psi)(x) = \psi((R,a)^{-1} x)$. (Notice that, in this case,  $U_{(R,a)}$ is unitary just because $(R,a)$ is an isometry of $\bR^n$ and thus leaves $dx$ invariant.)
The said action of $U_g$ leaves ${\cal S}(\bR^n)$ invariant.
 Furthermore  $U_g \Delta|_{{\cal S}(\bR^n)}U_g^{-1}= \Delta|_{{\cal S}(\bR^n)}$ (again because the action of $g$ is an isometry of $\bR^n$)
and  $U_g V|_{{\cal S}(\bR^n)}U_g^{-1}= V|_{{\cal S}(\bR^n)}$ by hypothesis,
so that
$U_g H_\hbar|_{{\cal S}(\bR^n)} U_g^{-1} = H_\hbar|_{{\cal S}(\bR^n)}$.
Taking the closure of both sides, we have (\ref{invH}).
\end{proof}

The result applies to the cases $V(x) := (x^2-1)^2$ and $G:=\bZ_2$ for $n=1$ and $G:= SO(n)$ if $n> 1$ in particular.

Taking Corollary \ref{corollarylink} into account, this discussion eventually produces a direct comparison of the classical limits referred to the eigenvectors $\varphi^{(0)}_\hbar$ 
and $\psi_\hbar^{(0)}$ respectively of 
$Q^B_\hbar(e^{-t(p^2+V)})$ and $e^{-t\overline{(-\hbar^2 \Delta + V)}}$
and referred to the maximal eigenvalues (i.e. the minimal if focusing on $\overline{-\hbar^2 \Delta + V}$).

\begin{corollary}\label{corcomp}
Assume that one of the two cases is valid.
\begin{itemize}
\item[(1)] The hypotheses  Theorem \ref{MT12} are valid (so that Thm \ref{MT2} holds 
with $\Lambda = \max_{(q,p)\in \bR^{2n}} e^{-t(p^2+V(q))}$),
\item[(2)] The hypotheses of Theorem \ref{thm:claslimNEWSchr} and Thm \ref{thm:claslimNW2}, 
with $\Lambda = \max_{(q,p)\in \bR^{2n}} e^{-t(p^2+V(q))}$,
are simultaneously valid
for the same group $G$  (this happens in particular for the case of proposition \ref{TGB}).
\end{itemize}
Then 
$$\lim_{\hbar \to 0^+} \left( \langle \psi^{(0)}_\hbar, Q_\hbar^B(f) \psi^{(0)}_\hbar\rangle
-\langle \varphi^{(0)}_\hbar, Q_\hbar^B(f)\varphi^{(0)}_\hbar\rangle \right) =0, \quad \mbox{for every $f\in C_0(\bR^{2n})$}.$$
\end{corollary}

\section{Spontaneous symmetry breaking as emergent phenomenon}\label{SSBEMERGENCE}
In this section we introduce the notion of spontaneous symmetry breaking in al algebraic framework. In particular, we see that SSB applies to commutative as well as non-commutative $C^*$-algebras, and therefore to classical and quantum theories. In the event that they are encoded by a continuous bundle of $C^*$-algebras (cf. Definition \ref{def:continuous algebra bundle}) this allows one to study SSB as a possibly emergent phenomenon by switching of a semiclassical parameter (e.g. Planck's constant occuring in some index set $I$). All that is needed are the continuity properties of the $C^*$-bundle specified by the continuous cross-sections. This approach is therefore perfectly suitable for studying emergent phenomena arising in the classical limit of an underlying quantum theory (see also \cite{Ven2021}). 

Let us now remind the reader the general context where the notion of spontaneous symmetry breaking takes place.
A $C^*$-{\bf dynamical system} $(\gA, \alpha)$ is a
 $C^*$-algebra $\gA$ equipped with a {\bf dynamical evolution}, i.e.,   a one-parameter group of $C^*$-algebra automorphisms $\alpha:= \{\alpha_t\}_{t \in \bR}$ that is {\em strongly continuous}  on $\gA$:   the map $\bR \ni t \mapsto \alpha_t(a) \in \gA$ is continuous for every $a\in \gA$.
\subsection{Dynamical symmetry groups and SSB}

Given a $C^*$-{\bf dynamical system} $(\gA, \alpha)$  and an {\bf $\alpha$-invariant state} $\omega$, 
i.e., $\omega(a) = \omega(\alpha_t(a))$ for every $a\in \gA$ and $t\in \bR$,
there is a  unique one-parameter group of unitaries $U := \{U_t\}_{t\in \bR}$ which implements $\alpha$ in the GNS representation, i.e.,  $\pi_\omega(\alpha_t(a))= U^{-1}_t \pi_\omega(a) U_t$, 
and leaves fixed the cyclic vector $U_t\Psi_\omega = \Psi_\omega$ (see e.g., \cite{Lan17, Mor}).  $t \mapsto U_t$ is {\em at least} strongly continuous on $\cal H_\omega$ as a consequence of the strong continuity of $t\mapsto \alpha_t$ on $\gA$ and the properties of the GNS construction. Indeed,  we have that $||U_t^{-1} \pi_\omega(a) \Psi_\omega - \pi_\omega(a) \Psi_\omega|| =
||U_t^{-1} \pi_\omega(a) U_t\Psi_\omega - \pi_\omega(a) \Psi_\omega||\leq  ||U_t^{-1} \pi_\omega(a) U_t - \pi_\omega(a)|| = ||\alpha_t(a) - a||$.
\begin{remark}
{\em The strong continuity of $t\mapsto \alpha_t$ of automorphims of $\gA$  is not always compatible  with {\em unbounded} selfadjoint generators of $U$ (and bounded generator is equivalent to say that $\{U_t\}_{t\in \bR}$ is continuous {\em also} in the operator norm of $\gB({\cal H}_\omega)$) as shown in Example 3.2.36 in \cite{BR1}.
 Even if we shall deal with unbounded selfadjoint generators, the case we are about to discuss is safe. That is  because  the relevant $C^*$ algebra is made of compact operators and the following general result applies.
\begin{proposition}\label{compact}
Let  $A \in \gA := \gB_\infty{(\cal H)}$ and $\{U_t\}_{t\in \bR}$ be a  one-parameter group of unitary operators in the Hilbert space ${\cal H}$ that is strongly continuous in $\gB({\cal H})$. The one-parameter group of $C^*$-algebra automorphisms induced by $U$ is strongly continuous in $\gA$:
$||U_t^{-1}A U_t - U_u^{-1} A U_u|| \to 0$ if $t\to u$ (also if the  selfajoint generator of $U$ is unbounded). 
\end{proposition}}
\end{remark}

\begin{proof} See Appendix \ref{appA}. \end{proof}

A {\bf ground 
state} of a $C^*$-dynamical system $(\gA, \alpha)$ is an algebraic state  $\omega: \gA \to \bC$ such that 
\begin{itemize}
\item[(a)]
the state is $\alpha$-invariant, i.e, $\omega(\alpha_t(a)) = \omega(a)$ for all $t\in \bR$ and all $a\in \gA$,
 \item[(b)] the selfadjoint generator $H$ of the strongly-continuous 
 one-parameter  unitary group $U_t= e^{-itH}$
which implements $\alpha$ in a given  GNS representation $({\cal H}_\omega, \pi_\omega, \Psi_\omega)$ under the requirement $U_t\Psi_\omega = \Psi_\omega$, has spectrum $\sigma(H) \subset [0,+\infty)$. 
\end{itemize}
It is not difficult to prove that the set  $S^{ground}(\gA, \alpha)$ of ground states of $(\gA, \alpha)$ is convex and $*$-weak closed (see e.g. \cite{Lan17}), so that it is also compact for the Banach-Alaoglu theorem. The Krein-Milman theorem implies that 
all ground states can be constructed out of limit points of convex combinations of  {\em extremal ground states}  in the $*$-weak topology. The relevance of 
the extremal ground states relies upon this property of them: they are the building blocks for constructing all other ground states exactly as {\bf pure states}, i.e.,  extremal 
states in the convex $*$-weak  compact set of all algebraic states on $\gA$, 
 are the building blocks for constructing all algebraic states.
 However the elements of $S^{ground}(\gA, \alpha)$ are not necessarily pure states.
Nonetheless,  in many cases of physical interest,  extremal ground states are 
exactly the pure states  which are also ground states \cite{Lan17}.

When $(\gA, \alpha)$ is a $C^*$-dynamical system also endowed with a group $G$ acting on $\gA$ with a group representation $\gamma : G \ni g
 \to \gamma_g$ in terms of $C^*$-automorphisms $\gamma_g : \gA \to \gA$, we say that $G$ is a {\bf dynamical symmetry group} if
$\gamma_g \circ \alpha_t = \alpha_t \circ \gamma_g\quad \mbox{for all $g\in G$ and $t\in \bR$.}$

According  to \cite{Lan17},  
{\bf spontaneous symmetry breaking (SSB)}  occurs for a dynamical system $(\gA, \alpha)$ endowed with  a dynamical symmetry group $G$ if there are no $G$-invariant 
ground states which are extreme points in $S^{ground}(\gA, \alpha)$. Within the usual situation where the extremal points in 
$S^{ground}(\gA, \alpha)$ are the ground pure states of $\gA$, occurrence of SSB means that $G$-invariant ground states must be necessarily mixed states.
A more frequent situation is when there are extreme points in $S^{ground}(\gA, \alpha)$ which are not $G$-invariant. In this case one says that {\bf weak SSB} takes place.

For the sake of shortness we can only stick to the succinctly illustrated  relevant technical definitions,  an exhaustive discussion on the physical importance  of SSB in various contexts
and on the different also inequivalent definitions of SSB is presented in \cite{Lan17,VGRL18}.

\subsection{SSB of ground states  as an emergent phenomenon in Berezin quantization on $\bR^{2n}$ with Schr\"odinger Hamiltonians}

The above definitions applies in particular to the commutative case where $\gA := C_0(X)$ endowed with the $C^*$-norm $||\cdot||_\infty$, referred to a symplectic manifold
 $X$ and the associated Poisson structure $(C^\infty(X), \{\cdot, \cdot\})$.  In this case the states $\omega$ are nothing but the {\em regular}\footnote{All positive Borel measures on $\bR^n$ are automatically regular due to Theorem 2.18  in \cite{Rudin}.} {\em Borel probability measures} $\mu_\omega$ over $X$.  More precisely, if $\omega : \gA \to \bC$ is an algebraic  state, the $C_0(X)$ version of  Riesz's representation theorem of generally complex  measures on locally compact Hausdorff spaces \cite{Rudin}, taking continuity, positivity  and $||\omega||=1$, into account,   proves that $({\cal H}_\omega, \pi_\omega, \Psi_\omega)$ has this form
$${\cal H}_\omega = L^2(X, \mu_\omega)\:, \quad  (\pi_\omega(f)\psi)(\sigma) =f(\sigma) \psi(\sigma)\:,\quad 
 \Psi_\omega(\sigma) = 1\:,  \quad \mbox{for all $f\in C_0(X)$, $\psi\in 
{\cal H}_\omega$ and  $ \sigma \in X$.}$$
 With this representation, the pure states are {\em Dirac measures} concentrated at any point $\sigma \in X$.

A $C^*$-dynamical system structure is constructed when 
the dynamical evolution is furnished by the pullback action of the Hamiltonian flow  $\phi^{(h)}$,
 {\em provided it is complete}, generated by a (real) hamiltonian function $h\in C^\infty(X)$, i.e., 
$\alpha^{(h)}_t(f) := f \circ \phi_t^{(h)}$ for every $f\in C_0(X)$ and $t\in \bR$. 

  It is easy to demonstrate  that $(C_0(X), \alpha^{(h)})$ is a $C^*$-dynamical system (in particular $\alpha^{(h)}$ leaves $C_0(X)$ invariant and is strongly continuous\footnote{To prove the continuity at $t=0$, use in particular  the fact that, if $f\in C_0(X)$ and $\delta>0$, then $|f(\phi^{(h)}_t(\sigma))| < \epsilon$ for $|t| \leq \delta$ and $\sigma \not\in K_{\epsilon, \delta}$, where 
the latter set is compact and $K_{\epsilon,\delta} := \Pi_X \Phi([-\delta,\delta] \times K_\epsilon)$ with 
$\Phi: (t,\sigma) \mapsto (t, \phi^{(h)}_t(\sigma))$, $\Pi_X : \bR \times X \to X$ being the canonical projection, and the compact $K_\epsilon \subset X$ is such that $|f(\tau)| < \epsilon$ if $\tau \not \in K_\epsilon$.}).  
By direct inspection one sees that $\pi_\omega(C_c^\infty(X))$
is included in the domain of the selfadjoint generator of the unitary implementation of $\alpha^{(h)}$ in the GNS representation of a state $\omega$. This  generator acts  as $-i \{h, \pi_\omega(f)\}= -i \{h, f\}$ if $f\in C_c^\infty(X)$, where $\{\cdot,\cdot\}$ is the Poisson bracket associated to the symplectic form.

\begin{proposition}\label{propNh}
The ground states of $(C_0(X), \alpha^{(h)})$ are all of the regular Borel  probability measures on $X$ whose support is contained in the closed set $N_h := \{\sigma \in X \:|\: dh(\sigma) =0\}$.  
\end{proposition}

\begin{proof} See Appendix \ref{appA}.
\end{proof}

If $\omega$ is a ground state of $(C_0(X), \alpha^{(h)})$, in view of the above discussion, 
  $\alpha^{(h)}$ is trivially implemented:  $U_t = I$ for every $t\in \bR$ and the positivity 
 condition on  the spectrum of the generator of $U_t$ is automatically fulfilled. In this case  the extremal elements of $S^{ground}(\gA, \alpha)$ are 
the Dirac measures concentrated at the points $\sigma\in X$ such that $dh(\sigma)=0$. In particular {\em extremal} ground states 
are {\em pure} states.

\begin{proposition}\label{WEAKSSB} Consider  the $C^*$-dynamical system $(C_0(\bR^{2n}), \alpha^{(h)})$ where $\alpha^{(h)}$ is generated by the Hamiltonian $h = p^2 + V$ with
\beq V(q) = (q^2-1)^2\:.\label{defVq} \eeq
Consider  the  natural action 
$g: (q,p) \mapsto (gq,gp)$ of $g\in G$  in terms of sympectomorphisms of $(\bR^{2n}, \sum_{k=1}^n dp_k \wedge dq^k)$
  as in proposition  \ref{TGB}
with   $G:= \bZ_2$ if $n=1$ or $G:= SO(n)$ if $n>1$.
Then $G$ is a dynamical symmetry group of $(C_0(\bR^{2n}), \alpha^{(h)})$  with action 
\beq \gamma_g f := f \circ g^{-1} \quad \mbox{for all $g\in G$ and $f\in C_0(\bR^{2n})$}\label{gammaf},\eeq
 and  weak SSB occurs.  
\end{proposition}

\begin{proof} The Hamiltonian flow $\alpha^{(h)}$ is complete since the level sets of $h$ are compact and every solution of Hamilton equations  is contained in one such set as $h$ is dynamically conserved.
Since  the action of $G$ is given by symplectomorphisms and every $\gamma_g$ leaves $h$ invariant and thus it commutes with the Hamiltonian flow, $G$ is a dynamical symmetry group of  the $C^*$-dynamical system $(C_0(\bR^{2n}), \alpha^{(h)})$. From the discussion above, the extremal ground states are defined by the Dirac measures concentrated at  the set of zeros of $dh$.  In all cases the only $G$-invariant extremal ground state is located at $(q_0,p_0) = (0,0)$. There is however a plethora of non-$G$ invariant extremal ground states located at the points $(q, 0)$ with $|q|=1$ if $n>2$ and exactly two non-$G$ invariant extremal ground states $(\pm1, 0)$ if $n=1$.
\end{proof}

The accumulated results permits us to discuss the phenomenon of (weak) spontaneous symmetry breaking 
as an {\em emergent phenomenon} when passing from the quantum realm to the classical world by switching off $\hbar$ \cite{Lan17,VGRL18}.
We address the reader to Chapter 10 of \cite{Lan17} for an wide discussion on the physical relevance of this viewpoint and the various implications in understanding the quantum-classical transition.

We now consider  the commutative  $C^*$-algebra $\gA_0:= C_0(\bR^{2n})$ is just the $\hbar=0$ fiber of the continuous $C^*$-bundle associated to the Berezin deformation quantization maps  for $\hbar>0$, $Q_\hbar^B : C_0(\bR^{2n})\to  \gA_\hbar := \gB_\infty(L^2(\bR^n dx))$.
Every bounded observable $A \in \gA_\hbar$ can be written as $Q^B_\hbar(f_A)$ for a suitable and generally non-unique $f_A\in C_0(\bR^{2n})$ as a consequence of (4) in theorem
 \ref{teoQB1}.
The dynamical   evolution described by a $\hbar$-parametrized family of  one-parameter group of $C^*$-authomorphisms 
$\bR \ni t \mapsto \alpha^{\hbar}_t$
is provided   by  a corresponding 
$\hbar$-parametrized family 
of one-parameter unitary
groups $\bR \ni t \mapsto U^\hbar_t := e^{-it H_\hbar}$:
\beq
\alpha^\hbar_t(A) := U^\hbar_{-t} A U^\hbar_t, \quad \mbox (A \in \gA_\hbar\:). \label{alphahbar}
\eeq
Notice that the one-parameter group $t\mapsto \alpha_t^\hbar$ on $\gA_\hbar$ is strongly continuous  due to proposition \ref{compact} and thus $(\gA_\hbar, \alpha^\hbar)$
is a $C^*$-dynamical system.
We assume that $H_\hbar := \overline{-\hbar^2 \Delta +V}$ where $V$ satisfies (V1)-(V3). In this case $H_\hbar$ is affiliated to $\gA_\hbar$ viewed as a von Neumann algebra,
because the resolvent of $H_\hbar$ is compact  (Theorem XIII.67 in \cite{RS4}) and thus, in particular, $H_\hbar$ is an observable of the physical system represented by $\gA_\hbar$.
 
That is not the whole story because, if choosing as before $V$ as in (\ref{defVq}), it turns out that $G = \bZ_2$, for $n=1$, or $G:= SO(n)$, if $n>1$,  becomes a 
dynamical symmetry group as a consequence of Prop. \ref{TGB} when the unitary action of $G$ is given by the $Q_\hbar^B$-equivariant representation (\ref{equivariantofQ23}), so that $U_g Q_\hbar^B(f) U_g^* = Q^B_\hbar(\gamma_g f)$, where $\gamma_g$ is the same as in (\ref{gammaf}).
Differently from the classical ($\hbar=0$) case here no SSB occurs. This fact should be  physically evident since there is only one ``ground state'' (in the sense of a vector state with minimal energy)  which is $G$-invariant. However the algebraic notion of ground state given above seems to be more complex and it deserves a closer scrutiny. We have the following general result.

\begin{proposition}\label{NOSSB}
Consider the $C^*$-dynamical system $(\gB_\infty(L^2(\bR^n,dx)), \alpha^{\hbar})$, the latter defined in  (\ref{alphahbar}), and with dynamical symmetry group 
$G$ whose unitary and $Q_\hbar^B$-equivariant  action is defined in  (\ref{equivariantofQ2})-(\ref{equivariantofQ23}).  No SSB (or weak SSB) occurs for $\hbar>0$.
\end{proposition}

\begin{proof} As is well known (see e.g., Theorem 7.75 in \cite{Mor}), if ${\cal H}$ is any complex Hilbert space, the 
algebraic states  on $\gB_\infty({\cal H})$ are all normal and coincide with the  {\bf statistical operators}:  trace class, unit trace, positive operators  $\rho :{\cal H} \to {\cal H}$. Here $\omega_\rho(A) = Tr(\rho A)$ for every $A \in \gB_\infty({\cal H})$. 
Let us  assume the dynamical  invariance property $\omega_\rho (\alpha_t^\hbar(A)) = \omega_\rho (A)$ for every $t\in \bR$ and $A\in \gB_\infty(L^2(\bR^n,dx))$.
 Taking $A = \langle  \cdot, \psi^{(n)}_\hbar \rangle \psi^{(m)}_\hbar$ where $H_\hbar \psi^{(n)}_\hbar=E^{(n)}_\hbar \psi^{(n)}_\hbar$\footnote{According to (\ref{specH}), $n=0,1,\ldots $ labels the Hilbert basis of  eigenvectors of $H_\hbar$ with non-decreasing eigenvalues.} it is easy to prove that $\rho$ must commute with the PVM of $H_\hbar$ so that, in the strong-topology $\rho = \sum_{n = 0}^{+\infty} p_n \langle \cdot, \psi^{(n)}_\hbar \rangle\psi^{(n)}_\hbar$. Where $p_n \geq 0$ and $\sum_{n} p_n =1$.
The eigenvectors $\psi^{(n)}_\hbar$ may be a  rearrangement of the initial ones separately in each eigenspace of $H_\hbar$. 
The GNS representation of $\rho$ takes this form (up to unitary equivalence)  as the reader can prove by direct inspection $${\cal H}_\rho := \oplus_{p_n \neq 0}L^2(\bR^n,dx)\:; \quad \pi_\rho(A) :=
\oplus_{p_n \neq 0}  A\:; \quad  \Psi_\rho :=\oplus_{p_n \neq 0} \sqrt{p_n} \psi^{(n)}_\hbar \:.$$
It is not difficult to see that the strongly continuous one-parameter group of unitaries $V$ which leaves $\Psi_\rho$ invariant and implements $\alpha^\hbar$ is here
$V_{t} \left(\oplus_{p_n \neq 0} \phi_n \right)= \oplus_{p_n \neq 0} e^{-it(H_\hbar -E^{(n)}_\hbar)}\phi_n$. Its generator is $K = \oplus_{p_n \neq 0}  (H_\hbar - E^{(n)}_\hbar I)$. It is obvious that $\sigma(K) \subset [0,+\infty)$ only if $p_n \neq 0$ for $n=0$ so that $\rho = \langle\cdot, \psi_\hbar^{(0)} \rangle
\psi_\hbar^{(0)}$.  The unique algebraic ground state is therefore $\omega^{(0)}_\hbar(A) = \langle \psi_\hbar^{(0)}, A \psi_\hbar^{(0)}\rangle$ for $A\in \gB_\hbar$.  This state is also $G$-invariant since the eigenspace of $H_\hbar$ with minimal eigenvalue has dimension $1$ as established in proposition \ref{propDIM}  and the PVM of $H_\hbar$ commutes with the unitary representation of $G$ since  (\ref{invH}) is valid.
\end{proof}
\noindent  Coming back to the potential $V(q) = (q^2-1)^2$ in particular,  the existence of the classical limits established in theorem \ref{thm:claslimNEWSchr} now specializes to
\begin{itemize}
\item[(1)] if $n>1$,
\begin{align}
\lim_{\hbar\to 0^+}\langle\psi^{(0)}_\hbar,Q_{\hbar}^B(f)\psi^{(0)}_\hbar\rangle=\int_{SO(n)} f(g\sigma_0)d\mu_{SO(n)}(g),  \quad  \mbox{for every $f\in C_0(\mathbb{R}^{2n})$;}\label{classical limit22X}
\end{align}
where $\mu_{SO(n)}$ is the normalized Haar measure of $SO(n)$ and $\sigma_0 = (q_0, 0)$ with $|q_0|=1$,
\item[(2)] if $n=1$,
\begin{align}
\lim_{\hbar\to 0^+}\langle\psi^{(0)}_\hbar,Q_{\hbar}^B(f)\psi^{(0)}_\hbar\rangle=\frac{1}{2} (f(1,0) + f(-1,0)),  \quad  \mbox{for every $f\in C_0(\mathbb{R}^{2})$.}\label{classical limit2b2X}
\end{align}
\end{itemize}
It is  of utmost relevance  to notice that 
  both right-hand sides can be recast to  integrals with respect to $SO(n)/\bZ_2$-invariant  probability measures $\mu$ and $\nu$  on $\bR^{2n}$ with supports given by the whole orbit  $G\sigma_0$, where $G:= SO(n)$ or $\bZ_2$ respectively. The following more general result holds.

\begin{proposition}\label{propmeas} 
Let $G$ be  a topological compact or finite  group  with a   (continuous in the first case)  on $\mathbb{R}^{m}$. Then  there are two regular Borel probability measures on $\bR^{m}$, respectively 
$\mu$ and $\nu$,
such that 
$$\int_{G} f(g\sigma_0)d\mu_G(g)  = \int_{\bR^m} f d\mu \:;\qquad     \frac{1}{N_G} \sum_{g\in G} f(g \sigma_0) = \int_{\bR^m} f d\nu,\qquad \mbox{for all $f\in C_0(\bR^{m})$,}$$ 
where $\mu_G$ is the normalized Haar measure on $G$ in the first case and $N_G$ is the number of elements of $G$ in the second case.
These measures are invariant under the action of $G$ on $\bR^{m}$ and each of their  supports  is  the whole orbit $G\sigma_0$.
\end{proposition}

\begin{proof}
It is sufficient to prove the thesis for the former case, the latter being easier. Noticing that $G\sigma_0$ is compact, one has that linear map
$C_0(\bR^{m})\ni f \mapsto \int_G f(g\sigma_0) d\mu_G(g) \in \bC$
 is  $||\cdot||_\infty$ continuous on $ C_0(\bR^{m})$. Riesz' theorem for complex measures implies that it is the integral with respect to a (uniquely defined) regular Borel complex measure
$\mu$  on $\bR^{2n}$. This measure is actually positive, since the map is positive. Riesz' theorem for positive measures implies that  the support of the measure is included in the closed set $G\sigma_0$,
since the map vanishes when evaluated on  $f \in C_c(\bR^{m})$  whose  support is included in an open set with empty intersection with $G\sigma_0$.
The integral 
produces the value $1$
if $f(\sigma)=1$ on the compact orbit $G\sigma_0$, hence $\mu$ is a probability one.  The invariance of the integral under the pullback action of $G$ on its argument $f$ proves that $\mu$ is also $G$ invariant, again using the uniqueness property in Riesz' theorem for positive measures.
Since $\mu(G\sigma_0)=1$, its support cannot be empty. If $\sigma$ belongs to the support also $g\sigma$ does for every $g\in G$. As the action of $G$ is transitive on $G\sigma_0$ we conclude that $\mbox{supp}(\mu) = G\sigma_0$. 
\end{proof}

These measures, which by definition are  $SO(n)/\bZ_2$-invariant ground states of the dynamical system $(\gA_0, \alpha^{0}) = (C_0(\bR^{2n}), \alpha^{(h)})$
with $h(q,p)= p^2+(q^2-1)^2$,
 are therefore  not concentrated on  single points, but they are concentrated on the set of points where $h$ attains its minimum value: the orbit 
$\{(p=0, q)\:|\: |q|=1\}$ in the $SO(n)$ case  and the set $\{(p=0, q=\pm1)\}$ in the $\bZ_2$ case. Hence they are exactly a case of non-extremal $G$-invariant ground states responsible for the weak SSB discussed above.

All that proves that (weak) SSB of $SO(n)/\bZ_2$ occurs in the classical limit $\hbar\to 0^+$, when achieving the theory in $\gA_0= C_0(\bR^{2n})$ with classical hamiltonian $h(q,p)=p^2+V(q)$, where $V(q) = (q^2-1)^2$, from the quantum theory in $\gA_\hbar =\gB_\infty( L^2(\bR^n, dx))$ with Hamiltonian $H_\hbar = \overline{-\hbar^2 \Delta + V}$.
In this sense SSB shows up here  as an {\em emergent phenomenon} \cite{Lan17} when passing from the quantum to the classical realm.

\begin{remark}\label{remfin}
{\em We stress that, this emergent  phenomenon is quite general in the framework of Berezin quantization on $\gA_0:= C_0(\bR^{2n})$
when we use the Berizin map $Q_\hbar^B(f)$ (with $Q_0^B(f):=f$) to describe the observables in  $\gA_\hbar =\gB_\infty( L^2(\bR^n, dx))$.
By collecting Theorem \ref{thm:claslimNEWSchr}, 
Prop. \ref{propNh},
Prop. \ref{NOSSB}, and Prop. \ref{propmeas},  we see that, when 
dealing with a classical  Hamiltonian $h(q,p)= p^2+V(q)$ and its quantum companion $H_\hbar = \overline{-\hbar^2 \Delta + V}$,  where  $V$ satisfies  (V1)-(V3)\footnote{Conditions (V1)-(V3), with $V(q) \to +\infty$ if $|q| \to +\infty$ in particular, imply that the flow of $h=p^2+V$ is complete.}, emergent weak SSB shows up   for some group $G$  provided that the following requirements hold.
\begin{itemize} 
 \item[(a)] 
$G$ is compact or finite (its action on $\bR^{2n}$ is continuous in the former case) and leaves  invariant  both $h$ and $H_\hbar$. The action of $G$ in the quantum Hilbert space is   the $Q_\hbar^B$-equivariant action  induced from  the action of $G$  on $\bR^{2n}$ by  symplectomorhisms (Prop. \ref{equivariance2});
\item[(b)] $h^{-1}(\{\min h\})$ (equivalently, $V^{-1}(\{\min V\})$) includes more than one point;
\item[(c)]  the action of $G$ is transitive on $h^{-1}(\{\min h\})$.
\end{itemize}
According to  Prop. \ref{TGB}, condition (a) is in particular valid when $G$ is made of some isometries  of $\bR^n$,  $(R, a) : \bR^n \ni x \mapsto a + Rx \in \bR^n$ ($a \in \bR^n$, $R \in O(n)$)
which leave $V$ invariant and their action  in terms of symplectomorphisms on $\bR^{2n}$ is the standard one $(R, a) : \bR^{2n} \ni (q,p) \mapsto (a + Rq, Rp) \in \bR^{2n}$.}
\end{remark}

\section*{Acknowledgments} C.J.F. van de Ven  is a Marie Sk\l odowska-Curie fellow of the Istituto
 Nazionale di Alta Matematica and is funded by the INdAM Doctoral Programme in Mathematics and/or 
Applications co-funded by Marie Sk\l odowska-Curie Actions, INdAM-DP-COFUND-2015, grant number 713485. 
The authors are grateful to Nicol\`o Drago,  Sonia Mazzucchi and Cosmas Zachos for helpful discussions.
\appendix
\section{Proofs of some propositions}\label{appA}

\noindent{\bf Proof of Proposition \ref{QWgen}}.
(1) If $f \in  {\cal S}(\bR^{2n})$ then  $\widehat{f} \in  {\cal S}(\bR^{2n})$ as well and
$$|\langle \psi, Q_\hbar^W(f) \phi \rangle |\leq ||\psi||\: ||\phi||\: \frac{1}{(2\pi)^n} \int_{\bR^{2n}} |\widehat{f}(a,b)| dadb,$$
so that $Q_\hbar^W(f)$ exists due to the Riesz lemma and is bounded  and thus  it can be extended to the whole Hilbert space. 
The estimate
(\ref{estimate}) holds trivially. Finally observe that in (\ref{Wgen}) $\widehat{f} \in {\cal S}(\bR^{2n})$ so that the right-hand side is defined for generic vectors $\psi, \phi \in L^2(\bR^2,dx)$ since the map 
$a,b \mapsto  \langle e^{-ia \cdot X}\psi,  e^{ib\cdot P} \phi \rangle e^{-\frac{i\hbar}{2} a\cdot b} $ is bounded. A direct application of the Lebesgue dominated convergence theorem also  using  the fact that ${\cal S}(\bR^n)$ is dense in $L^2(\bR^n, dx)$ proves that 
if valid for generic vectors $\psi,\phi$ and where $Q_\hbar(f)$ is the unique bouded linear extension of the  operator initially defined on  ${\cal S}(\bR^n)$.\\
(2)  $f \in {\cal S}'(\bR^{2n})$  so that $\langle \psi, Q_\hbar^W(f) \phi \rangle$ equals
$$\frac{1}{(2\pi)^n} \int_{\bR^{2n}} \langle e^{-ia \cdot X}\psi,  e^{ib\cdot P} \phi \rangle e^{-\frac{i\hbar}{2} a\cdot b}  (2\pi)^{n/2}\delta(b) \widehat{f}(a) dadb=\frac{1}{(2\pi)^{n/2}} \int_{\bR^{2n}} \langle e^{-ia \cdot X}\psi,  \phi \rangle \widehat{f}(a) da\:.$$
We can now exploit the spectral decomposition \cite{Mor} of $e^{-ia \cdot X}$, 
$\langle  e^{-ia \cdot X}\psi, \phi \rangle =
\int_{\bR} e^{ia \cdot \lambda} d\mu_{\psi, \phi}(\lambda)$,
noticing that this function of $a$ belongs to ${\cal S}(\bR^n)$(it being the Fourier transform of a function of that space).
From the definition of Fourier transform of distributions and the definition of Fourier transform of (finite) complex measures, we conclude that 
$$\langle \psi, Q^W_\hbar(f) \phi \rangle =  \frac{1}{(2\pi)^{n/2}} \int_{\bR^{n}} \widehat{f}(a)\int_{\bR}   e^{ia \cdot \lambda} d\mu_{\psi, \phi}(\lambda) da
= \int_\bR f(\lambda) d\mu_{\psi, \phi}(\lambda)\:.$$
Hence
$\langle \psi, Q^W_\hbar(f) \phi \rangle =  \langle \psi, f(X) \phi \rangle$
and the final bound in the thesis is valid.\\
(3) The proof is strictly analogous to that of (2).\\
(4) With a procedure analogous to the one of the   previous cases using in particular the last expression in (\ref{decwayl}),  we easily prove that the  operator 
$Q_\hbar^W(f_1f_2)$, which is 
bounded and everywhere defined in view of the case (1),
is completely determined by its quadratic form (\ref{formfin}).  \hfill $\Box$\\

\noindent {\bf Proof of proposition \ref{prooff}}. 
(a) Let $f\in C_0(\bR^{n})$ be a function of the variable $q$.   If $\phi \in {\cal S}(\bR^n)$, so that all integration can be interchanged and using
$$\frac{1}{(2\pi\hbar)^n}\int_{\bR^{n}} \int_{\bR^n} e^{i\frac{p(x-x')}{\hbar}}  g(x') dx' dp = g(x')\:, \quad (g \in {\cal S}(\bR^n))\:;$$
we have
$$(Q^B_\hbar(f)\phi)(x) = \frac{1}{2^n(\pi\hbar)^{3n/2}}\int_{\bR^{2n}} f(q) e^{-\frac{(x-q)^2}{2\hbar}}\int_{\bR^n} e^{i\frac{p(x-x')}{\hbar}}   
e^{-\frac{(x'-q)^2}{2\hbar}}\phi(x')dx' dqdp$$ 
$$=  \frac{1}{(\pi\hbar)^{n/2}}\int_{\bR^{n}} f(q) e^{-\frac{(x-q)^2}{2\hbar}}
e^{-\frac{(x-q)^2}{2\hbar}}\phi(x) dq = \frac{1}{(\pi\hbar)^{n/2}}\int_{\bR^{n}} f(q) e^{-\frac{(x-q)^2}{\hbar}} \phi(x) dq\:.
$$
In summary,
$(Q^B_\hbar(f)\phi)(x) =\left( \frac{1}{(\pi\hbar)^{n/2}}\int_{\bR^{n}} f(q) e^{-\frac{(x-q)^2}{\hbar}}  dq\right) \phi(x)\:.$
An easy density argument of ${\cal S}(\bR^n)$ in $L^2(\bR^n, dx)$ extends the above result to the general case
of $\phi \in L^2(\bR^n, dx)$.  To conclude the proof it is sufficient to prove that, if $f\in C_0(\bR^n)$, then 
\beq \frac{1}{(\pi\hbar)^{n/2}}\int_{\bR^{n}} f(q) e^{-\frac{(x-q)^2}{\hbar}}  dq \to f(x) \quad \mbox{uniformly in $x\in \bR^n$ if $\hbar \to 0^+$,}\label{uniff}\eeq
since, considering the functions a multiplicative operators, 
$$\left|\left| \frac{1}{(\pi\hbar)^{n/2}}\int_{\bR^{n}} f(q) e^{-\frac{(\cdot -q)^2}{\hbar}}  dq -f \right|\right|_{\gB(L^2(\bR^n,dx))}
= \left|\left| \frac{1}{(\pi\hbar)^{n/2}}\int_{\bR^{n}} f(q) e^{-\frac{(\cdot -q)^2}{\hbar}}  dq -f \right|\right|_{\infty}\:.$$
To prove (\ref{uniff}), observe that since
$\int_{\bR^{n}} e^{-\frac{(x-q)^2}{\hbar}}  dq  =(\pi\hbar)^{n/2}$,
(\ref{uniff}) can be re-written
\beq I_\hbar(x) := \frac{1}{(\pi\hbar)^{n/2}}\int_{\bR^{n}} (f(q) - f(x) )e^{-\frac{(x-q)^2}{\hbar}}  dq \to 0 \quad \mbox{uniformly in $x\in \bR^n$ if $\hbar \to 0^+$.}\label{uniff2}\eeq
Let us prove that it is valid.  We start from the decomposition 
$$I_\hbar(x) =  \frac{1}{(\pi\hbar)^{n/2}}\int_{|x-q|< \delta}\sp\sp\sp  (f(q) - f(x) )e^{-\frac{(x-q)^2}{\hbar}}  dq  
+ \frac{1}{(\pi\hbar)^{n/2}}\int_{|x-q|\geq  \delta} \sp\sp\sp (f(q) - f(x) )e^{-\frac{(x-q)^2}{\hbar}}  dq  .$$
The crucial observation is that $f\in C_0(\bR^n)$ is necessarily uniformly continuous and thus, for every $\epsilon>0$, there is $\delta>0$ such that $|f(q) - f(x) | < \epsilon/2$ if  $|x-q|< \delta$. Hence,
$$I_\hbar(x) \leq  \epsilon/2
+ \frac{1}{(\pi\hbar)^{n/2}}\int_{|x-q|\geq  \delta} \sp\sp\sp (f(q) - f(x) )e^{-\frac{(x-q)^2}{\hbar}}  dq  \leq \epsilon/2 + 2||f||_\infty  \frac{\hbar^{n/2}}{(\pi)^{n/2}}\int_{|y|\geq  \delta/\hbar} \sp\sp\sp e^{-y^2}  dy 
\leq \epsilon/2 + C \hbar^{n/2}\:,$$
where $y= (q-x)/\sqrt{\hbar}$.
We can finally choose $H_\epsilon>0$ such that $C \hbar^{n/2} < \epsilon/2$ if $0<\hbar <H_\epsilon$, concluding the proof for $f\in C_0(\bR^{n})$.\\
(b) Let $f\in {\cal S}(\bR^n)$ be a function of the variable $p$. If $\phi \in {\cal S}(\bR^n)$ so that all integrals can be interchanged,  we have 
$$(Q^B_\hbar(f)\phi)(x) = \frac{1}{2^n(\pi\hbar)^{3n/2}}\int_{\bR^{2n}} f(p) e^{-\frac{(x-q)^2}{2\hbar}}\int_{\bR^n} e^{i\frac{p(x-x')}{\hbar}}   
e^{-\frac{(x'-q)^2}{2\hbar}}\phi(x')dx' dqdp$$ 
\beq =  \frac{1}{2^{n/2}(\pi\hbar)^{n}}\int_{\bR^{2n}}  \check{f}(x-x') e^{-\frac{(x-q)^2}{2\hbar}}e^{-\frac{(x'-q)^2}{2\hbar}}\phi(x')dx' dq 
 =  \frac{1}{2^{n/2}(\pi\hbar)^{n}}\int_{\bR^{n}}  \check{f}(x-x') G(x-x')\phi(x')dx' \label{conv},\eeq
where 
$G(x) :=  \int_{\bR^n}e^{-\frac{(x-z)^2}{2\hbar}}e^{-\frac{z^2}{2\hbar}}dz\:.$
Observe that, the convolution thorem in ${\cal S}(\bR^n)$ implies that
$\widehat{G}(p) =(2\pi \hbar)^{n/2} \hat{g}(p)^2$
where $\hat{g}$ is the Fourier transform of $g(z) := e^{-\frac{z^2}{2\hbar}}$
$\hat{g}(p) = e^{-\frac{p^2}{2\hbar}}\:.$
Using again the convolution theorem in (\ref{conv}), we conclude that
$$\widehat{(Q^B_\hbar(f)\phi)}(p) =  \frac{(2\pi \hbar)^{n/2}}{2^{n/2}(\pi\hbar)^{n}} (f * \hat{g}^2)(p) \hat{\phi}(p)
=  \left(\frac{1}{(\pi\hbar)^{n/2}}\int_{\bR^n} e^{-(p-b)^2/\hbar} f(b) db\right) \hat{\phi}(p)\:.$$  As $\hat{\phi} \in {\cal S}(\bR^n)$
 (since that space  is invariant under Fourier transform),
this identity can be extended to every $\hat{\phi} \in L^2(\bR^n, dp)$ by means of an elementary density argument.
Hence $Q^B_\hbar(f)$ acts as a multiplicative operator on the Fourier-Plancherel  transforms $\hat{\phi}$ of the wavefunctions  $\phi\in L^2(\bR^n, dx)$.
In other words, if $F_\hbar$ denotes the Fourier-Plancherel unitary operator,
$$\left((F_\hbar Q^B_\hbar(f) F_\hbar^{-1}) F_\hbar\phi\right)(p)   = \left(\frac{1}{(\pi\hbar)^{n/2}}\int_{\bR^n} e^{-(p-b)^2/\hbar} f(b) db\right)  (F_\hbar\phi)(p)\:.$$
Looking at the right-hand side, since $f \in {\cal S}(\bR^n) \subset C_0(\bR^n)$, exploiting the same argument as in (a) 
$||F_\hbar Q^B_\hbar(f) F_\hbar^{-1}- f|| \to 0 \quad \mbox{for $\hbar \to 0$.}$ That is equivalent to
$||Q^B_\hbar(f) - F_\hbar^{-1} f F_\hbar|| \to 0 \quad \mbox{for $\hbar \to 0$.}$
This is the thesis into an equivalent form.  We finally observe that, if $f(p)= e^{-tp^2}$ for a given  $t>0$, then 
$F_\hbar^{-1} f F_\hbar = e^{t\hbar^2 \Delta}\:,$
as is well known. concluding the proof.
\hfill $\Box$\\

\noindent {\bf Proof of Lemma \ref{fclaim}}.
Let us define $\Gamma := e^{-1}(\Lambda)$.  This set is $G$-invariant, i.e.,    $g(\Gamma) \subset \Gamma$ for every $g\in G$ because $e(g\sigma) = e(\sigma)$ for every $g\in G$ and $\sigma \in \bR^{2n}$. Since the action of $G$ on $\Gamma$ is transitive, we also have that $\Gamma = \{g\sigma_0\:|\: g\in G\}$ for every chosen $\sigma_0 \in \Gamma$.
This identity has the important consequence that $\Gamma$ is compact under our main hypotheses. Indeed,    if $G$ is finite,  then  $\Gamma$ is made of a finite number of points and thus it is compact {\em a fortiori}.
If $G$ is a topological compact group and its action is continuous, as requested in the main hypotheses, then  $\Gamma$ is the image of a compact set under the  continuous function $G \ni g \mapsto
g\sigma_0$  and thus $\Gamma$ is compact as well.\\
 If $\delta>0$ a $\delta$-{\em covering}  of $\Gamma$  is a set of the form
$C_\delta := \bigcup_{\sigma\in \Gamma} B_\delta(\sigma)\:,$
where $B_\delta(\sigma)$
is an open ball of radius $\delta$ centred at $\sigma$.\\
Since $\Gamma$ is compact, there is a closed ball $B$ centred at the origin of finite positive radius such that $\Gamma$ is completely contained in the interior of $B$. All other balls we shall consider in this proof will be assumed to be contained in the interior of $B$ as well. 
Since $|e(\sigma)| \to 0$ for $|\sigma| \to +\infty$ and $\Lambda \neq 0$, we can always fix the radius of $B$ such that $|\Lambda-e(\sigma)|> \eta $, for some $\eta>0$,  if $\sigma \not \in B$.\\
Our next step consists of  proving that, given a $\delta$-covering $C_\delta$ of $\Gamma$, with $\delta>0$
arbitrarily taken,  there exists $m_\delta\in \bN$ such that 
${\cal V}_{1/m_\delta} \subset C_\delta$. Indeed, suppose that it is not the case for some $\delta>0$. As a consequence, for every 
$m \in \bN$, it must be ${\cal V}_{1/m} \not \subset C_\delta$, 
so that there is $\sigma_m \in {\cal V}_{1/m}$ not included in $C_\delta$. 
Choosing $m>1/\eta$ and keeping in mind that $|e(\sigma_m)-\Lambda| <1/m$, we can also exclude that $\sigma_m$ stays outside $B$. In other words, for every $m\in \bN$ sufficiently large, we have a point  $\sigma'_m \in B \setminus C_\delta$ with $|e(\sigma'_m)-\Lambda| <1/m$. Since $B \setminus C_\delta$ is compact, we can extract a subsequence  $\sigma'_{m_k} \to \sigma_0' \in 
B \setminus C_\delta$ for $k\to +\infty$. As $1/m_k \to 0$ as $k\to+\infty$, continuity imposes that $e(\sigma_0') =\Lambda$ and thus $\sigma'_0 \in \Gamma$. This is not possible because $\Gamma \subset C_\delta$ that is disjoint from $B \setminus C_\delta$. We have proved that  every $C_\delta$ covering of $\Gamma$ contains ${\cal V}_{1/m_\delta}$ if $m_\delta \in \bN$ is sufficiently large.\\
Now take  $\sigma_0 \in \Gamma$. Noticing that $B$ is compact and $F$ is continuous thereon, we can use its uniform continuity. 
Given $\epsilon >0$, there is $\delta_\epsilon>0$ such that $|F(\sigma)-F(\sigma')| < \epsilon/2$ if $|\sigma -\sigma'| < \delta_\epsilon$. With this remark, consider a $C_{\delta_\epsilon}$ covering of $\Gamma$.  If $\tau \in C_{\delta_\epsilon}$ we have
$|F(\tau) - F(\sigma_0)|= |F(\tau) - F(\sigma_0^{\tau})|$
where $\sigma_0^\tau \in \Gamma$ is the centre of $B_{\delta_\epsilon}(\sigma_0^\tau)$ which contains $\tau$. The identity above is valid because $F$ is constant in $\Gamma$. Uniform continuity therefore implies that 
$|F(\tau) - F(\sigma_0)|  < \epsilon/2 \quad \mbox{if $\tau \in C_{\delta_\epsilon}$}\:.$
In summary,  given $\epsilon>0$, if $m_\epsilon\in \bN$ is sufficiently large to assure  that ${\cal V}_{1/m_{\epsilon}} \subset C_{\delta_\epsilon}$, we have the thesis
$\sup_{\sigma \in {\cal V}_{1/m_\epsilon}}|F(\sigma) - F(\sigma_0)| \leq \sup_{\sigma \in C_{\delta_\epsilon}}|F(\sigma) - F(\sigma_0)| < \epsilon/2$,
concluding the proof.
\hfill $\Box$\\

\noindent {\bf Proof of Proposition \ref{IMPORTANT}}.
As already observed in the main text, $H_\hbar$ is selfadjoint.  Since $H_\hbar$  is also  bounded below, it is also  the generator of a strongly-continuous one-parameter semigroup   in $L^2(\bR^n, dx)$ as immediately arises from spectral calculus (see, e.g., \cite{Mor}). 

The Weyl and Berezin quantization procedures are equivalent when applied to $e^{-t(p^2+ V)}$ since it is an element of the Schwartz space and 
$|| Q^W_\hbar(f) - Q^B_\hbar(f)|| \to 0 \quad \mbox{for $\hbar \to 0^+$ and a given  $f \in {\cal S}(\bR^n)$}\:.$
Since $e^{-t(p^2+V(x))}$ defines a function of ${\cal S}(\bR^{2n})$ for every given $t>0$,
 the proof consists of establishing the thesis with $Q^B_\hbar$ replaced for $Q^W_\hbar$. 
With this replacement, the thesis immediately arises collecting the theses of  the following pair of lemmata whose hypotheses are fulfilled if $V$ satisfies (V1)-(V3).

\begin{lemma}\label{lemmaA}  Let  $V \in   L^2_{loc}(\bR^n, dx)$ satisfy $\inf_{x\in \bR^n} V(x) = c >-\infty$ for $x\in \bR^n$.
If $e^{-tV} \in {\cal S}(\bR^n)$ for some $t>0$, then
$Q^{W}_\hbar(e^{-t(p^2+ V)})  \to e^{-t V}\quad \mbox{for $\hbar \to 0^+$}$
in strong sense.
\end{lemma}

\noindent {\em Proof of Lemma \ref{lemmaA}}.   
Notice that, with the said hypotheses,  $ e^{-t(p^2+ V(x))}$ defines a function of ${\cal S}(\bR^{2n})$ for every $t>0$.
Using (\ref{formfin}) with $f(x,p) = e^{-t(p^2+ V(x))} = e^{-tV(x)} e^{-tp^2} = f_1(x) f_2(p)\:,$ we find, choosing $\psi \in L^2(\bR^n,dx)$ and $\phi \in {\cal S}(\bR^n)$,
$$\langle \psi, Q_\hbar^W(e^{-t(p^2+ V)})\phi \rangle =
\frac{1}{(2\pi)^{n/2}} \int_{\bR^n} \widehat{f}_2(b) \left\langle \psi, e^{ib\cdot P} e^{-tV\left(x+ \frac{\hbar}{2}bI\right)} \phi
 \right\rangle\: db$$
$$=\int_{\bR^n} \frac{\widehat{f}_2(b)}{(2\pi)^{n/2}}  \left\langle \psi,  e^{-tV\left(x-\frac{\hbar}{2}bI\right)} e^{ib\cdot P}
 \phi \right\rangle\: db = \int_{\bR^n\times \bR^n} \frac{\widehat{f}_2(b)}{(2\pi)^{n/2}}  \overline{\psi(x)}  e^{-tV\left(x-\frac{\hbar}{2}b\right)} \phi\left(x+\hbar b \right) dx db$$
$$=  \int_{\bR^n}  \overline{\psi(x)}  \int_{\bR^n}  \frac{\widehat{f}_2(b)}{(2\pi)^{n/2}}   e^{-tV\left(x-\frac{\hbar}{2}b\right)} \phi\left(x+\hbar b \right) dx db\:, $$
where we have used the Fubini-Tonelli theorem 
since the function in the argument of the $\bR^n\times \bR^n$ integral is absolutely integrable. 
Indeed, $|\phi(x)| \leq C_k (1+ |x|)^{-k}$ for every $k\in \bN$ and some $C_k\geq 0$,   and we can take 
 $0\leq \hbar \leq \delta$ for some $\delta >0$. As a consequence we can find $\phi_0 \in L^2(\bR^n,dx)$ such that 
 $|\phi\left(x+ \hbar b \right)| \leq |\phi_0(x)|$ when $0\leq \hbar \leq \delta$.  Now observe that $\psi \phi_0 \in L^1(\bR^n, dx)$ since $\psi, \phi_0 \in L^2(\bR^n,dx)$, and $e^{-tV\left(x-\frac{\hbar}{2}b\right)} \leq e^{-tc}$. Summing up, since $\widehat{f}_2$ is $L^1(\bR^n, db)$ by construction,
 the argument of the $\bR^n\times \bR^n$ integral is bounded by the  $L^1(\bR^n\times \bR^n, dx\otimes db)$ map
$|e^{-tc}\widehat{f}_2 \psi \phi_0| $
and thus it is absolutely integrable as well. We have so far established that, for $\phi \in {\cal S}(\bR^n)$ and $\psi \in L^2(\bR^n, dx)$,
\beq \langle \psi, Q_\hbar^W(e^{-t(p^2+ V)})\phi \rangle =  \int_{\bR^n}  \overline{\psi(x)}  \int_{\bR^n}  \frac{\widehat{f}_2(b)}{(2\pi)^{n/2}}   e^{-tV\left(x-\frac{\hbar}{2}b\right)} \phi\left(x+\hbar b \right) dx db\:. \label{Qexp1}\eeq
Since $\psi \in L^2(\bR^n, dx)$ is arbitrary, we conclude that 
\beq \left(Q_\hbar^W(e^{-t(p^2+ V)})\phi\right)(x)
=  \int_{\bR^n}  \frac{\widehat{f}_2(b)}{(2\pi)^{n/2}}   e^{-tV\left(x-\frac{\hbar}{2}b\right)} \phi\left(x + \hbar b \right) db \:.\label{Qexp2}\eeq
Explicitly,
$\widehat{f}_2(b) = \frac{e^{-b^2/(4t^2)}}{(2t)^{n/2}}\:.$
At this juncture, with our choice of $\phi$ and $\hbar$, exploiting the same argument used for applying Fubini-Tonelli theorem, 
Lebesgue's dominated convergences implies that, for $\hbar \to 0$,
$$\left(Q_\hbar^W(e^{-t(p^2+ V)})\phi\right)(x)
\to   \int_{\bR^n}  \frac{\widehat{f}_2(b)}{(2\pi)^{n/2}}   e^{-tV\left(x\right)} \phi\left(x\right) db $$
$$ =   \int_{\bR^n} e^{ib\cdot 0} \frac{\widehat{f}_2(b)}{(2\pi)^{n/2}} db\:   e^{-tV\left(x\right)} \phi\left(x\right)   = f_2(0)  \: e^{-tV\left(x\right)} \phi\left(x\right) =  e^{-tV\left(x\right)} \phi\left(x\right) \:.$$
To conclude, consider $\psi \in L^2(\bR^n, dx)$. We have
$$||(Q^{W}_\hbar(e^{-t(p^2+ V)})  - e^{-t V}) \psi||  \leq 
||(Q^{W}_\hbar(e^{-t(p^2+ V)})  -e^{-t V}) (\psi- \phi)||+ ||(Q^{W}_\hbar(e^{-t(p^2+ V)})  - e^{-t V})\phi|| $$
$$\leq (||(Q^{W}_\hbar(e^{-t(p^2+ V)}) || + ||e^{-t V}||)  ||\psi- \phi|| + ||(Q^{W}_\hbar(e^{-t(p^2+ V)})  - e^{-t V})\phi||$$
$$ \leq (C + 1)  ||\psi- \phi|| +  ||(Q^{W}_\hbar(e^{-t(p^2+ V)})  - e^{-t V})\phi|| \:.$$
where $C$ arises form (\ref{estimate}).
In summary, 
$$||(Q^{W}_\hbar(e^{-t(p^2+ V)})  - e^{-t V}) \psi||  \leq (C + 1)  ||\psi- \phi|| +  ||(Q^{W}_\hbar(e^{-t(p^2+ V)})  - e^{-t V})\phi|| \:.$$
The found inequality permits us to prove the thesis of the lemma for the said $\psi$. Indeed, 
for every $\epsilon>0$ we can first take $\phi \in {\cal S}(\bR^n)$ such that $ (C + 1)  ||\psi- \phi|| < \epsilon/2$. With that choice of $\phi$,
we can next take $\delta>0$ such that $0\leq \hbar < \delta $ implies, for the first part of the proof, $ ||(Q^{W}_\hbar(e^{-t(p^2+ V)})  - e^{-t V})\phi|| < \epsilon/2$. Hence
$||(Q^{W}_\hbar(e^{-t(p^2+ V)})  - e^{-t V})\psi|| < \epsilon$ if $0\leq \hbar < \delta $ that is the thesis.
$\hfill \Box$

\begin{lemma}\label{lemmaB} If $V\in  L^2_{loc}(\bR^n, dx)$ satisfies $\inf_{x\in \bR^n}V(x) = c >-\infty$ for $x\in \bR^n$,  $\psi \in L^2(\bR^n,dx)$, and $T>0$, then
$\sup_{t\in[0,T]}||(e^{-t (\overline{-\hbar^2 \Delta + V})}  -e^{-tV})\psi|| \to 0\quad \mbox{for $\hbar \to 0$}\:,$
where $-\hbar^2 \Delta + V$ is defined on $C_c^\infty(\bR^n)$. 
\end{lemma}

\noindent {\em Proof of Lemma \ref{lemmaB}}. Let us first assume $c=0$. If $\lambda >0$,
from 
Theorem X.28 \cite{RS2}, we have that $(\overline{-\lambda \Delta + V})$ is selfadjoint and positive and thus defines the generator of  a  semigroup.
We take advantage of Corollary 3.18 in \cite{Davies}.  Now the relevant  semigroups are
$T^\lambda_t := e^{-t (\overline{-\lambda \Delta + V})}$ and 
$T_t:= e^{-tV}$ where $\lambda = \hbar^2$. Evidently 
$||T_t^\lambda||$, $||T_t|| \leq 1$ since the generators are positive. Next consider $\cD := C_c^\infty(\bR^n)$.
It easy to prove that this dense subspace is a core of $V$ viewed as a selfadjoint multiplicative operator (it is sufficient to observe that $V|_\cD$ is symmetric and   the defect spaces of $V|_{\cD}$ are trivial\footnote{$(V|_\cD)^* \psi = \pm i\psi$ implies in particular that 
$\int_{\bR^n} (V(x) \pm i)g(x)\psi(x) dx = 0$ for every $g \in \cD = C_c^\infty(\bR^n)$ and thus  $(V(x) \pm i)\psi(x) =0$ a.e.. Since  $(V(x) \pm i)\neq 0$, it must be $\psi=0$ a.e..}). With our definitions, if $f \in \cD$, then $f \in D(\overline{-\lambda \Delta + V})$ and
$$\lim_{\lambda \to 0^+} \overline{-\lambda \Delta + V} f =
\lim_{\lambda \to 0^+}( -\lambda \Delta f + V f) = Vf,$$ holds trivially. All that means that the hypotheses of of Corollary 3.18 in \cite{Davies} are true and thus the thesis is valid, that is nothing but the thesis of our lemma when $c=0$.
If $c \neq 0$, let us define ${V}_0 := V -c$. The thesis is true if replacing $V$ for ${V}_0$ using the above proof. The thesis is also valid for $V$ just because
$e^{-t (\overline{-\hbar^2 \Delta + V})}   = e^{-tc}e^{-t (\overline{-\hbar^2 \Delta + {V}_0})}  $
so that 
$$\sup_{t\in[0,T]}||(e^{-t (\overline{-\hbar^2 \Delta + V})}  -e^{-tV})\psi|| = 
\sup_{t\in[0,T]}|e^{-ct}|\:|(e^{-t (\overline{-\hbar^2 \Delta + {V}_0})}  -e^{-t{\cal V}_0})\psi|| \to 0,$$
because $|e^{-ct}| < k <+\infty$ for $t\in [0,T]$ nomatter the sign of $c$.
$\hfill \Box$\\
$\Box$\\

\noindent {\bf Proof of proposition \ref{IMPORTANT2}}. Let us first assume $c=0$.
According  to the discussion before proposition \ref{IMPORTANT2},
it is sufficient to prove (1) since (2) immediately arises form the fact that  the norm of $Q_\hbar^B(e^{-t(p^2 + V)})$ coincides to $\lambda_\hbar^{(0)}$. 
Furthermore, since $||Q_\hbar^B(e^{-t(p^2 + V)})|| \to ||e^{-t(p^2 + V)}||_\infty = 1$ for $\hbar \to 0^+$ (that is nothing but the Rieffel condition which is valid 
for theorem \ref{BerezinC0} since $e^{-t(p^2 + V)} \in C_0(\bR^n)$), to prove the whole proposition it is enough to demonstrate that
$E_\hbar^{(0)} \to 0^+$ as $\hbar \to 0^+$. Let us prove this fact. The spectral decomposition of the operator $H_\hbar \geq 0$ 
with pure point spectrum
implies that
$E_\hbar^{(0)} = \inf_{\psi \in D(H_\hbar)\:, ||\psi||=1} \langle \psi, H_\hbar \psi \rangle\:.
$
Since $C_c^\infty(\bR^{2n})$ is a core for $H_\hbar$, it is not difficult to prove that
$E_\hbar^{(0)} = \inf_{\psi \in C_c^\infty(\bR^{2n})\:, ||\psi||=1} \langle \psi, H_\hbar \psi \rangle\:.
$
Let us prove that 
\beq
 \inf_{\psi \in C_c^\infty(\bR^{2n})\:, ||\psi||=1} \langle \psi, H_\hbar \psi \rangle \to 0 \quad \mbox{if $\hbar \to 0^+$}
\label{liminf},\eeq
to conclude the whole proof.  From the definition of $H_\hbar$ it immediately arises that, for $\psi \in C_c^\infty(\bR^n)$
and $\hbar \geq  \hbar'>0$,
$ \langle \psi, H_\hbar \psi \rangle = -\hbar^2 \langle \psi, \Delta \psi \rangle  + \langle \psi, V \psi \rangle \geq
 -\hbar'^2 \langle \psi, \Delta \psi \rangle  + \langle \psi, V \psi \rangle  =   \langle \psi, H_{\hbar'} \psi \rangle \geq 0\:.$
As a consequence 
$$\exists \lim_{\hbar \to 0^+} \inf_{\psi \in C_c^\infty(\bR^{2n})\:, ||\psi||=1} \langle \psi, H_\hbar \psi \rangle 
=\inf_{\hbar>0} \inf_{\psi \in C_c^\infty(\bR^{2n})\:, ||\psi||=1} \langle \psi, H_\hbar \psi \rangle \geq 0\:.  $$
Therefore (\ref{liminf}) is true if, for every $\epsilon >0$, there is $\psi_\epsilon \in  C_c^\infty(\bR^{2n})$ with $||\psi_\epsilon||=1$ and such that 
\begin{align}
\langle \psi_\epsilon, H_\hbar \psi_\epsilon \rangle =  -\hbar^2 \langle \psi_\epsilon, \Delta  \psi_\epsilon \rangle  + \langle  \psi_\epsilon, V \psi_\epsilon\rangle < \epsilon, \label{conditionif}
\end{align}
if $\hbar>0$ is sufficiently small.  It is not difficult to find this $\psi_\epsilon$. Let $x_0 \in \bR^n$ be such that $V(x_0)=0 $ according to (V2), and, since $V$ is continuous,  
 let $B\ni x_0$ be an open ball of finite radius
such that $|V(x)| < \epsilon/2$  if $x\in B$. Finally define $\psi_\epsilon \in  C_c^\infty(\bR^{2n})$ with $||\psi_\epsilon||=1$ and $supp(\psi_\epsilon) \subset B$.
By construction,
$0 \leq  -\hbar^2 \langle \psi_\epsilon, \Delta  \psi_\epsilon \rangle  + \langle  \psi_\epsilon, V \psi_\epsilon\rangle <  -\hbar^2 \langle \psi_\epsilon, \Delta  \psi_\epsilon \rangle
+ \epsilon/2 \:.$
To conclude it is sufficient to choose $\hbar >0$ such that, for the said $\psi_\epsilon$, we have
$ -\hbar^2 \langle \psi_\epsilon, \Delta  \psi_\epsilon \rangle < \epsilon/2\:.$
Therefore, \eqref{conditionif} holds and we have established (\ref{liminf}). This proves the case $c=0$.
The case $c\neq 0$ immediately arises from the case $c=0$ when defining ${V}_0 := V -c$, by notincing that the thesis is therefore valid for $V$ replaced by ${V}_0$ and observing that 
$e^{-tH_\hbar} = e^{-tc} e^{-t(\overline{-\hbar \Delta +{V}_0})}$ and $Q_\hbar^B(e^{-t(p^2+V)})= e^{-ct}  Q_\hbar^B(e^{-t(p^2+{V}_0)})$.
\hfill $\Box$\\

\noindent {\bf Proof of Proposition \ref{propDIM}}. 
Concerning $e^{-tH_\hbar}$ the thesis is just the statement of Theorem XIII.47 in \cite{RS4} since (V1)-(V3)   imply the validity of the hypotheses of that theorem.  Regarding $Q_\hbar^B(e^{-t(p^2+V)})$, the thesis arises from Theorem XIII.43 of \cite{RS4} if we prove that $Q_\hbar^B(e^{-t(p^2+V)})\in L^2(\bR^n, dx)$ is (a) positive preserving and  (b) ergodic.
To this end we take advantage of (\ref{Delta2n}).
In particular we have that, if henceforth
$e(x,p) := e^{-t(p^2+V)}$, we then have that
$${\cal E}(x,p)  :=\left(e^{-\hbar \Delta_{2n}/4} e\right)(x,p)  = 
{\cal E}_1(x) {\cal E}_2(p)\:, \quad \mbox{where}$$
$$
{\cal E}_1(x) := \int_{\bR^n}  \frac{e^{-(x-y)^2/\hbar} e^{-tV(y)}}{(\pi \hbar)^{n/2}} dy\:, \quad
{\cal E}_2(p) := \int_{\bR^n}  \frac{e^{-(p-r)^2/\hbar} e^{-tr^2} }{(\pi \hbar)^{n/2}}dr\:.$$
Evidently ${\cal E}_1(x), {\cal E}_2(p)>0$ for every $x,p \in \bR^{n}$ since the integrands of the integrals above are continuous and strictly positive.  Furthermore 
${\cal E} \in {\cal S}(\bR^{2n})$ since it is the product of functions in $ {\cal S}(\bR^{n})$ in the variable $x$ and $p$ respectively.
We can exploit (\ref{Qexp2}) obtaining
$$ \left(Q_\hbar^B(e^{-t(p^2+ V)})\phi\right)(x)
=  \int_{\bR^n}  \frac{\widehat{{\cal E}}_2(b)}{(2\pi)^{n/2}}   {\cal E}_1\left(x-\frac{\hbar}{2}b\right) \phi\left(x + \hbar b \right) db\:.$$
Observe that $\widehat{{\cal E}}_2(b) >0$ for every $b \in \bR^n$ since, up to positive factors,  it is the product of two Gaussian functions, because 
 ${\cal E}_2$ is the convolution of such pair of functions by definition.  We conclude that, if $\phi(x) \geq 0$ a.e.
with $\phi\in L^2(\bR^n,dx)$, then 
$ \left(Q_\hbar^B(e^{-t(p^2+ V(x))})\phi\right)(x) \geq 0$ a.e..  In summary, $Q_\hbar^B(e^{-t(p^2+ V)})$ is positive preserving provided  that  $||\phi|| \neq 0$ for $\phi(x) \geq 0$ a.e., implies $Q_\hbar^B(e^{-t(p^2+ V(x))})\phi \neq 0$.
We prove this property simultaneously to the  ergodicity property just by establishing that, for $\psi(x), \phi(x) \geq 0$ a.e. such that $ \psi, \phi \in L^2(\bR^n, dx)\setminus\{ 0\}$,  then $\langle \psi, Q_\hbar^B(e^{-t(p^2+ V(x))})\phi \rangle \neq 0$.
Taking advantage of (\ref{Qexp1}), we find for $\phi, \psi \in L^2(\bR^n,dx)$ with $\psi(x), \phi(x) \geq 0$ a.e.,  and $||\psi||, ||\phi|| \neq 0$,
$$ \langle \psi, Q_\hbar^B(e^{-t(p^2+ V)})\phi \rangle =\int_{\bR^{2n}} \frac{ \widehat{\cal E}_2(b) }{(2\pi)^{n/2}}  {\cal E}_1\left(x-\frac{\hbar}{2}b\right)  \psi(x) \phi\left(x+\hbar b \right) dx db\:. $$
If, in addition to $\psi(x), \phi(x) \geq 0$ a.e.,  it holds $||\psi||,||\psi|| \neq 0$,  then  there must exist $k_1,k_2 >0$ such that, defining  the measurable sets
$E_1:= (\phi^2)^{-1}(k_1,+\infty)$ and  $E_2:= (\psi^2)^{-1}(k_2,+\infty)$, it holds $m(E_1), m(E_2) >0 $, where $m$  is the Lebesgue measure on $\bR^n$.  Let us change coordinates on $\bR^n \times \bR^n$ from $x,b$ to $x, z:= x+\hbar b$. With this transformation, 
$$ \langle \psi, Q_\hbar^B(e^{-t(p^2+ V)})\phi \rangle  =\frac{1}{(2\pi \hbar^2)^{n/2}} \int_{\bR^{2n}}   \widehat{\cal E}_2\left(\frac{z-x}{\hbar}\right)   {\cal E}_1\left(\frac{3x-z}{2}\right)  \psi(x) \phi\left(z \right) dx dz$$
$$\geq \frac{\sqrt{k_1k_2}}{(2\pi \hbar^2)^{n/2}} \int_{E_1\times E_2}    \widehat{\cal E}_2\left(\frac{z-x}{\hbar}\right)    {\cal E}_1\left(\frac{3x-z}{2}\right) dx dz>0 $$
where the last integral is well defined since, up to a non-singular linear change of coordinates,  the integrand is the product of  two ${\cal S}(\bR^n)$ functions which is therefore a 
 ${\cal S}(\bR^{2n})$. The  integral is eventually strictly positive because the  integrand is everywhere strictly positive and $E_1\times E_2$ has strictly positive measure\footnote{If $m(F)>0$ and $g: F \to \bR$ is measurable and strictly positive, then $F = \cup_{n\in \bN} g^{-1}(1/n,+\infty)$. It must be $m(g^{-1}(1/n_0,+\infty))= m_0 >0$ for some $n_0$ otherwise $m(F)=0$ for internal regularity. Hence $\int_E g dx \geq  \int_{g^{-1}(1/n_0,+\infty)} g dx \geq  m_0/n_0>0$.}: $m_{\bR^{2n}}(E_1\times E_2)=
m(E_1)m(E_2) >0$.
\hfill $\Box$\\

\noindent {\bf Proof of Lemma \ref{lemmaSCH1}}.
Take $\varepsilon>0$. Given $\phi\in C^\infty_c(\mathbb{R}^{2n})$, and using that $H_\hbar$ and in particular $V$ are positive, for any $\hbar>0$ we observe
\begin{align*}
\langle\phi,(H_\hbar + \varepsilon I)\phi\rangle=-\hbar^2\langle\phi,\Delta\phi\rangle+\langle\phi,(V+ \varepsilon I)\phi\rangle\geq  \langle\phi,(V+ \varepsilon I)\phi\rangle\geq 0.
\end{align*}
Since $\phi \in D(A)$ implies $\phi \in D(\sqrt{A})$ for a selfadjoint operator $A\geq 0$, we conclude that, if $\phi\in C_c^\infty(\bR^b)$
then
\beq
|| \sqrt{H_\hbar +\varepsilon I} \phi ||^2 \geq ||\sqrt{V+ \varepsilon}\phi||^2\:. \label{phin}
\eeq
Since $C_c^\infty(\bR^n)$ is a core for $H_\hbar$,  it is also a core for $H_\hbar + \varepsilon I$. In particular, if $\psi \in D(H)$, there is a sequence $C_c^\infty(\bR^n) \ni \phi_n \to \psi$ such that $(H_\hbar + \varepsilon I) \phi_n \to (H_\hbar + \varepsilon I)  \psi$. Applying the bounded operator 
$(H_\hbar + \varepsilon I)^{-1/2}$ on both sides we also have 
\beq
\sqrt{H_\hbar + \varepsilon I} \phi_n \to \sqrt{H_\hbar + \varepsilon I} \psi\:. 
\eeq
Since 
$|| \sqrt{H_\hbar +\varepsilon I} (\phi_n - \phi_m) ||^2 \geq ||\sqrt{V+ \varepsilon}(\phi_n - \phi_m)||^2$,
we conclude that the sequence of the $\sqrt{V+ \varepsilon}\phi_n$ is Cauchy and thus it must converge in the Hilbert space when $\phi_n \to \psi$. As  
$\sqrt{V+ \varepsilon}$  is closed (it being selfadjoint), we have that 
$\sqrt{V+ \varepsilon}\phi_n \to \sqrt{V+ \varepsilon}\psi$ so that
$D(H_\hbar + \varepsilon I) \subset D(\sqrt{V+ \varepsilon})$.
From (\ref{phin}),
\beq
|| \sqrt{H_\hbar +\varepsilon I} \psi ||^2 \geq ||\sqrt{V+ \varepsilon}\psi||^2.
\eeq
Specializing the result to the normalized ground state eigenvector $\psi_\hbar^{(0)}$ of $H_\hbar$, we have
$$  \lambda_\hbar^{(0)} + \varepsilon= || \sqrt{H_\hbar +\varepsilon I} \psi_\hbar^{(0)} ||^2  \geq ||\sqrt{V+ \varepsilon}\psi_\hbar^{(0)}||^2   = \int_{[0 +\infty)} s^2  d\mu^\varepsilon_{\psi_{\hbar}^{(0)}}(s)\:,$$ 
where $\mu^\varepsilon_{\psi_{\hbar}^{(0)}}$ is the spectral probability  measure $\mu^\varepsilon_{\psi_{\hbar}^{(0)}}(F) := \langle \psi_{\hbar}^{(0)}, P_F^{\sqrt{(V+ \varepsilon I)}}\psi_\hbar^{(0)} \rangle$ defined on $[0, +\infty)$. 
The function $e^{-tx}$ is convex on $[0, +\infty)$.
We can therefore apply {\em Jensen's inequality} for probability measures, obtaining 
$$e^{-t(\lambda_\hbar^{(0)} +\varepsilon)} \leq e^{-t \int_{[0, +\infty)} s^2 d\mu^\varepsilon_{\psi_{\hbar}^{(0)}}(s)} \leq   \int_{[0, +\infty)} e^{-ts^2} d\mu^\varepsilon_{\psi_{\hbar}^{(0)}}(s) = \langle \psi_\hbar^{(0)}, e^{-t (V+\varepsilon I)}\psi_\hbar^{(0)}\rangle.$$
As a consequence 
$e^{-t(\lambda_\hbar^{(0)} +\varepsilon)} \leq  e^{-t\varepsilon} \langle \psi_\hbar^{(0)}, e^{-t V}\psi_\hbar^{(0)}\rangle \leq 1\:,$
for every $\varepsilon>0$, so that
$e^{-t\lambda_\hbar^{(0)}} \leq  \langle \psi_\hbar^{(0)}, e^{-t V}\psi_\hbar^{(0)}\rangle \leq 1$.
Using the fact that $\lim_{\hbar\to 0}\lambda_\hbar^{(0)}=0$, we obtain the former identity in (\ref{convJensen}) for every given $t>0$,
$
 \lim_{\hbar\to 0}\langle\psi_\hbar^{(0)},e^{-tV}\psi_\hbar^{(0)}\rangle=1\:. $
With a strictly analogous procedure we also find the latter identity in (\ref{convJensen})  for every given $t>0$.
$\hfill \Box$\\

\noindent {\bf Proof of Lemma \ref{fclaim2}}. The proof is very similar to the one of  lemma \ref{fclaim}.
If $e(q,p):= e^{-t(p^2+ V)}$ for some $t>0$, let us define $\Gamma := e^{-1}(\{\max e\}) = h^{-1}(\{\min h\})$.  This set satisfies 
$\Gamma = \{g\sigma_0\:|\: g\in G\}$ for every chosen $\sigma_0 \in \Gamma$ due to $G$-invariance of $h$ and transitivity of $G$
on $\Gamma$ (hypothesis (b)), and is compact as said in the proof of lemma \ref{fclaim}.
 If $\delta>0$ a $\delta$-{\em covering}  of $\Gamma$ is defined as in the proof of lemma \ref{fclaim}.
Since $\Gamma$ is compact, there is a closed ball $B$ centred at the origin of finite positive radius such that $\Gamma$ is completely contained in the interior of $B$. All other balls we shall consider in this proof will be assumed to be contained in the interior of $B$ as well. 
Since $|e(\sigma)| \to 0$ for $|\sigma| \to +\infty$ and $\max e \neq 0$, we can always fix the radius of $B$ such that $|\max e-e(\sigma)|> \eta $, for some $\eta>0$,  if $\sigma \not \in B$.
The  next step consists of  proving that, given a $\delta$-covering $C_\delta$ of $\Gamma$, with $\delta>0$
arbitrarily taken,  there exists $m_\delta\in \bN$ such that 
${\cal U}_{1/m_\delta} \subset C_\delta$, where ${\cal U}_\delta$ is defined as in (\ref{NBH}), so that (\ref{Ct}) is in particular valid.
The proof is the same as in the proof of  lemma \ref{fclaim} with $\Lambda$ replaced for $\max e$ and  using  (\ref{Ct}) where necessary.
Now take  $\sigma_0 \in \Gamma$. Noticing that $B$ is compact and $F$ is continuous thereon, we can use its uniform continuity. 
Given $\epsilon >0$, there is $\delta_\epsilon>0$ such that $|F(\sigma)-F(\sigma')| < \epsilon/2$ if $|\sigma -\sigma'| < \delta_\epsilon$. With this remark, consider a $C_{\delta_\epsilon}$ covering of $\Gamma$.  If $\tau \in C_{\delta_\epsilon}$ we have
$|F(\tau) - F(\sigma_0)|= |F(\tau) - F(\sigma_0^{\tau})|$
where $\sigma_0^\tau \in \Gamma$ is the centre of $B_{\delta_\epsilon}(\sigma_0^\tau)$ which contains $\tau$. The identity above is valid because $F$ is constant in $\Gamma$. Uniform continuity therefore implies that 
$|F(\tau) - F(\sigma_0)|  < \epsilon/2 \quad \mbox{if $\tau \in C_{\delta_\epsilon}$}\:.$
In summary,  given $\epsilon>0$, if $m_\epsilon\in \bN$ is sufficiently large to assure  that ${\cal U}_{1/m_{\epsilon}} \subset C_{\delta_\epsilon}$, we have the thesis
$\sup_{\sigma \in {\cal U}_{1/m_\epsilon}}|F(\sigma) - F(\sigma_0)| \leq \sup_{\sigma \in C_{\delta_\epsilon}}|F(\sigma) - F(\sigma_0)| < \epsilon/2$,
concluding the proof.
\hfill $\Box$\\

\noindent {\bf Proof of Proposition  \ref{compact}}. The group composition law permits us to prove the thesis just for $u=0$ without loss of generality.  A compact operator $A$ is the norm-operator limit of  $A_N = \sum_{k=1}^N \langle \psi_k \:, \cdot \rangle \phi_k$ for suitable $\psi_k$ and $\phi_k$.  
Writing $A= A_N + R_N$, where $||R_N|| \to 0$ as $N\to +\infty$,  the following chain of inequalities hold
$||U_t A U_{-t} -A|| \leq ||U_{-t} A_N U_{t} -A_N|| + ||U_{-t}  R_N U_{t} -R_N|| \leq ||U_{-t}  A_N U_{t} -A_N|| + ||U_{-t}  R_N U_{t}|| + ||R_N|| 
= ||U_{-t}  A_N U_{t} -A_N|| + 2||R_N|| $. As
 $||R_N|| \to 0$ as $N\to +\infty$ independently of $t$,
to conclude it is sufficient to prove that, for finite $N$,  $||U_{-t} A_N U_{t} -A_N|| \to 0$ for $t\to 0$.
In turn,  exploiting the triangular inequality $N$ times, the thesis is valid provided it holds for $N=1$, i.e., for $A_{\psi,\phi}=  \langle \psi \:, \cdot \rangle \phi$.
Per direct inspection $|| \langle \psi \:, \cdot \rangle \phi|| \leq ||\psi||\:||\phi||$ so that, defining $\psi_t := U_t\psi$ and $\phi_t:= U_t\phi$, we have 
$||U_{-t}  A_{\psi, \phi} U_{t} -A_{\psi,\phi}|| = ||\langle \psi_{-t} \:, \cdot \rangle \phi_{-t} - 
\langle \psi \:, \cdot \rangle \phi ||  \leq ||\langle \psi_{-t} \:, \cdot \rangle( \phi_{-t} - \phi)|| + 
||\langle \psi_{-t}-\psi \:, \cdot \rangle \phi || \leq ||\psi|| \:||\phi_{-t}-\phi|| + ||\psi-\psi_{-t}||\: ||\phi|| \to 0$ for $t\to 0$.
\hfill $\Box$.\\

\noindent {\bf Proof of Proposition  \ref{propNh}}.  If the regular Borel probability measure $\mu$ on $X$ has support in $N_h$, then it is invariant under the flow of $h$ and in its GNS representation the action of that flow is trivial.  Therefore  its selfadjoint generator is the zero operator which fulfills the thesis. Suppose {\em vice versa} that a  regular Borel probability measure $\mu$ on $X$  defines a state invariant under the Hamiltonian flow, then $\mu$ itself must be $\phi^{(h)}$-invariant.  Passing to the GNS representation, the condition of $\phi^{(h)}$ invariance  implies
$\langle \Psi_\mu, \{h, \pi_\omega(a) \} \Psi_\mu\rangle = 0$ for $a \in C^\infty_c(X)$, and the condition of 
 positive self-adjoint generator implies 
$-i\langle \Psi_\mu, \pi_\mu(a) \{h, \pi_\omega(a) \} \Psi_\mu\rangle \geq 0$ for $a \in C^\infty_c(X)$. In other words,
$\int_{X} \{h, a \} d\mu =0$ and $-i\int_{X} \overline{a}\{h, a\} d\mu  \geq 0$ must be valid for $a \in C^\infty_c(X)$.
Decomposing $a= f+ig$ with $f$ and $g$ real valued and using the former condition, the latter yields 
$\int_{X} f\{h, g \} d\mu  - \int_{X} g\{h, f \} d\mu \geq 0$ for all real valued $f,g \in C_c^\infty(X)$. Replacing $f$ with $-f$, we conclude that actually the identity holds $\int_{X} f\{h, g \} d\mu  - \int_{X} g\{h, f \} d\mu = 0$. Noticing that $g\{h, f \} = \{h, gf \}- f\{h,g \}$ and using again $\int_{X} \{h, a \} d\mu =0$, we conclude that $\int_{X} f\{h, g \} d\mu =0$ must be valid for every real valued $f,g \in C_c^\infty(X)$. Taking, e.g., $g\in C_c^\infty(X)$ such that
$g(q,p)=q^1$ in an open set including the support of $f$, we have that  $\int_{X} f \frac{\partial h}{\partial q^1} d\mu =0$. In general 
$\int_{X} f  (dh)_k d\mu =0$ for every $k=1,\ldots, 2n$ and $f\in C_c^\infty(X)$ real valued.  If, for a given $k$, $(dh)_k(\sigma_0) >0$ (the case $<0$ is analogous), then $(dh)_k(\sigma_0) >c>0$ in an open neighbourhood $O \ni \sigma_0$. Taking $f \in C_c^\infty(X)$ supported in $O$ with $f\geq 0$, we have $0=\int_{X} f  (dh)_k d\mu > c \int_O f d\mu \geq 0$, so that $ \int_O f d\mu = 0$. 
Since the functions of $C_c^\infty(O)$ are uniformly dense in $C_c(O)$, the result extends to $f \in C_c(O)$.
Arbitrariness of $f$ and the uniqueness part of Riesz' theorem for positive measures  implies that $\mu(O)=0$. In summary, the points $\sigma$ with $(dh)_k(\sigma) \neq 0$ stay outside the support of $\mu$. That is the thesis.
\hfill $\Box$.


\begin{thebibliography}{999}
\bibitem{Ber} F.A. Berezin, General concept of quantization, {\em Communications in Mathematical Physics} {\bf 40}, 153--174, (1975).
\bibitem{BMS94} M. Bordemann, E. Meinrenken, and M. Schlichenmaier, Toeplitz quantization of K\"{a}hler manifolds and $gl(N)$, $N\to\infty$ limits,  {\em  Communications in Mathematical Physics}, {\bf 165}, 281--296 (1994).  
\bibitem{BR1} O. Bratteli and  D.W. Robinson, {\em Operator Algebras and Quantum Statistical Mechanics. Vol.\ I:
Equilibrium States, Models in Statistical Mechanics}, Springer (1981). 
\bibitem{BR2} O. Bratteli and  D.W. Robinson, {\em Operator Algebras and Quantum Statistical Mechanics. Vol.\ II:
Equilibrium States, Models in Statistical Mechanics}, Springer (1981). 
\bibitem{BuchGrun} D. Buchholz and H. Grundling. The resolvent algebra: a new approach to canonical
quantum systems. {\em Journal of Fucntional Analyis}, {\bf 254}, 2725--2779, (2008).
\bibitem{Charles} L. Charles,  Berezin Toeplitz Operators, a semi-classical approach. {\em Communications in Mathematical Physics} {\bf 239}, 1--28 (2003).
\bibitem{Com} M. Combescure, D. Robert {\em Coherent States and Applications in Mathematical
Physics}, Springer, (2012).
\bibitem{Con} A. Connes, {\em Noncommutative Geometry}, Academic Press, San Diego (1994).
\bibitem{Davies} E. B.  Davies, {\em One-Parameter Semigroups}, Academic Press (1980).
\bibitem{Porte} A. Deleporte,  Toeplitz operators with analytic symbols. {\em The Journal of Geometric Analysis} {\bf 31}, 3915--3967, (2021).
\bibitem{DimSjo} M. Dimassi, J. Sj\"{o}strand, {\em Spectral Asymptotics in the Semi-Classical Limit} Cambridge University Press (2010).
\bibitem{GriSjo} A. Grigis, J. Sj\"{o}strand, {\em Microlocal Analysis for Differential Operators, an Introduction}, London Mathematical Society Lecture Note Series, (2013).
\bibitem{Hel} B. Helffer, {\em Semi-classical Analysis for the Schr\"{o}dinger Operator and Applications.} Heidelberg: Springer (1988).
\bibitem{SJHE} B. Helffer, J. Sj\"{o}strand. Multiple Wells in the Semi-Classical Limit 1,{\em Communications in Partial Differential Equations}, {\bf 9}, 337--408 (1984).
\bibitem{Jona}  G. Jona-Lasinio, F. Martinelli and E. Scoppola,  New approach to the semiclassical limit of
807 quantum mechanics, {\em Communications in Mathematical Physics} {\bf 80}, 223 (1981),
\bibitem{Lan98} N.P. Landsman, {\em Mathematical Topics Between Classical and Quantum Theory}, Springer, (1998).
\bibitem{Lan17} K. Landsman, {\em Foundations of Quantum Theory: From Classical Concepts to Operator Algebras}, Springer (2017). Open Access at  \verb#http://www.springer.com/gp/book/9783319517766#.
\bibitem{LMV} K. Landsman, V. Moretti, C.J.F. van de Ven,  Strict Deformation Quantization of the state space of $M_k(\bC)$ with applications to the Curie-Weiss model.
{\em Reviews in Mathematical Physics} {\bf 32},  2050031 (2020).
\bibitem{Lieb} E.H. Lieb, The classical limit of quantum spin systems. {\em Communications in Mathematical
Physics} {\bf 62}, 327--340,  (1973).
\bibitem{Mor} V. Moretti, {\em Spectral Theory and Quantum Mechanics}, 2nd Edition, Springer (2018).
\bibitem{MV} V. Moretti, C.J.F. van de Ven,  Bulk-boundary asymptotic equivalence of two strict deformation quantizations,
{\em Letters in Mathematical Physics}  {\bf 110},  2941--2963 (2020) -.
\bibitem{MV2} S. Murro, C.J.F. van de Ven, Injective tensor products in strict deformation quantization. arXiv:2010.03469.
\bibitem{RS1} M. Reed and B. Simon, {\em Methods of Modern Mathematical Physics} Vol I, Academic Press (1975).
\bibitem{RS2} M. Reed and B. Simon, {\em Methods of Modern Mathematical Physics} Vol II, Academic Press (1975).
\bibitem{RS4} M. Reed and B. Simon, {\em Methods of Modern Mathematical Physics} Vol IV, Academic Press (1975).
\bibitem{Rie89} M.A. Rieffel, Deformation quantization of Heisenberg manifolds, {\em  Communications in Mathematical Physics} {\bf 121}, 531--562 (1989).
\bibitem{Rudin} W. Rudin, {\em Real and complex analysis},  McGraw-Hill (1986).
\bibitem{Rie94} M.A. Rieffel, Quantization and $C^*$-algebras, {\em Contemporary Mathematics} {\bf 167}, 67--97 (1994). 
\bibitem{SCHL} M. Schlichenmaier,  Berezin-Toeplitz quantization for compact Kahler  manifolfds. A review of results.{\em Advances in Mathematical Physics} (2010).
\bibitem{Sim84} B. Simon, Semiclassical Analysis of Low Lying Eigenvalues, II. Tunneling. {\em Annals of Mathematics} {\bf 120} 89--118 (1984).
 \bibitem{Sim85} B. Simon, Semiclassical analysis of low lying eigenvalues. IV. The flea on the elephant, {\em Journal of Functional Analysis} {\bf 63}, 123 (1985).
\bibitem{VGRL18}  C. J. F. van de Ven, G. C. Groenenboom, R. Reuvers, N. P. Landsman,  Quantum spin systems versus Schr\"{o}dinger operators: A case study in spontaneous symmetry breaking  \emph{SciPost Phys.} {\bf 8}, 022 (2020) .
\bibitem{Ven2020} C.J.F. van de Ven,  The classical limit of mean-field quantum theories, {\em Journal of Mathematical Physics}  {\bf 61}, 121901 (2020).
\bibitem{Ven2021} C.J.F. van de Ven,  The classical limit and spontaneous symmetry breaking in algebraic quantum theory. arXiv:2109.05653.
\bibitem{MZ} M. Zworski {\em Semiclassical Analysis}, Graduate Studies in Mathematics 138, American Mathematical Society, (2012).
\end{thebibliography}
\end{document}